\documentclass[prx,twocolumn,superscriptaddress,aps,10pt, frontmatterverbose, reprint,footinbib,longbibliography=true]{revtex4-2}

\usepackage{graphicx}
\usepackage{dcolumn}
\usepackage{bm}
\usepackage[colorlinks,
citecolor=teal,
linkcolor=teal,
urlcolor=teal]{hyperref}
\usepackage{mathrsfs}
\usepackage{amsfonts,amssymb,amsmath,mathtools}
\usepackage{caption}
\usepackage{subcaption}
\usepackage{ragged2e}
\usepackage{xcolor}
\usepackage{titlesec}
\usepackage{lineno}
\usepackage{tikz} 
\usetikzlibrary{arrows.meta} 
\titlespacing*{\subsubsection}{0pt}{1\baselineskip}{1\baselineskip}

\makeatletter
\long\def\@makecaption#1#2{%
  \vskip\abovecaptionskip
  \small
  \justifying
  #1.~#2\par
  \vskip\belowcaptionskip}
\makeatother

\usepackage{dsfont}
\usepackage{bbm}
\usepackage{physics}
\usepackage{comment}
\usepackage{orcidlink}
\usepackage{amsthm}
\usepackage[T1]{fontenc}
\newtheorem{theorem}{Theorem}

\newtheorem{lemma}{Lemma}
\newtheorem{proposition}{Proposition}

\newtheorem{remark}{Remark}

\newcommand{\ignore}[1]{}

\newcommand{\rev}[1]{\textcolor{red}{#1}}
\newcommand{\trans}{\mathsf{T}}

\usepackage{pifont}

\newcommand{\NL}{\mathrm{FNL}}
\newcommand{\RENYI}{\alpha}
\newcommand{\MNL}{\mathrm{M}_2^{\NL}}

\providecommand{\DIFaddtex}[1]{{\protect\color{red}#1}} 
\providecommand{\DIFdeltex}[1]{} 
\providecommand{\DIFaddbegin}{\protect\color{red}} 
\providecommand{\DIFaddend}{\protect\color{black}} 
\providecommand{\DIFdelbegin}{} 
\providecommand{\DIFdelend}{} 
\providecommand{\DIFdelFL}[1]{} 
\providecommand{\DIFadd}[1]{\texorpdfstring{\DIFaddtex{#1}}{#1}} 
\providecommand{\DIFdel}[1]{\texorpdfstring{\DIFdeltex{#1}}{}} 
\RequirePackage{listings} 
\lstdefinelanguage{DIFcode}{ 
  moredelim=[il][\color{white}\tiny]{\%DIF\ <\ }, 
  moredelim=[il][\sffamily\bfseries]{\%DIF\ >\ } 
} 
\lstdefinestyle{DIFverbatimstyle}{ 
	language=DIFcode, 
	basicstyle=\ttfamily, 
	columns=fullflexible, 
	keepspaces=true 
} 
\lstnewenvironment{DIFverbatim}{\lstset{style=DIFverbatimstyle}}{} 
\lstnewenvironment{DIFverbatim*}{\lstset{style=DIFverbatimstyle,showspaces=true}}{} 
\begin{document}

\title{Non-Local Magic Resources for Fermionic Gaussian States}

\author{Daniele Iannotti\orcidlink{0009-0009-0738-5998}}
\email{d.iannotti@ssmeridionale.it}
\thanks{These two authors contributed equally}
\affiliation{Scuola Superiore Meridionale, Largo S. Marcellino 10, 80138 Napoli, Italy}
\affiliation{INFN Sezione di Napoli, via Cintia, 80126 Napoli, Italy}
\affiliation{ Institut für Theoretische Physik, Universität zu Köln, Zülpicher Strasse 77, 50937 Köln, Germany}

\author{Beatrice Magni\orcidlink{0009-0009-8577-0525}}
\email{bmagni@uni-koeln.de}
\thanks{These two authors contributed equally}
\affiliation{ Institut für Theoretische Physik, Universität zu Köln, Zülpicher Strasse 77, 50937 Köln, Germany}

\author{Riccardo Cioli\orcidlink{0009-0001-1286-3773}}
\email{riccardo.cioli3@unibo.it}
\affiliation{Dipartimento di Fisica e Astronomia, Università di Bologna and INFN, Sezione di Bologna, via Irnerio 46, I-40126 Bologna, Italy}

\author{Alioscia Hamma\orcidlink{0000-0003-0662-719X}}
\email{alioscia.hamma@unina.it }
\affiliation{Scuola Superiore Meridionale, Largo S. Marcellino 10, 80138 Napoli, Italy}
\affiliation{Dipartimento di Fisica `Ettore Pancini', Universit\`a degli Studi di Napoli Federico II, Via Cintia, 80126 Napoli, Italy}
\affiliation{INFN Sezione di Napoli, via Cintia, 80126 Napoli, Italy}

\author{Xhek Turkeshi\orcidlink{0000-0003-1093-3771}}
\email{xturkesh@uni-koeln.de}
\affiliation{ Institut für Theoretische Physik, Universität zu Köln, Zülpicher Strasse 77, 50937 Köln, Germany}

\begin{abstract}
Entanglement and magic are fundamental resources that capture the complexity of quantum many-body systems. Non-local magic isolates the irreducible nonstabilizerness intrinsically tied to entanglement. However, evaluating this quantity generally requires a prohibitive minimization over the full Hilbert space, making it computationally inaccessible beyond a few qubits. Here, we overcome this bottleneck by establishing a closed-form expression for the non-local stabilizer entropies of fermionic Gaussian states over local Gaussian unitaries, which we prove at R\'enyi index $\RENYI=2$ for arbitrary subsystem size, and which can be evaluated in polynomial time directly from the eigenvalues of the reduced Majorana covariance matrix. We apply this framework to characterize fermionic non-local magic across diverse physical regimes: we derive an exact Page-like curve for typical random states, reveal logarithmic scaling at the quantum critical point of the XY model, and establish a quasiparticle picture for magic generation during out-of-equilibrium quantum quenches. Crucially, because our result relies solely on two-point correlation functions, it provides a scalable route for the experimental estimation of fermionic non-local magic in large-scale quantum processors via fermionic shadow tomography.
\end{abstract}

\maketitle

\section{Introduction}
Entanglement and magic, or nonstabilizerness, are two fundamental notions of quantum complexity that capture complementary aspects of what makes a many-body system genuinely quantum~\cite{Horodecki2009,Chitambar2019resources,Veitch2014resource}.
Their operational significance is equally distinct: states with low entanglement admit efficient classical simulation via tensor-network methods~\cite{Schollw_ck_2011}, whereas states with low magic are accessible to stabilizer-based techniques such as tableau simulation~\cite{Gottesman1998,Bravyi2019simulationofquantum} or Pauli propagation~\cite{dowling2025magic,rudolph2025paulipropagationcomputationalframework}.
While entanglement is by now a well-characterized concept in many-body physics~\cite{Laflorencie_2016,amico2008entanglement}, the role of magic has only recently come into focus, driven by the development of stabilizer entropies~\cite{Leone2022SRE,leone2024stabilizer}, scalable probes of nonstabilizerness whose computational cost is comparable to that of entanglement entropy~\cite{huang2026fastexactapproachstabilizer,xiao2026exponentiallyacceleratedsamplingpauli,sierant2026computingquantummagicstate} and which have an important operational meaning~\cite{Bittel_2026, Iannotti_2026, cusumano2025nonstabilizernessviolationschshinequalities}.
This advance has stimulated a broad program of research, from equilibrium phenomena~\cite{piro2,tarabunga2023manybody,oshino2026stabilizer,hoshino2025stabilizerrenyientropyencodes, aditya2025mpembaeffectsquantumcomplexity} to quantum-circuit and Hamiltonian dynamics~\cite{Turkeshi_2025,odavic2025stabilizer,tirrito2025anticoncentration,falcao2025nonstabilizerness, PhysRevB.106.214316, p8dn-glcw, lio2026quantumstatedesignsmagic, aditya2025growthspreadingquantumresources, Viscardi_2026}.

A central open question concerns the interplay between entanglement and magic~\cite{tirrito2024quantifying,szombathy,Iannotti2025entanglement}.
The recently introduced \emph{non-local magic}~\cite{Cao_2025,qian2025quantumnonlocalnonstabilizerness, munizzi2026magicnoncliffordgatestopological, cao2024non} addresses precisely this problem: it isolates the irreducible nonstabilizerness, as quantified by stabilizer entropy, that cannot be removed by local basis changes, thereby separating genuinely correlated magic from purely local contributions.
However, this conceptual advance comes at a steep computational cost: in general, non-local magic is defined through a minimization over all local unitaries acting on the full Hilbert space, rendering it intractable beyond a few qubits. 
Although estimation schemes exist for matrix-product states~\cite{Cao_2025}, where entanglement remains bounded, which broader classes of quantum matter admit an analytic characterization of non-local magic, especially in regimes with extensive entanglement, has remained an open question.

Here, we answer this question for fermionic Gaussian states and unitaries. We introduce the fermionic non-local (FNL) magic, which restricts the minimization to Gaussian unitaries. 
We prove (cf. Theorem~\ref{thm:global}) that the FNL magic of order $\RENYI=2$ admits a closed-form expression (Eq.~\eqref{eq:analytic_NL}) in terms of the eigenvalues of the reduced Majorana covariance matrix on a subsystem of size $\ell$, computable with $\mathrm{poly}(\ell)$ resources. 
The key insight behind this result is that the minimization over local \emph{Gaussian} unitaries is solved exactly by the BCS canonical form~\cite{BoteroReznik2004}. 
This is in stark contrast with the total magic, which requires Majorana Monte Carlo~\cite{Bera_2025} or perfect sampling~\cite{collura2026nonstabilizernessfermionicgaussianstates}, both converging in general with an exponential number of samples.
We leverage this result to obtain a series of rigorous physical insights: an exact Page-like curve for typical fermionic Gaussian states, corroborated numerically by SYK$_2$ eigenstates~\cite{Bera_2025,santra2025complexitytransitionschaoticquantum,jasser2025stabilizer, malvimat2026multipartitenonlocalmagicsyk, njgn-fksh}; logarithmic scaling of non-local magic at the XY critical point, captured analytically via the Fisher--Hartwig approach; a quasiparticle picture~\cite{calabrese2009} for magic spreading after quantum quenches; and diffusive growth in brickwork Gaussian random circuits.
Because our result relies solely on two-point correlation functions, it also provides a scalable route for the experimental estimation of fermionic non-local magic in large-scale quantum processors via fermionic shadow tomography~\cite{Wan_2023}.

\section{Methods}
\label{sec:methods}
We begin by introducing the notation and key definitions; we refer to Refs.~\cite{Mbeng2024,surace,sierant2026fermionic,sierant2026theorymatchgatecommutant} for pedagogical introductions.
We consider fermionic Gaussian systems on $N$ qubits in a lattice $\Lambda=\{1,\dots, N\}$, described by a Hilbert space of dimension $d=2^N$.
The natural degrees of freedom are the $2N$ Majorana operators $\gamma_j$, obtained via the Jordan--Wigner transformation $\gamma_{2j-1} = \bigl(\prod_{m=1}^{j-1} Z_m\bigr)X_j$ and $\gamma_{2j} = \bigl(\prod_{m=1}^{j-1} Z_m\bigr)Y_j$, where $X_j$, $Y_j$, $Z_j$ are the Pauli operators on qubit $j$.
These operators are Hermitian, $\gamma_j=\gamma_j^\dagger$, and satisfy the Clifford algebra $\lbrace \gamma_j,\gamma_k\}=2 \delta_{jk} \mathbb{I}$.

Fermionic Gaussian unitaries (FGU), or matchgates, are the operators that act adjointly on Majorana as rotations $U_O\gamma_j U_O^\dagger = \sum_{m}O_{mj}\gamma_m$, with $O\in \mathrm{O}(2N)$; they form a group, which we denote $\mathfrak{M}_N$. 
Orthogonal matrices can be decomposed as $O=R^{(1-\det O)/2}Q$ with $Q\in \mathrm{SO}(2N)$ and $R=\mathrm{diag}(1,-1,\dots,-1)$  a reflection.
This implies that $U_O=X_1^{(1-\det(O))/2}\exp[-iH]$ is split into a part induced by special orthogonal transformations, captured by the so-called fermionic Gaussian Hamiltonians 
\begin{equation}
    H= \frac{i}{4}\sum_{m,n=1}^{2N} h_{m,n} \gamma_m\gamma_n\;,\label{eq:ham}
\end{equation}
where $h=-h^T$ is the $2N\times 2N$ antisymmetric matrix such that $Q=\exp[h]$, and the implementation of the reflection $R$ via the Pauli operator $X_1=\gamma_1$. 

A pure fermionic Gaussian state (FGS), or free fermionic state, is defined by $|\Psi_O\rangle=U_O|\mathbf{0}\rangle$ for the fermionic Gaussian unitary induced by $O\in \mathrm{O}(2N)$, and $|\mathbf{0}\rangle\langle\mathbf{0}| = d^{-1}\prod_{m=1}^{N}(\mathbb{I}-i\gamma_{2m-1}\gamma_{2m})$ is the fermionic vacuum state. 
Denoting for any pure state $|\Psi\rangle$ the covariance matrix $\Gamma_{mn}=i\langle \Psi|[\gamma_m,\gamma_n]|\Psi\rangle/2$, fermionic Gaussian states are fully characterized by the covariance matrix. In fact, given the vacuum covariance matrix
\begin{equation}
    \Gamma_0= { \bigoplus_{m=1}^N \begin{pmatrix}
        0 &-1 \\1&0
    \end{pmatrix}}\;,
\end{equation}
using the adjoint transformation, one obtains that $|\Psi_O\rangle\langle \Psi_O|$ is uniquely fixed by $\Gamma= O\Gamma_0 O^T$. 
For pure fermionic Gaussian states and any bipartition $\Lambda=A\sqcup B$ with $|A|=\ell\le N/2$ (made precise below), Ref.~\cite{BoteroReznik2004} showed that local FGU, $U_{O_A} \otimes U_{O_B}\in \mathfrak{M}_{N}$, are sufficient to bring the state into a canonical modewise form of independent two-mode BCS-like pairs 
\begin{equation}
\begin{split}
     |\Psi_{\mathrm{can}}\rangle
:= (T_{\sigma_{A}}\otimes T_{\sigma_{B}})\bigotimes_{i=1}^{\ell} [|\mathrm{B}(\theta_i)\rangle]\otimes |0\rangle^{\otimes (N-2\ell)}\;,
\end{split}
\label{eq:can_form}
\end{equation} 
where $|\mathrm{B}(\theta_i)\rangle=\cos\theta_i\,|00\rangle_i+\sin\theta_i\,|11\rangle_i$  and $T_{\sigma_X}$ are permutations that act separately on qubits in $X\in \{A,B\}$. This form is analogous to the Schmidt decomposition for pure fermionic Gaussian states, and it is fundamental for the main result.

\subsection{Non-local magic and its fermionic counterpart}
We now present our main quantity: non-local magic.
We recall that the stabilizer entropy of order $\RENYI$ is defined for a pure state $|\Psi\rangle$ as $M_\RENYI(|\Psi\rangle)=(1-\RENYI)^{-1}\log_2 \zeta_\RENYI(|\Psi\rangle)$~\cite{Leone2022SRE}, where $\zeta_\RENYI= d^{-1}\sum_{P\in \mathcal{P}_N} \langle \Psi|P|\Psi\rangle^{2\RENYI}$ is the stabilizer purity and $\mathcal{P}_N=\{ \mathbb{I},X,Y,Z\}^{\otimes N}$  is the set of $N$-qubit Pauli strings. This quantity is a magic monotone for integer $\RENYI\ge 2$~\cite{piro,leone2024stabilizer}.

The non-local magic (NLM) is defined for the bipartition $\Lambda = A\sqcup B$ and a pure state as 
\begin{equation}
    M_\RENYI^\textup{NL}:= \min_{U_A , U_B} M_\RENYI\left[T_{\pi_{AB}} (U_A\otimes U_B)T_{\pi_{AB}}^\dagger|\Psi\rangle\right]\;,\label{eq:defNLmagic}
\end{equation}
where $A=\{i_1,\dots,i_\ell\}\subset \Lambda$ contains $\ell$ qubits and $B=\Lambda\setminus A$ the remaining $N-\ell$ ones, $\pi_{AB}\in \mathrm{S}_N$ is any permutation such that $\pi_{AB}(A)=\{1,\dots,\ell\}$ and $\pi_{AB}(B)=\{\ell+1,\dots, N\}$, and the minimization is over $\mathrm{U}(d_A)\otimes \mathrm{U}(d_B)$ with $d_A=2^{\ell}$ and $d_B = 2^{N-\ell}$. This is the magic that cannot be erased or extracted by local unitary operations. This notion is cognate to (but distinct from) that of {\em long-range magic}, that is, the non-stabilizerness that cannot be removed by finite-depth local unitary circuits \cite{Korbany2025}. One salient distinction is that maximally entangled states have vanishing non-local magic as they are locally equivalent to stabilizer states.
This minimization is superexponentially costly, rendering non-local magic intractable beyond a handful of qubits.
In Ref.~\cite{Cao_2025}, it was shown that for $\RENYI = 2$ a reliable proxy for non-local magic can be obtained from the Schmidt spectrum of the state over the same bipartition.

To study fermionic Gaussian states and their non-local magic, we introduce the fermionic non-local magic
\begin{equation}
    M_\RENYI^\textup{FNL}:= \!\!\min_{U_{O_A} , U_{O_B}} \!\!M_\RENYI\left[T_{\pi_{AB}} (U_{O_A}\otimes U_{O_B})T_{\pi_{AB}}^\dagger|\Psi\rangle\right],\label{eq:defFNLmagic}
\end{equation}
where now the minimization is over the Gaussian unitaries induced by $\mathrm{O}(2\ell)\oplus \mathrm{O}(2(N-\ell))$, hence $ M_\RENYI^\textup{NL} \leq  M_\RENYI^\textup{FNL}$. Up to a non-Gaussian permutation, this quantity corresponds to the minimization over the full unitary group for Gaussian systems.
Throughout this work our main interest is the case $\RENYI=2$, the smallest R\'enyi index for which the stabilizer entropy is a magic monotone~\cite{piro,leone2024stabilizer}, the most studied in the literature, and the one we implicitly refer to when writing simply \emph{fermionic non-local magic}. Nevertheless, we will also comment on results for $\RENYI>2$, in particular in Appendix~\ref{Sec:generaliz}.
Our main result is that, for fermionic Gaussian states, the fermionic non-local magic can be evaluated exactly in $\mathrm{poly}(N)$ resources.

\section{Results}
Consider a pure fermionic Gaussian state $|\Psi_O\rangle$ with covariance matrix $\Gamma=O \Gamma_0 O^T$ and the bipartition $\Lambda = A\sqcup B$, where we denote without loss of generality $A$ the smallest partition with $|A|=\ell\le N/2$. 
Our main result shows that the FNL magic is entirely determined by the spectrum of the reduced covariance matrix. Let $\{\lambda_i\}_{i=1,\dots,\ell}$ denote the non-negative eigenvalues of the $2\ell\times 2\ell$ reduced covariance matrix $i\Gamma_A := i\Gamma|_{\mathcal{I}_A}$ of the Majorana modes restricted to $A$, $\mathcal{I}_A:=\{2m-1,2m\;|\; m\in A\}$. Then, our main finding is the following closed-form expression at $\RENYI=2$.

\begin{theorem}[Fermionic non-local magic of Gaussian states at $\RENYI=2$]
\label{thm:global}
Let $|\Psi_O\rangle$ be an arbitrary pure fermionic Gaussian state on $N$ modes and let $\Lambda=A\sqcup B$ with $|A|=\ell$. Then
\begin{equation}
\begin{split}
    M_{2}^{\mathrm{FNL}}(|\Psi_O\rangle)&=\sum_{i=1}^{\ell}\mathfrak{m}_2(\lambda_i^2)\;,\\
    \mathfrak{m}_2(x)&=-\log_2\!\left(1-x+x^2\right),
\end{split}
    \label{eq:analytic_NL}
\end{equation}
or, equivalently, in terms of the stabilizer purity,
\begin{equation}
\begin{split}
\max_{\substack{O_A\in \mathrm{O}(2\ell)\\ O_B\in \mathrm{O}(2(N-\ell))}}\!\!\!\!\zeta_2&\left[T_{\pi_{AB}}(U_{O_A}\otimes U_{O_B})T_{\pi_{AB}}^\dagger|\Psi_O\rangle\right]\\ &=\prod_{i=1}^{\ell}\left(1-\lambda_i^2+\lambda_i^4\right)\;,
\end{split}
\label{eq:thm_zeta}
\end{equation}
the maximum being attained in the Botero--BCS canonical axes of Eq.~\eqref{eq:can_form}.
\end{theorem}

The proof, holding for any $1\le\ell\le N/2$\DIFaddend , is given in full detail in Appendix~\ref{app:proof}. It assumes that the state $|\Psi_O\rangle$ is pure and Gaussian.
The key point is that local Gaussian unitaries over a bipartition of the system bring any pure fermionic Gaussian state into a canonical modewise form $\ket{\Psi_{\rm can}}$ in Eq.~\eqref{eq:can_form}, consisting of independent two-mode Bell-like pairs~\cite{BoteroReznik2004}, a normal form that underlies several equivalent descriptions of fermionic Gaussian entanglement~\cite{SpeeSchwaigerGiedkeKraus2018,MorralYepesLangerGammonSmithKrausPollmann2025,langer2026matchgatecircuitrepresentationfermionic}.
We prove that this representative is the \emph{global} maximizer of the stabilizer purity $\zeta_2$ on the local Gaussian orbit, i.e.
$M_{2}^{\mathrm{FNL}}(|\Psi_O\rangle) = M_{2}(|\Psi_{\rm can}\rangle)$,
from which Eq.~\eqref{eq:analytic_NL} follows by direct computation of the two-mode stabilizer entropy~\cite{Cao_2025,qian2025quantumnonlocalnonstabilizerness}.

For integer $\RENYI>2$ we conjecture that the same structure survives, with the R\'enyi weight adapted accordingly,
    \begin{equation} 
    \begin{split}
      M_{\RENYI}^{\mathrm{FNL}}(|\Psi_O\rangle)\! &= \sum_{i=1}^{\ell}\mathfrak{m}_\RENYI(\lambda_i^2)\;,
    \end{split}
    \label{eq:analytic_NL_alpha}
    \end{equation}
\noindent
where $\mathfrak{m}_\RENYI(x)$ is the weight function
\begin{equation}
    \begin{split}
      \mathfrak{m}_\RENYI(x)&:=
    \frac{1}{1-\RENYI}\log_2\! \left[\frac{\left(1-x \right)^{\RENYI }\!\!+1+ x^{ \RENYI }}{2} \right]\;, 
    \end{split}
    \label{eq:analytic_NL_weights}
    \end{equation}
which reduces to Eq.~\eqref{eq:analytic_NL} at $\RENYI=2$. This conjecture is supported by the extensive variational optimization reported in Appendix~\ref{app:variational}, and it reproduces the same phenomenology as the $\RENYI=2$ case, as we discuss in Appendix~\ref{Sec:generaliz}.

Several remarks are in order.
First, the proof of Theorem~\ref{thm:global} relies heavily on the specific quadratic (collision) functional form. Proving the result for $\RENYI>2$ would require a different proof strategy, which we leave open for future investigations.
Second, for small subsystems the restriction to Gaussian unitaries is immaterial. For $\ell\le 2$, structural simplifications make the Gaussian minimization coincide with the full one, $M_2^{\mathrm{NL}}=M_2^{\mathrm{FNL}}$, and Eq.~\eqref{eq:analytic_NL} reproduces the expression obtained in Ref.~\cite{Cao_2025} from the Schmidt spectrum. For larger subsystems, this is no longer the case: non-Gaussian local transformations can further decrease the magic even when the initial state is Gaussian, so that Theorem~\ref{thm:global} provides an upper bound on the generic non-local magic, $M_2^{\mathrm{NL}}\le M_2^{\mathrm{FNL}}$.
Lastly, our formula is considerably simpler even than computing the \emph{total} stabilizer entropy $M_\RENYI$ for fermionic Gaussian states, which requires either Majorana Monte Carlo~\cite{Bera_2025} or perfect sampling~\cite{collura2026nonstabilizernessfermionicgaussianstates}, both requiring an exponential number of samples for convergence in the worst case. In contrast, evaluating $M_\RENYI^\mathrm{FNL}$ via Eqs.~\eqref{eq:analytic_NL} and~\eqref{eq:analytic_NL_alpha} requires only $\mathrm{poly}(\ell)$ operations.

In the following, we will consider throughout the case $\RENYI=2$, where the stabilizer entropy is a magic monotone and where Theorem~\ref{thm:global} applies. (Nevertheless, as aforementioned, we have checked that the same phenomenology holds for $\RENYI>2$, see Appendix~\ref{Sec:generaliz}.)
We will show that the structure highlighted in Eq.~\eqref{eq:analytic_NL} enables several analytical results, both at equilibrium and out-of-equilibrium.

\subsection{Simple example: entanglement without non-local magic}
\label{sec:simple}
Before presenting applications, we illustrate the physical content of the non-local magic with a simple but instructive example that clarifies the distinction between entanglement and non-local magic.
Consider $N$ qubits with $\ell = N/2$ and the rainbow state
\begin{equation}
    |\Psi_{\mathrm{rb}}\rangle = \bigotimes_{j=1}^{\ell} \frac{1}{\sqrt{2}}\bigl(|0_j 0_{\bar\jmath}\rangle + |1_j 1_{\bar\jmath}\rangle\bigr)\;,
\end{equation}
where $\bar\jmath = N+1-j$ is the mirror site in $B=\{\ell+1,\dots,N\}$, for the bipartition $A=\{1,\dots,\ell\}$. This is a fermionic Gaussian state: it is the ground state of a free-fermion hopping Hamiltonian with suitable long-range pairing. Each pair $(j,\bar\jmath)$ forms a maximally entangled Bell pair, so every eigenvalue of $i\Gamma_A$ vanishes, $\lambda_j = 0$, and the entanglement entropy is maximal: $S = \ell$ bits, a volume law. Yet, this state is a stabilizer state: it belongs to the Clifford orbit of the computational basis. Consequently its magic, and therefore its non-local magic, vanishes identically: $M_2^{\mathrm{NL}} = 0$.

This result is transparent from our formula: $\mathfrak{m}_2(\lambda_j^2)$ vanishes both at $\lambda_j = 0$  (maximally entangled mode) and at $\lambda_j = 1$  (unentangled mode). Non-local magic requires modes at \emph{intermediate} entanglement, $0 < \lambda_j < 1$, where the bipartite correlations are neither trivial nor of the Bell-pair type.

One might attempt to introduce magic into the rainbow state by applying local unitaries, e.g.\ a non-Clifford single-qubit rotation $R(\phi)$ on site $j\in A$. This produces a state $(R(\phi)\otimes \mathbb{I})|\mathrm{Bell}\rangle_{j\bar\jmath}$ that carries nonzero total magic. However, by construction, this magic is entirely \emph{local}: it can be removed by the inverse rotation $R(-\phi)\otimes \mathbb{I}$ acting on $A$ alone, and therefore does not contribute to $M_2^{\mathrm{NL}}$. The only way to create genuine non-local magic is to entangle the two qubits of each Bell pair in a non-stabilizer fashion, for instance, by applying a joint non-Clifford \emph{Gaussian} rotation $V_{j\bar\jmath}$~\cite{huang2026dismagickerunitarygatenonstabilizerness} that acts simultaneously on sites in $A$ and $B$, so that the state remains Gaussian and Theorem~\ref{thm:global} applies. Such a rotation shifts $\lambda_j$ away from $0$ and $1$, placing it in the interior of $[0,1]$ where $\mathfrak{m}_2(\lambda_j^2) > 0$. In other words, entanglement is necessary but not sufficient for non-local magic: a state can be maximally entangled yet carry zero nonstabilizerness across the bipartition.

\subsection{Experimental observability}
A key practical advantage of Eq.~\eqref{eq:analytic_NL} is that the fermionic non-local magic of Gaussian states is entirely determined by the reduced covariance matrix $\Gamma_A$, a polynomial number of two-point correlation functions.
While a proof-of-principle demonstration of non-local magic has recently been reported on a superconducting processor in a two-qubit setting~\cite{ahmad2025experimentaldemonstrationnonlocalmagic}, our result opens a scalable route applicable to large fermionic Gaussian systems.

A natural experimental protocol is \emph{fermionic shadow tomography}~\cite{Wan_2023}: random matchgate circuits are applied, followed by computational-basis measurements, and $\Gamma_A$ is inferred from the resulting classical-shadow data.
A conservative sample-complexity estimate yields, for failure probability at most $\delta$,
\[
N_{\mathrm{shots}} = O\!\left(\frac{\ell^2 \log(\ell/\delta)}{\epsilon_\Gamma^2}\right),
\]
to reconstruct \(\Gamma_A\) with accuracy \(\epsilon_\Gamma\), while the propagation of this error  to \(M_2^{\mathrm{FNL}}\) is discussed in  Appendix~\ref{app:experimental}.

\subsection{Fermionic Non-local magic of typical Gaussian states}
Having established the exact formula, we first investigate the fermionic non-local magic of \emph{typical} Gaussian states, i.e., states $|\Psi_O\rangle=U_O|\mathbf{0}\rangle$ generated by Haar-random $O\in \mathrm{O}(2N)$.
Our goal is to compute the average $\overline{M}_2^\mathrm{FNL}\equiv \mathbb{E}_{O\sim \mathrm{O}(2N)}[M_2^\mathrm{FNL}(|\Psi_O\rangle)]$, which for our closed-form expression reduces to a random matrix theory problem~\cite{Mehta2004,Livan_2018,vidmar2017entanglement}.
For a fixed subsystem $A$ of size $\ell\le N/2$, let us define the stochastic variables $x_k:=\lambda_k^2$, whose joint probability distribution is a Jacobi ensemble~\cite{Mehta2004}
\begin{equation}
    \mathbb{P}(x_1,\dots,x_\ell) = \frac{\Delta(x_1,\dots,x_\ell)}{Z_\ell}  \prod_{m=1}^\ell \frac{(1-x_m)^{N-2\ell}}{\sqrt{x_m(1-x_m)}}\;,
\end{equation}
with $\Delta(x_1,\dots,x_\ell)=\prod_{1\le i<j\le \ell} (x_i-x_j)^2$ the squared Vandermonde determinant  and $Z_\ell$ the normalization ensuring $\int \mathbb{P}(x_1,\dots,x_\ell)\prod_{m=1}^\ell (dx_m)=1$. 
Since $M_2^\mathrm{FNL}=\sum_{m=1}^\ell \mathfrak{m}_2(x_m)$ is a sum over separate $x_m$ contributions, the only ingredient required is the one-point density function 
\begin{equation}
\begin{split}
    \varrho^{(\ell)}(x):&=\mathbb{E}_\mathbb{P}\left[\sum_{m=1}^\ell \delta(x-x_m)\right] \\&= \ell \int \mathbb{P}(x,x_2,\dots,x_\ell)\prod_{m=2}^\ell (dx_m)\:,
\end{split}
\end{equation}
where we used relabelling in the second step. Notice that $\int_0^1 \varrho^{(\ell)}(x) dx = \ell$.
Thus, the average fermionic non-local magic is 
\begin{equation}
    \overline{M}_2^\mathrm{FNL}=\int_0^1 \mathfrak{m}_2(x) \varrho^{(\ell)}(x) dx\;.
\end{equation}
For the Jacobi ensemble, the one-point density is given by the standard orthogonal-polynomial kernel.
Given $w(x) = (1-x)^{N-2\ell}/\sqrt{x(1-x)}$, we have~\cite{sierant2026theorymatchgatecommutant,sierant2026fermionic}
\begin{equation}
\varrho^{(\ell)}(x) = w(x) \sum_{n=0}^{\ell-1} [p_n(x)]^2\;,
\label{eq:ziopino}
\end{equation}
where $p_n$ are polynomials in the domain $[0,1]$, that are orthonormal with respect to the weight $w(x)$, namely $p_n(x) = \mathrm{P}_n^{(-1/2,N-2\ell-1/2)}(1-2x)/\sqrt{h_n}$ in terms of the Jacobi polynomials $\mathrm{P}_n^{(a,b)}$ and with normalization 
\begin{equation}
    h_n=\frac{\Gamma(n+1/2) \Gamma(n+N-2\ell+1/2)}{(2n+N-2\ell) \Gamma(n+1)\Gamma(n+N-2\ell)}.
\end{equation} 
Eq.~\eqref{eq:ziopino} yields the exact average non-local magic at finite size $N$ for a given $\ell$, which we benchmark against exact numerics in Fig.~\ref{fig:typicality}(a), finding perfect agreement. We emphasize that this is a \emph{quenched} average \DIFdelbegin \DIFdel{, the average is }\DIFdelend \DIFaddbegin \DIFadd{--- }\DIFaddend performed outside the logarithm \DIFdelbegin \DIFdel{, }\DIFdelend \DIFaddbegin \DIFadd{--- }\DIFaddend which is typically harder to evaluate than its annealed counterpart.

\begin{figure}[t!]
    \centering
    \includegraphics[width=\columnwidth]{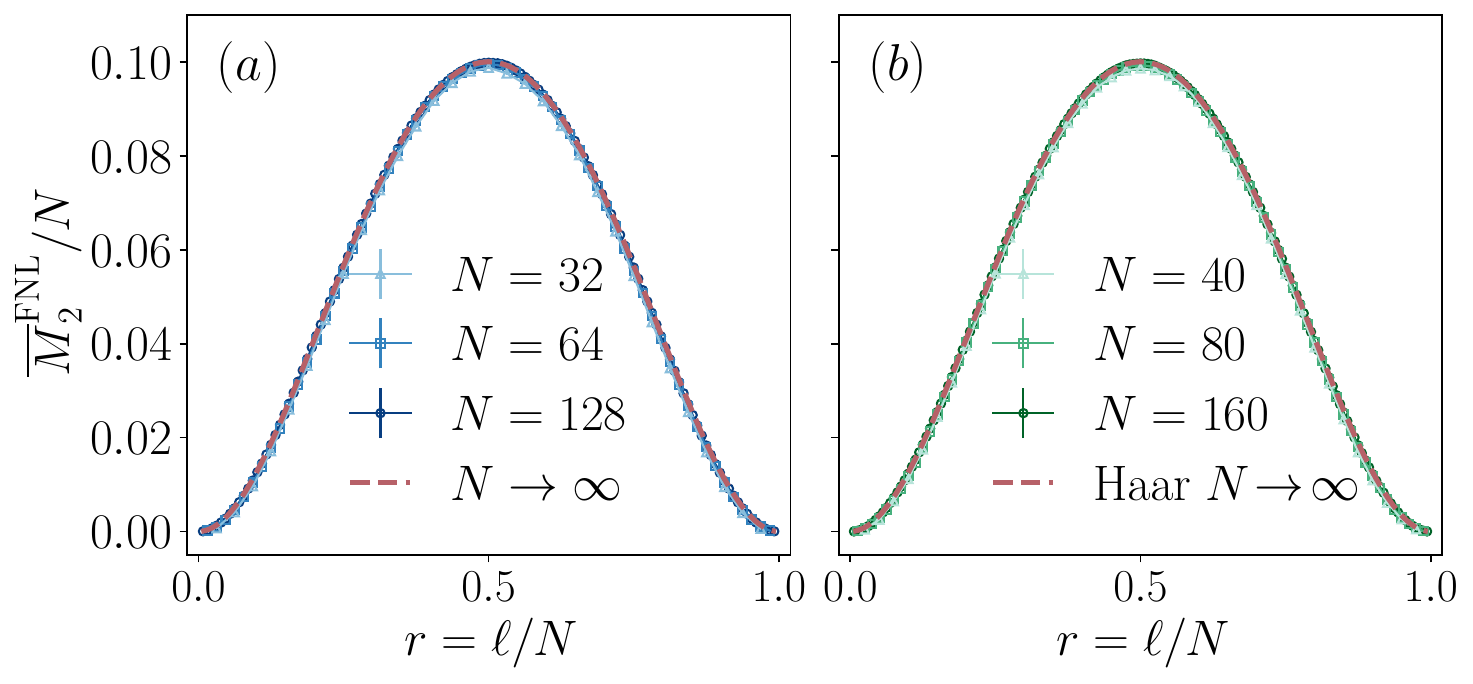}
    \caption{\textbf{Fermionic non-local magic of typical fermionic Gaussian states.}
    (a)~Page curve for Haar-random FGS: $\overline{M_2}^{\mathrm{FNL}}/N$ is plotted against the ratio $r=\ell/N$. The results are compared with the closed form expression for $N\to\infty$ in Eq.~\eqref{eq:ziopin2}.
    (b)~The fermionic non-local magic Page curve is tested against the numerically computed average non-local magic for SYK$_2$ mid-spectrum eigenstates. The results are in quantitative agreement with Eq.~\eqref{eq:ziopin2}. }
    \label{fig:typicality}
\end{figure}

We now turn to the thermodynamic limit $\ell,N\to\infty$ at fixed ratio $r=\ell/N\in[0,1/2]$, defining the fermionic non-local magic density $\overline{\mathsf{m}}_2^\mathrm{FNL}:=\overline{M}_2^\mathrm{FNL}/N$ (the final result is symmetric under $r\to 1-r$).
In this regime, the rescaled  one-point density $\varrho^{(\ell)}/N$ converges to the free Jacobi law with support $x\in [0,c]$, $c=4r(1-r)$,
\begin{equation}
    \varrho_r(x) = \frac{\sqrt{(c-x)x}}{2\pi x(1-x)}\mathbf{1}_{[0,c]}(x)\;,
\end{equation}
where $\mathbf{1}_A(x)$ is the characteristic function of $A$, $\mathbf{1}_A(x)=1$ if and only if $x\in A$. 
The total mass of the density is $\int_0^c \varrho_r(x) dx=r$, which is precisely the fraction of the number of eigenvalues over the system size. 
Thus, in the scaling limit, we have
\begin{equation}
\overline{\mathsf{m}}_2^\mathrm{FNL}(r)\simeq \int_0^c \mathfrak{m}_2(x) \varrho_r(x)dx\;.
\end{equation}
A more convenient parametrization is obtained by setting $x=c\sin^2\theta$  with $\theta\in[0,\pi/2]$. Then the integral is 
\begin{equation}
\overline{\mathsf{m}}_2^\mathrm{FNL}(r)\simeq \frac{c}{\pi}\int_0^{\pi/2} \frac{\cos^2 \theta}{1-c\sin^2\theta} \mathfrak{m}_2(c\sin^2 \theta)d\theta\;.\label{eq:ziopin2}
\end{equation}
At exactly half bipartition, $r=1/2$, this integral can be performed exactly and the formula reduces to 
$\overline{\mathsf{m}}_2^\mathrm{FNL}(r)\simeq 2-\log_2(2+\sqrt{3})$.
On the other hand, for small $r\simeq 0$ (and by symmetry $r\simeq 1$), we have $c\simeq 4r$ and the integral yields $\overline{\mathsf{m}}_2^\mathrm{FNL}(r\simeq 0)\simeq \frac{r^2}{\ln(2)}+\mathcal{O}(r^3)$. 
This {\em fermionic non-local magic Page curve} shows a  quadratic onset,  for every \DIFdelbegin \DIFdel{$\alpha$ }\DIFdelend \DIFaddbegin \DIFadd{$\RENYI$ }\DIFaddend (see Appendix Sec.~\ref{Sec:generaliz}), with a transparent physical origin: when $\ell$ is of order one while $N\to\infty$, the $\ell$ modes of $A$ are almost, or exactly, maximally entangled and thus carry no non-local magic. 
This is in sharp contrast with the entanglement Page curve, which vanishes only at $|A|=0$ and grows linearly for small $r$~\cite{PhysRevE.106.034118,PhysRevLett.125.180604,PhysRevB.103.L241118}.

To test the universality of this result beyond Haar-random states, we consider the Sachdev--Ye--Kitaev model at quadratic order, $\mathrm{SYK}_2$.
This model is specified by the  Hamiltonian  $H$ of the form in  Eq.~\eqref{eq:ham} with  $h$ a random antisymmetric matrix whose elements are drawn from the normal distribution $h_{m,n}\sim \mathcal{N}(0,{J^2}/{N})$ and set $J=1$. 
The single-body Hamiltonian $ih$ is Hermitian with eigenvalues in $\pm\varepsilon_k$ pairs ($k=1,\dots,N$), each pair defining an effective fermionic mode. A many-body eigenstate is specified by a filling $\mathcal{F}\subset\{1,\dots,2N\}$ of $N$ single-particle modes with eigenvectors $|v_k\rangle$, yielding the one-body density matrix $\Xi=\sum_{k\in\mathcal{F}}|v_k\rangle\langle v_k|$ and the covariance matrix $\Gamma = i(2\Xi-\mathbb{I})$. 
The ground state fills the $N$ lowest-energy modes; to probe typicality, we construct mid-spectrum eigenstates $H|E_n\rangle=E_n|E_n\rangle$  by flipping a random subset of $N/4$ to $N/2$ modes from the ground-state configuration, replacing occupied negative-energy modes with their positive-energy partners. For each of $N_{\mathrm{dis}}=1000$ disorder realizations, we compute $M_2^{\mathrm{FNL}}$ from $\Gamma$ and average over 200 excited eigenstates $|E_n\rangle$ per realization, obtaining $\overline{M}_2^{\mathrm{FNL}}=\mathbb{E}_{h}\,\mathbb{E}_{n}[M_2^{\mathrm{FNL}}(|E_n\rangle)]$.
In Fig.~\ref{fig:typicality}(b), we plot the resulting non-local magic density $\overline{\mathsf{m}}_2^\mathrm{FNL}$ versus $r=\ell/N$ for $N=40,80,160$.
The data converge to the analytical Haar prediction Eq.~\eqref{eq:ziopin2} as $N$ increases, confirming that mid-spectrum eigenstates of SYK$_2$ reproduce the typicality of fermionic non-local magic expected from random fermionic Gaussian states.

\subsection{Fermionic Non-local magic in ground states}

\begin{figure*}
    \centering
    \includegraphics[width=\linewidth]{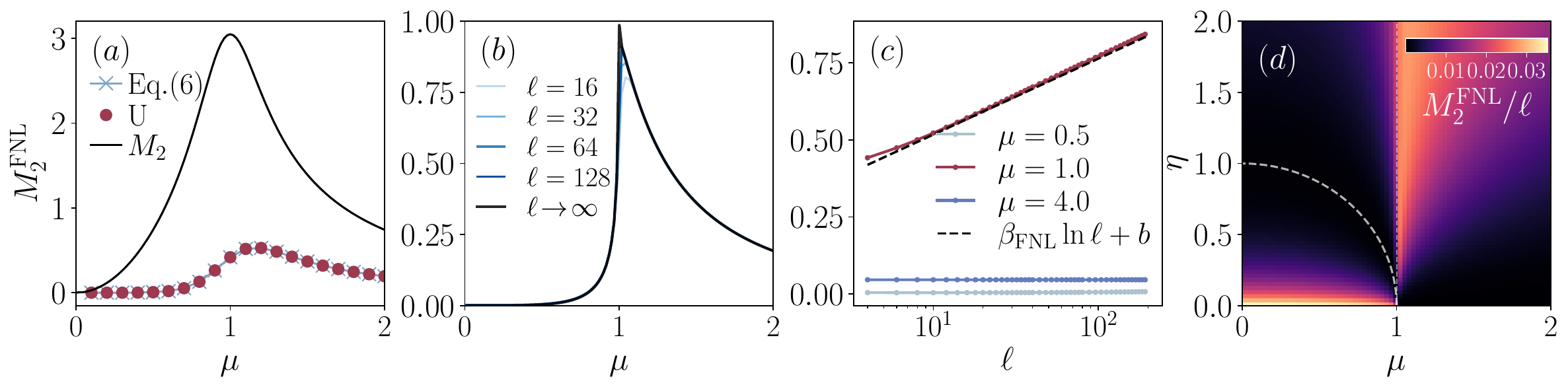}
    \caption{%
    \textbf{Equilibrium fermionic non-local magic in the Ising and XY models.}
    (a)~Validation of Eq.~\eqref{eq:analytic_NL} for the Ising ground state with $N=8$, $\ell=4$: the fermionic non-local magic $M_2^{\mathrm{FNL}}$ from Eq.~\eqref{eq:analytic_NL} (crosses) agrees with the brute-force minimization over all local unitaries $U_A \otimes U_B$  (red marker) to numerical precision. The black line shows the total stabilizer entropy $M_2$ for reference.
    (b)~$M_2^{\mathrm{FNL}}$ vs.\ transverse field $\mu$ in the thermodynamic limit for subsystem sizes $\ell = 16, 32, 64, 128$; the black curve is the exact $\ell \to \infty$ result from the Peschel entanglement spectrum, Eq.~\eqref{eq:NL_peschel2}. The fermionic non-local magic peaks at the critical point $\mu_c = 1$ and is suppressed in both gapped phases.
    (c)~Scaling of $M_2^{\mathrm{FNL}}(\ell)$ at criticality ($\mu = 1$, red): logarithmic growth: the dashed line is the fit $\beta_{\mathrm{fit}}\ln\ell+b$,  with $\beta_{\mathrm{fit}} \approx 0.1072$  and $b \approx 0.27$, in agreement with the Fisher--Hartwig prediction Eq.~\eqref{eq:alpha_NL}. In the gapped phases ($\mu = 0.5$, $\mu = 4.0$), $M_2^{\mathrm{FNL}}$ saturates to a finite value.
    (d)~Heatmap of $M_2^{\mathrm{FNL}}/\ell$ over the XY phase diagram ($\mu$, $\eta$), computed from the exact Peschel formula. The fermionic non-local magic is concentrated near the critical lines and is strongly suppressed deep in the gapped phases.
    }
    \label{fig:Ising_model}
\end{figure*}

We now apply Eq.~\eqref{eq:analytic_NL} to study the equilibrium properties of fermionic non-local magic in the ground state of the XY chain
\begin{equation}
    H = -\sum_{j=1}^{N}\!\left[\frac{1+\eta}{2}X_j X_{j+1} + \frac{1-\eta}{2}Y_j Y_{j+1} + \mu Z_j\right],
    \label{eq:xy}
\end{equation}
with periodic boundary conditions.
Using the Jordan--Wigner transformation, Eq.~\eqref{eq:xy} is mapped to the quadratic Hamiltonian $H$of the form in Eq.~\eqref{eq:ham} with single-body matrix entries $h_{2m-1,2m} = 2\mu$, $h_{2m,2m+1} = 1+\eta$, $h_{2m+2,2m-1} = 1-\eta$, and all remaining elements fixed by antisymmetry~\cite{Mbeng2024}. 
Throughout this Section, the XY chain is studied as a \emph{fermionic} model, i.e.\ as a Majorana chain: after the Jordan--Wigner mapping the ground state is naturally written --- and exactly solvable --- as a fermionic Gaussian state, and the natural notion of locality is that of the fermionic modes. 
Accordingly, $M_\RENYI^{\mathrm{FNL}}$ is the non-local magic of the \textit{free-fermion problem}, and the local Gaussian unitaries of Eq.~\eqref{eq:defFNLmagic} are, for a contiguous cut, implemented by genuine tensor products $U_{O_A}\otimes U_{O_B}$ of qubit unitaries (cf. also Appendix~\ref{app:proof}).
Following standard practice in the literature, we simplify the nomenclature and keep referring to it as the XY model~\footnote{Note that, minimizing over all local spin unitaries would give a quantity bounded from above by ours, $M_\RENYI^{\mathrm{NL}}\le M_\RENYI^{\mathrm{FNL}}$, as discussed in the first remark below Theorem~\ref{thm:global}.}. 
We restrict to system sizes $N$ divisible by $4$, for which the ground state lies in the Neveu--Schwarz sector with antiperiodic fermionic boundary
conditions $\gamma_{N+1}=-\gamma_1$~\cite{Mbeng2024}.
The allowed momenta are 
$k=k(n)= \pi(2n+1)/N$ for $n = 0,\dots, N/2-1$, 
spanning $(0,\pi)$ with $N/2$ equally spaced points; 
in the thermodynamic limit $N\to\infty$, sums over $k(n)$ 
become integrals over $[0,\pi]$.
Then, the single-body Hamiltonian $ih$ is diagonalized by a Fourier transform to momentum space followed by a $2\times 2$ Bogoliubov rotation at each wave-vector $k$~\cite{Mbeng2024}.
The Bogoliubov angle $\theta_k$ is defined through
\begin{equation}
    \cos\theta_k = \frac{\mu - \cos k}{\varepsilon_k}\;,\qquad \sin\theta_k = \frac{\eta\sin k}{\varepsilon_k}\;,
    \label{eq:bogoliubov_angle}
\end{equation}
yielding the single-particle dispersion
\begin{equation}
    \varepsilon_k = \sqrt{(\mu-\cos k)^2 + \eta^2\sin^2 k}\;.
    \label{eq:dispersion}
\end{equation}
In the diagonal basis, the Hamiltonian reads $H = \frac{i}{2}\sum_k \varepsilon_k\,\tilde\gamma_{2k-1}\tilde\gamma_{2k}$, where $\tilde\gamma_m = \sum_n G_{mn}\gamma_n$ are the rotated Majorana operators.
The ground state is the vacuum of all Bogoliubov modes, with covariance matrix $\Gamma = G\,\Gamma_0\, G^T$.

From now on, we consider the bipartition fixed by the contiguous interval $A=\{1,\dots,\ell\}$ and its complement $B=\Lambda\setminus A$.
In the thermodynamic limit $N\to\infty$ with $\ell$ fixed, the reduced covariance matrix $\Gamma_A$ is a $2\ell\times 2\ell$ block-Toeplitz matrix with $2\times 2$ blocks depending only on $m-n$,
\begin{equation}
\begin{split}
    (\Gamma_A)_{mn} &= \int_{-\pi}^{\pi}\frac{dk}{2\pi}\,e^{ik(m-n)}\,g(k)\;,
    \\
    g(k) &= \begin{pmatrix} 0 & e^{i\theta_k} \\ -e^{-i\theta_k} & 0 \end{pmatrix}.
\end{split}
\label{eq:block_toeplitz}
\end{equation}
The fermionic non-local magic $M_2^{\mathrm{FNL}} = \sum_n \mathfrak{m}_2(\lambda_n^2)$ is therefore a spectral function of $\Gamma_A$ with symbol $g(k)$, providing a direct route to efficient numerical evaluation.
For a given Hamiltonian and subsystem size $\ell$, we construct $\Gamma_A$ by evaluating the integral in Eq.~\eqref{eq:block_toeplitz}, exact in the thermodynamic limit $N\to\infty$, and compute the non-local magic via Eq.~\eqref{eq:analytic_NL} from the eigenvalues $\pm\lambda_n$ ($\lambda_n\in[0,1]$) of $i\Gamma_A$.
This procedure reaches large subsystem sizes $\ell$ at classical computational cost $\mathcal{O}(\ell^3)$.

In the scaling limit $N\gg \ell\gg 1$, the asymptotic behavior of spectral functions of Toeplitz matrices is controlled by two classical results: the \emph{strong Szeg\H{o} theorem}, applicable when the symbol is smooth, and the \emph{Fisher--Hartwig formula}, which governs the corrections arising from singularities in the symbol~\cite{Jin_2004,Its_2005}.

Since the symbol~\eqref{eq:block_toeplitz} satisfies $|\det g(k)| = 1$ for all $k$, the Szeg\H{o} theorem rules out volume-law scaling of $M_2^{\mathrm{FNL}}$.
When $g(k)$ is smooth, as occurs in the gapped phases $\mu \neq 1$ where $\theta_k$ is analytic, the strong Szeg\H{o} theorem further implies that $M_2^{\mathrm{FNL}}$ saturates to a finite value as $\ell \to \infty$, with exponentially small corrections~\cite{simon2010szego}.
Logarithmic growth can arise only from singularities in the symbol, which for the Ising model ($\eta=1$) occur at the critical point $\mu = \mu_c = 1$, where the Bogoliubov angle $\theta_k$ develops a jump discontinuity at $k = 0$\DIFdelbegin \DIFdel{(and $k = \pi$)}\DIFdelend .
We now exploit this structure to obtain exact scaling forms in both regimes.

\subsubsection{Exact results off criticality via corner transfer matrices}
The reduced density matrix of subsystem $A$ is a fermionic Gaussian state, and can be written as~\cite{Peschel2004}
\begin{equation}
    \rho_A = \frac{\exp[{-H_{\mathrm{ent}}}]}{Z_{\mathrm{ent}}}\;, \quad H_{\mathrm{ent}} = \frac{i}{2}\sum_{n=1}^{\ell} \varepsilon_n^{\mathrm{ent}}\,\tilde\gamma_{2n-1}\tilde\gamma_{2n}\;,
    \label{eq:rho_ent}
\end{equation}
where $H_{\mathrm{ent}}$ is the entanglement Hamiltonian, $Z_{\mathrm{ent}}=\mathrm{Tr}(\exp[{-H_{\mathrm{ent}}}])$ is the normalization and $\tilde\gamma_m$ are the Majorana operators that diagonalize $\rho_A$, and $\varepsilon_n^{\mathrm{ent}} \geq 0$ are the single-particle entanglement energies.
The covariance matrix $\Gamma_A$ and the single-particle entanglement Hamiltonian $h_{\mathrm{ent}}$ implicitly defined in Eq.~\eqref{eq:rho_ent} are related by
\begin{equation}
    i\Gamma_A = \tanh\!\left(\frac{i\,h_{\mathrm{ent}}}{2}\right),
    \label{eq:Gamma_tanh}
\end{equation}
so that the eigenvalues of $i\Gamma_A$ are $\pm\lambda_n$ with $\lambda_n = \tanh\!\left(\frac{\varepsilon_n^{\mathrm{ent}}}{2}\right)$. 
Here $\lambda_n = 1$ corresponds to a pure (unentangled) mode and $\lambda_n < 1$ to an entangled one.

For the XY model away from criticality, the entanglement spectrum of a semi-infinite bipartition fixed by $A=\mathbb{Z}_+=\{1,2,\dots\}$ is known in closed form through the corner transfer matrix approach~\cite{Peschel2009,Eisler2017}.
As we show below, this leads to the exact computation of $M_2^\mathrm{FNL,\mathbb{Z}_+}$ in this regime. 
This regime approximates the scaling limit $N\gg\ell\gg1$.
The single-particle entanglement energies are equally spaced,
\begin{equation}
    \varepsilon_n^{\mathrm{ent}} = \begin{cases}
        (2n+1)\,\varepsilon & \mu > 1 \;,\\
        2n\,\varepsilon & \mu < 1 \;,
    \end{cases}
    \label{eq:peschel_spectrum}
\end{equation}
where $\varepsilon = \pi K(\kappa')/K(\kappa)$ is the level spacing, $K(x)$ denotes the complete elliptic integral of the first kind, $\kappa' = \sqrt{1-\kappa^2}$, and $\kappa$ is conditioned on the phase~\cite{Peschel2009}
\begin{equation}
    \kappa = \begin{cases}
        \displaystyle\frac{\eta}{\sqrt{\eta^2 + \mu^2 - 1}} & \mu > 1 \;,\\[8pt]
        \displaystyle\frac{\sqrt{\eta^2 + \mu^2 - 1}}{\eta} & \mu < 1,\; \eta^2 + \mu^2 > 1\;,\\[8pt]
        \displaystyle\sqrt{\frac{1 - \eta^2 - \mu^2}{1 - \mu^2}} & \mu < 1,\; \eta^2 + \mu^2 < 1\;.
    \end{cases}
    \label{eq:elliptic_modulus}
\end{equation}
This leads to the exact expression for the off-critical phases
\begin{equation}
    M_2^{\mathrm{FNL},\mathbb{Z}_+} = \sum_{n=0}^{\infty} \mathfrak{m}_2\!\left[\tanh^2\left(\frac{\varepsilon_n^{\mathrm{ent}}}{2}\right)\right].
    \label{eq:NL_peschel}
\end{equation}
For a block of $\ell$ sites in the infinite chain with periodic boundary conditions, two separate boundaries contribute, so that the fermionic non-local magic is twice the single-cut result
\begin{equation}
    M_2^{\mathrm{FNL}}(\ell) \simeq  2M_2^{\mathrm{FNL},\mathbb{Z}_+} = 2\sum_{n=0}^{\infty} \mathfrak{m}_2\!\left[\tanh^2\left(\frac{\varepsilon_n^{\mathrm{ent}}}{2}\right)\right],
    \label{eq:NL_peschel2}
\end{equation}
up to corrections exponentially small in $\ell/\xi$, where $\xi$ is the correlation length of the problem. 
Since $\mathfrak{m}_2(x)\simeq(1-x)/\ln 2$ vanishes linearly as $x \to 1$, and $1-\lambda_n^2 \to 0$ exponentially for large $n$, the series converges rapidly: for $\mu$ away from criticality, only a few terms contribute, consistently with the Szeg\H{o} prediction of a finite $M_2^{\mathrm{FNL}}$.

In Fig.~\ref{fig:Ising_model} we benchmark these predictions against numerical simulations with periodic boundary conditions for the Ising model ($\eta=1$).
In Fig.~\ref{fig:Ising_model}(a), we validate Eq.~\eqref{eq:analytic_NL} by comparing against a brute-force minimization over the full unitary group $\mathrm{U}(d_A)\otimes\mathrm{U}(d_B)$ for $N=8$, $\ell=4$: the two results agree to numerical precision.
In Fig.~\ref{fig:Ising_model}(b), we compare exact numerics for $N=10^6$ qubits and $\ell = 16,32,64,128$ with the asymptotic formula Eq.~\eqref{eq:NL_peschel2}, finding quantitative agreement. Notably, for $\mu<1$ the non-local magic is strongly suppressed, $M_2^{\mathrm{FNL}}$ being exponentially small in the level spacing $\varepsilon$ of the entanglement spectrum and vanishing exactly only in the GHZ limit $\mu\to0$. This is the ordered (ferromagnetic) phase, whose ground state approaches the GHZ state $\frac{1}{\sqrt{2}}(|0\rangle^{\otimes N}+|1\rangle^{\otimes N})$ as $\mu\to 0$. The suppression of non-local magic in this regime can be understood from the Peschel entanglement spectrum, Eq.~\eqref{eq:peschel_spectrum}: for $\mu<1$ and $\eta=1$, the single-particle entanglement energies are $\varepsilon_n^{\mathrm{ent}}=2n\varepsilon$. The $n=0$ mode has $\varepsilon_0=0$, yielding $\lambda_0=\tanh(0)=0$ (a maximally entangled, stabilizer-like Bell pair), while all $n\ge 1$ modes have $\lambda_n=\tanh(n\varepsilon)$ exponentially close to $1$ (nearly unentangled modes). Since $\mathfrak{m}_2(0)=\mathfrak{m}_2(1)=0$, every mode contributes an exponentially small amount of non-local magic. This is a striking example of a phase with nonzero entanglement entropy (area law) but negligible  non-local magic. In contrast, the entanglement entropy receives a finite area-law contribution from the $n=0$ maximally entangled mode.

At the critical point $\mu_c=1$ ($\eta=1$), the gap closes and the correlation length diverges as $\xi\simeq |\mu-1|^{-1}$; the spacing $\varepsilon \to 0$ (since $\kappa \to 1$ and $K(\kappa) \to \infty$), so the series in Eq.~\eqref{eq:NL_peschel2} diverges.
The Bogoliubov angle $\theta_k$ develops a jump discontinuity at $k = 0$, placing the problem squarely in the Fisher--Hartwig regime~\cite{Jin_2004}.
The key result is that the asymptotic density of the positive eigenvalues $\lambda_n \in [0,1)$  of $i\Gamma_A$ --- one per $\pm\lambda_n$ pair, with two Fisher--Hartwig boundary points for the block --- scales as $\rho(\lambda) \simeq 2\ln\ell\,[\pi^2(1-\lambda^2)]^{-1}$~\cite{Its_2005}.
Inserting this spectral density into the fermionic non-local magic,
\begin{equation}
    M_2^{\mathrm{FNL}}(\ell) = \sum_{n} \mathfrak{m}_2(\lambda_n^2)\;\simeq\; \beta_{\mathrm{FNL}}\,\ln\ell,
\end{equation}
with constant
\begin{equation}
    \beta_{\mathrm{FNL}} = \frac{2}{\pi^2}\int_{0}^{1} \frac{\mathfrak{m}_2(\lambda^2)}{1-\lambda^2}\, d\lambda= \frac{\pi ^2-\ln ^2\left(7-4 \sqrt{3}\right)}{\pi ^2 \ln (16)}\;,
    \label{eq:alpha_NL}
\end{equation}
where in the second step we performed the integral explicitly.
This logarithmic scaling is the non-local magic analogue of $S \simeq (c/3)\ln\ell$~\cite{Calabrese2004entanglement}, with the central charge $c$ replaced by $\beta_{\mathrm{FNL}}$.
In Fig.~\ref{fig:Ising_model}(c), a fit to the numerical data for $\ell > 50$ at $\mu = 1$ yields $\beta_{\mathrm{fit}} \approx 0.1072$, in excellent agreement with Eq.~\eqref{eq:alpha_NL}.

Finally, in Fig.~\ref{fig:Ising_model}(d) we extend our analysis to the full XY phase diagram. The fermionic non-local magic is concentrated near the critical lines and is strongly suppressed deep in the gapped phases, where it saturates to a small $\mathcal{O}(1)$ value.

\subsection{Area law for fermionic non-local magic}
\label{sec:FNL_area_law}
The results presented in the previous subsection are more general. In fact,  Eq.~\eqref{eq:analytic_NL} immediately implies an area law for fermionic non-local magic in ground states of one-dimensional local gapped free-fermion Hamiltonians. We present the argument for arbitrary $\RENYI\ge1$; it is rigorous at $\RENYI=2$ by Theorem~\ref{thm:global}, while for $\RENYI\neq2$ it inherits the status of the conjectured weights of Eq.~\eqref{eq:analytic_NL_alpha}.

Let $|\Psi_O\rangle$ be such a Gaussian ground state and let $A$ be a contiguous interval. The Gaussian entanglement entropy is
\begin{equation}
    S(A)=
    \sum_{i=1}^{\ell}
    \mathfrak{f}\!\left(\frac{1+\lambda_i}{2}\right),
    \label{eq:gaussian_entropy_area_law_sec}
\end{equation}
where $\mathfrak{f}(x):=-x\log_2x-(1-x)\log_2(1-x)$. On the other hand, taking the limit \(\RENYI\to1\) in Eq.~\eqref{eq:analytic_NL_weights} gives
\begin{equation}
    M_{1}^{\mathrm{FNL}}(|\Psi_O\rangle)
    =
    \frac{1}{2}
    \sum_{i=1}^{\ell}
    \mathfrak{f}(\lambda_i^2).
    \label{eq:M1_FNL}
\end{equation}
Moreover, $M_{\alpha}^{\mathrm{FNL}}\le M_{\beta}^{\mathrm{FNL}}$ for $\alpha\ge\beta$, since R\'enyi entropies are non-increasing in the index. By noticing that for every $\lambda\in[0,1)$
\begin{equation}
    \frac{1}{2}\mathfrak{f}(\lambda^2)
    \le
    2\, \mathfrak{f}\!\left(\frac{1+\lambda}{2}\right),
    \label{eq:m1_entropy_bound}
\end{equation}
then, summing over the single-particle entanglement modes gives
\begin{equation}
    M_{\RENYI}^{\mathrm{FNL}}(|\Psi_O\rangle)
    \le
    M_{1}^{\mathrm{FNL}}(|\Psi_O\rangle)
    \le
    2\, S(A),
    \qquad
    \RENYI\ge1.
    \label{eq:FNL_bound_by_entropy_all_alpha}
\end{equation}
This is the Gaussian analogue of the general entanglement-controlled bound discussed for $M_2^{\mathrm{NL}}$ in Ref.~\cite{Cao_2025}.

Finally, by Hastings' theorem~\cite{Hastings_2007}, the entanglement entropy of a unique gapped ground state of a one-dimensional local Hamiltonian obeys
\begin{equation}
    S(A)\le C_{\mathrm{ent}}|\partial A|,
    \label{eq:Hastings_area_law}
\end{equation}
with $C_{\mathrm{ent}}$ independent of $|A|$ and of the total system size. Therefore
\begin{equation}
    M_{\RENYI}^{\mathrm{NL}}(|\Psi_O\rangle)
    \le
    M_{\RENYI}^{\mathrm{FNL}}(|\Psi_O\rangle)
    \le
    2\, C_{\mathrm{ent}}|\partial A|,
    \label{eq:FNL_area_law_final}
\end{equation}
whenever $\RENYI\ge1$. 

\subsection{Quench dynamics and quasiparticle picture}

\begin{figure*}
    \centering
    \includegraphics[width=\linewidth]{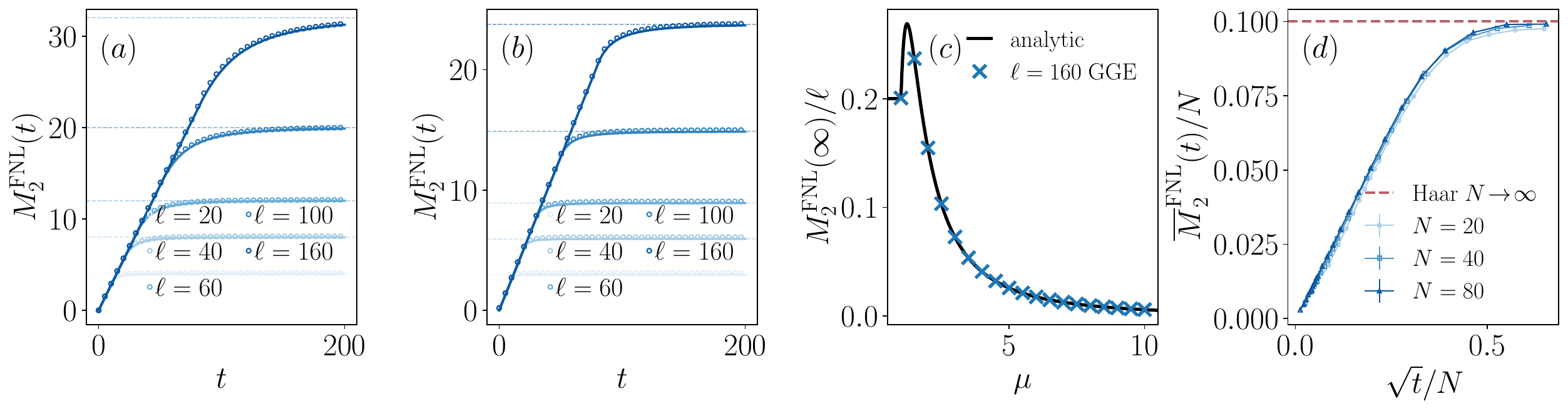}
    \caption{%
    \textbf{Dynamics of fermionic non-local magic.}
    (a)~Quench at $\eta=1$ from the fully polarized state ($\mu_0\to\infty$) to the critical point ($\mu=1$): the exact numerics (circles) agrees quantitatively with the analytical prediction in Eq.~\eqref{eq:fagotti_mnl} (solid lines) for subsystem sizes $\ell = 40, 60, 100, 160$. Dashed lines mark the GGE stationary value.
    (b)~Same for a finite quench $\mu_0 = 2 \to \mu = 1$, showing weaker fermionic non-local magic production due to the smaller Bogoliubov overlap $\sin^2\Delta\theta_k$.
    (c)~Stationary-state density $M_2^{\mathrm{FNL}}(\infty)/\ell$ as a function of the post-quench field $\mu$ ($\mu_0\to\infty$): analytical prediction from Eq.~\eqref{eq:mnl_stationary} (solid) and GGE numerics at $\ell = 160$ (crosses).
    (d)~Random Gaussian circuit dynamics:  $\overline{M_2}^{\mathrm{FNL}}/N$ is plotted against $\sqrt{t}/N$ for $N = 20, 40, 80$, exhibiting diffusive growth $\sim\!\sqrt{t}$ and saturation to the fermionic Gaussian Haar-random value $\overline{\mathsf{m}}_{2}^{\mathrm{FNL}}(r=1/2)=2-\log_2(2+\sqrt{3})$  (dashed line).}
    \label{fig:dynamics}
\end{figure*}

We now turn from equilibrium to the non-equilibrium dynamics of fermionic non-local magic, focusing on quantum quenches in the Ising chain.
The system is prepared in the ground state $|\Psi_0\rangle$ of $H(\mu_0)$ and, at time $t=0$, the transverse field is suddenly quenched to $\mu \neq \mu_0$; the state then evolves unitarily as $|\Psi(t)\rangle = e^{-iH(\mu)t}|\Psi_0\rangle$.
Since both the initial state and the time evolution preserve the Gaussian structure, $|\Psi(t)\rangle$ remains a fermionic Gaussian state at all times, fully characterized by its time-dependent covariance matrix $\Gamma(t)$.
In particular, the time-dependent block-Toeplitz symbol of $\Gamma_A(t)$ takes the form~\cite{Calabrese2005evolution}
\begin{equation}
    (\Gamma_A)_{mn}(t) = \int_0^{2\pi}\!\frac{dk}{2\pi}\,e^{ik(m-n)}\,g(k,t)\;,\label{eq:mustang}
\end{equation}
with time-dependent symbol
\begin{equation}
    g(k,t) = \begin{pmatrix} -q(k,t) & p(k,t) \\ -p(-k,t) & q(k,t) \end{pmatrix},
    \label{eq:symbol_t}
\end{equation}
where we defined the functions
\begin{equation}
\begin{split}
    p(k,t) &= e^{-i\theta_k}\left[\cos(\Delta\theta_k) + i\sin(\Delta\theta_k)\cos (2\varepsilon_k t) \right], \\
    q(k,t) &= i\,\sin(\Delta\theta_k)\,\sin (2\varepsilon_k t)\,.
    \label{eq:pq_symbol}
\end{split}
\end{equation}
Here $\varepsilon_k $ is the post-quench dispersion, $\theta_k$ is the post-quench Bogoliubov angle, and $\Delta\theta_k = \theta_k - \theta_k^{(0)}$ is the mismatch between the post-quench and pre-quench Bogoliubov angles, with
\begin{equation}
    \cos(\Delta\theta_k) = \frac{(\mu-\cos(k))(\mu_0-\cos (k))+\sin^2 (k)}{\varepsilon_k\,\varepsilon_k^{(0)}}\,.
    \label{eq:cos_delta_theta}
\end{equation}
At $t=0$, $q(k,0) = 0$ and $p(k,0) = e^{-i\theta_k^{(0)}}$, recovering the ground-state symbol of $H(\mu_0)$.
The diagonal component, $q(k,t)$, is generated dynamically by the quench and oscillates as $\sin(2\varepsilon_k t)$.

\subsubsection{Stationary state via dephasing}
At long times $t\to\infty$, the oscillating terms $\cos (2\varepsilon_k t)$ and $\sin(2\varepsilon_k t)$ in $p(k,t)$ and $q(k,t)$ average to zero due to dephasing across modes with different group velocities~\cite{Calabrese2005evolution}.
Only the time-independent part of $p$ survives, $p(k,\infty) = e^{-i\theta_k}\cos(\Delta\theta_k)$, while $q(k,\infty) = 0$.
The resulting stationary covariance matrix is that of a generalized Gibbs ensemble (GGE)~\cite{Calabrese2005evolution,vidmar2017entanglement} and is again a block-Toeplitz matrix, now with a \emph{time-independent} symbol.
The eigenvalues $\pm\lambda_k^\infty$ of the stationary $i\Gamma_A^\infty$ are determined by the Bogoliubov angle mismatch alone, $\lambda_k^\infty=|\cos\Delta\theta_k|$, 
 and the stationary fermionic non-local magic density is obtained via the Szeg\H{o} theorem~\cite{Calabrese2005evolution}
\begin{equation}
    M_2^{\mathrm{FNL}}(\infty) = \frac{\ell}{\pi}\int_0^\pi\! dk\; \mathfrak{m}_2\!\left(\sin^2\!\Delta\theta_k\right)\;.
    \label{eq:mnl_stationary}
\end{equation}
This result reveals that for any $\mu\neq \mu_0$, the quench generically produces extensive (volume-law) fermionic non-local magic; the expression is symmetric under $\mu \leftrightarrow \mu_0$.

\subsubsection{Quasiparticle picture}
Eq.~\eqref{eq:mnl_stationary} bears a striking structural similarity to the entanglement entropy of the same quench~\cite{Calabrese2005evolution}, $S=\frac{\ell}{\pi} \int_0^\pi\! dk\; \mathfrak{f}\bigl(\sin^2(\Delta\theta_k/2)\bigr)$, with $\mathfrak{f}(x):= -(1-x)\log_2(1-x)-x\log_2(x)$  the binary entropy: in terms of the stationary eigenvalues $\lambda_k^{\infty}=|\cos\Delta\theta_k|$, the entanglement weight $\mathfrak{f}\bigl((1+\lambda_k^{\infty})/2\bigr)$ is replaced by the magic weight $\mathfrak{m}_2\bigl((\lambda_k^{\infty})^2\bigr)=\mathfrak{m}_2(\sin^2\!\Delta\theta_k)$, using the symmetry $\mathfrak{m}_2(x)=\mathfrak{m}_2(1-x)$.
This parallel extends beyond the stationary state and allows us to derive, from first principles, a quasiparticle picture for the spreading of fermionic non-local magic.

The fermionic non-local magic can be written as a contour integral over the characteristic polynomial $D_\ell(\lambda) = \det(\lambda \mathbb{I}_{2\ell} - i\Gamma_A)$,
\begin{equation}
    M_2^{\mathrm{FNL}} = \frac{1}{4\pi i}\oint_\mathcal{C}d\lambda\;\mathfrak{m}_2(\lambda^2)\;\frac{d}{d\lambda}\ln D_\ell(\lambda)\;,\label{eq:circccc}
\end{equation}
where $\mathcal{C}$ encircles all eigenvalues counterclockwise; the factor $1/2$ compensates the double counting of the $\pm\lambda_n$ pairs.
Expanding $\ln D_\ell(\lambda)$ in powers of $\lambda^{-1}$ expresses $M_2^{\mathrm{FNL}}$ as a series in the even moments $\mathrm{Tr}(i\Gamma_A)^{2n}$.
The $2\ell \times 2\ell$ matrix $i\Gamma_A$ has a block-Toeplitz structure inherited from the symbol $g(k,t)$.
Exploiting this structure, one can show that its $2\ell$ eigenvalues coincide with those of two Hermitian matrices $W_\pm = \mathcal{H} \pm i\mathcal{T}$, where $\mathcal{T}_{jl} = \int \frac{dk}{2\pi}\,e^{-ik(j-l)}\,q(k,t)$ is Toeplitz and $\mathcal{H}_{jl} = \int \frac{dk}{2\pi}\,e^{-ik(j+l-\ell-1)}\,p(k,t)$ is Hankel. This enables the exact time dependence of $M_2^\mathrm{FNL}(t)$ following the strategy in Ref.~\cite{Fagotti2008}. 
The moments $\mathrm{Tr}(i\Gamma_A)^{2n} = \mathrm{Tr}\,W_+^{2n} + \mathrm{Tr}\,W_-^{2n}$ are then computed by expanding the matrix products into $2n$-fold momentum integrals.
As shown in Ref.~\cite{Fagotti2008}, these integrals can be evaluated exactly in the limit $\ell \gg 1$ via a multidimensional stationary phase approximation.
This procedure leads to the asymptotic expression\begin{equation}
\begin{split}
   \!\!\! \mathrm{Tr}\,(i\Gamma_A)^{2n}&\simeq \frac{2\ell}{\pi}\int_0^\pi\!\!\!dk\Bigl[1-\bigl(1-\cos^{2n}(\Delta\theta_k)\bigr)\\
   &\qquad\qquad\qquad\times\min\!\left(1, \frac{2|v_k|t}{\ell}\right)\Bigr],
\end{split}
\label{eq:moments_asymptotic}
\end{equation}\DIFaddend 
where $v_k = d\varepsilon_k/dk$ is the group velocity.

At $t=0$ the bracket equals one and all moments take the pure-state value $2\ell$, while the term proportional to $\min(1,2|v_k|t/\ell)$ interpolates ballistically towards  the stationary (GGE) moments $\frac{2\ell}{\pi}\int_0^\pi\cos^{2n}(\Delta\theta_k)\,dk$, capturing the propagation of quasiparticle pairs.
Inserting Eq.~\eqref{eq:moments_asymptotic} into  Eq.~\eqref{eq:circccc} and carrying out the contour integration \DIFaddend yields the quasiparticle formula
\begin{equation}
    M_2^{\mathrm{FNL}}(\ell,t) \simeq \int_0^\pi\!\frac{dk}{\pi}\,\min\!\bigl(2|v_k|\,t,\;\ell\bigr)\;\mathfrak{m}_2\!\bigl(\sin^2(\Delta\theta_k)\bigr),
    \label{eq:fagotti_mnl}
\end{equation}
valid in the scaling limit $\ell, t \gg 1$ with $t/\ell$ arbitrary.
Eq.~\eqref{eq:fagotti_mnl} predicts linear growth $\MNL \propto t$ for $t < t^* = \ell/(2v_{\max})$ and saturation to the extensive value $\MNL(\infty)/\ell = \pi^{-1}\int_0^\pi \mathfrak{m}_2(\sin^2\!\Delta\theta_k)\, dk$ for $t \gg t^*$.
Eq.~\eqref{eq:fagotti_mnl} has a transparent physical interpretation in terms of the quasiparticle picture~\cite{Calabrese2005evolution,Fagotti2008}: the quench creates pairs of entangled quasiparticles at each point, which propagate ballistically with mode-dependent velocities $\pm v_k$.
A pair at momentum $k$ contributes to $M_2^{\mathrm{FNL}}$ of the block $A$ only when one partner has entered $A$ while the other remains outside, hence the $\min(2|v_k|t,\ell)$ factor.
The weight $\mathfrak{m}_2(\sin^2\Delta\theta_k)$  quantifies the fermionic non-local magic carried by each entangled pair: it vanishes both when $\Delta\theta_k = 0$ (no quench, no entanglement production) and when $\Delta\theta_k = \pi/2$ (maximally entangled but stabilizer-like mode), and is maximal at intermediate values.

We benchmark Eq.~\eqref{eq:fagotti_mnl} at $\eta=1$ (Ising chain) by  evaluating numerically Eq.~\eqref{eq:mustang}. 
Our results are reported in Fig.~\ref{fig:dynamics} for the quench $\mu_0 = \infty \to \mu = 1$ (panel (a))  and $\mu_0 = 2 \to \mu = 1$ (panel (b)) at subsystem sizes $\ell = 40, 60, 100, 160$. 
The exact time evolution (markers), obtained by constructing the full block-Toeplitz $\Gamma_A(t)$ via Eq.~\eqref{eq:pq_symbol} and diagonalizing, is in excellent agreement with the semiclassical prediction (solid lines) from Eq.~\eqref{eq:fagotti_mnl}.
The dashed horizontal lines indicate the stationary values from Eq.~\eqref{eq:mnl_stationary}, which are approached at $t \simeq t^\star$.
In Fig.~\ref{fig:dynamics}(c) we benchmark the steady state $M_2^{\mathrm{FNL}}(\infty)$ with the analytical expression Eq.~\eqref{eq:mnl_stationary} for the GGE, again in perfect agreement. 

\subsection{Random Gaussian circuit dynamics}
Finally, we consider the dynamics of fermionic non-local magic under random Gaussian circuits~\cite{tirrito2025magicphasetransitionsmonitored}.
The circuit is composed of nearest-neighbor fermionic Gaussian unitaries arranged in a brick-wall geometry on $N$ sites.
A single layer $O_{\mathrm{layer}} = O_{\mathrm{odd}}\,O_{\mathrm{even}}$ consists of an even and an odd half-step, namely $O_{\mathrm{even}} = \prod_{j=1}^{N/2} R_{2j-1,2j}$ and $O_{\mathrm{odd}} = \prod_{j=1}^{N/2-1} R_{2j,2j+1}$, where each $R_{j,j+1} \in \mathrm{SO}(4)$ is a random orthogonal matrix drawn independently from the Haar measure, acting on the four Majorana operators $(\gamma_{2j-1}, \gamma_{2j}, \gamma_{2j+1}, \gamma_{2j+2})$ of sites $j$ and $j+1$.
The full circuit after $t$ layers is $O(t) = \prod_{s=1}^t O_{\mathrm{layer}}^{(s)} \in \mathrm{SO}(2N)$.
The initial state is the product state $|\mathbf{0}\rangle$ (vacuum of all fermionic modes), with covariance matrix $\Gamma_0$. 

The time-evolved state remains Gaussian at all times, with covariance matrix $\Gamma(t) = O(t)\,\Gamma_0\,O(t)^T$.
The fermionic non-local magic is therefore entirely determined by the restricted covariance matrix $\Gamma_A(t)$ and its eigenvalues $\pm\lambda_j(t)$.

In contrast to the ballistic quasiparticle spreading found in the integrable quench setting, operator spreading in local random free-fermion circuits is \emph{diffusive}~\cite{dias2021diffusiveoperatorspreadingrandom,swann}.
Since the fermionic non-local magic is a smooth spectral functional of $\Gamma_A(t)$, it inherits this diffusive transport.
At early times, expanding the single-mode contribution for small $\lambda$, $\mathfrak{m}_2(\lambda^2) \simeq \frac{\lambda^2}{\ln 2} + \mathcal{O}(\lambda^4)$, so that $M_2^{\mathrm{FNL}}(t) \simeq \frac{1}{2\ln 2}\mathrm{Tr}[(i\Gamma_A(t))^2]$, which is governed by the diffusive spreading of two-point correlations across the entanglement cut~\cite{dias2021diffusiveoperatorspreadingrandom}.
This argument yields
\begin{equation}
    M_2^{\mathrm{FNL}}(t) \propto \sqrt{t}\;, \qquad 1 \ll t \ll N^2\;.
\end{equation}
At late times, $M_2^{\mathrm{FNL}}$ saturates to the Haar-random fermionic Gaussian value, and the natural scaling variable is $\sqrt{t}/N$.
In Fig.~\ref{fig:dynamics}(d), we plot $\overline{M_2}^{\mathrm{FNL}}(t)/N$ versus $\sqrt{t}/N$ for $N = 20, 40, 80$, finding an excellent data collapse onto a single scaling function that interpolates between $\sqrt{t}$ growth and saturation at the Page-curve value $\overline{\mathsf{m}}_{2}^{\mathrm{FNL}}(r=1/2)=2-\log_2(2+\sqrt{3})$  (dashed line).
The saturation timescale $t\sim\mathcal{O}(N^2)$ is consistent with diffusive dynamics.

\section{Discussion}

\subsection{Summary}
We have proved an exact, closed-form expression for the fermionic non-local magic of Gaussian states at R\'enyi index $\RENYI=2$. A generalized expression for $\RENYI>2$ is verified numerically and shares the same qualitative features.
By reducing the minimization over Gaussian unitaries, the computation simplifies to a single eigenvalue decomposition of the reduced covariance matrix, completely bypassing the prohibitive variational minimization over the full Hilbert space. This establishes fermionic Gaussian states as a broad class of highly entangled many-body systems where the related non-local magic is exactly tractable, providing a useful proxy for the non-local magic over the full Unitary group.
Beyond the physical result, the proof introduces two ingredients that we expect to be of independent interest: a sharp bound controlling the collision probability of a single Slater determinant by its transpose overlap; and a Hadamard-type determinant inequality for positive matrices of rank at most two, proved using fermionic creation operators.

We applied this framework across three key physical settings.
For typical fermionic Gaussian states, we derived an exact Page-like curve via random matrix theory, a prediction independently confirmed by SYK$_2$ mid-spectrum eigenstates.
At equilibrium, we characterized the fermionic non-local magic of the XY spin chain ground state, obtaining exact results in both gapped and critical phases: a finite saturation value off criticality via corner transfer matrices, and logarithmic scaling at the critical point via the Fisher--Hartwig formula.
Out of equilibrium, we established a quasiparticle picture for the spreading of fermionic non-local magic after quantum quenches, finding linear growth up to a crossing time $t^\star= \ell/(2v_\mathrm{max})$  followed by saturation to a volume-law value.
In contrast, random Gaussian circuits exhibit diffusive growth, with non-local magic scaling as $\sqrt{t}$ before saturating at the Haar-random value on a timescale $t\sim N^2$.

\subsection{Outlook}
The analytical tractability established here opens several natural directions for future investigation.

\emph{Higher dimensional systems.}
While our results are formulated for arbitrary bipartitions of one-dimensional chains, fermionic Gaussian states in higher dimensions $D>1$ offer a rich arena where entanglement and magic may exhibit qualitatively different behavior. In two and three dimensions, a richer structure tied to the geometry of the Fermi surface is expected, and the entanglement entropy obeys an area law with multiplicative logarithmic corrections $\mathcal{O}(\ell^{D-1}\ln\ell)$~\cite{swingle2010entanglement}. 
Investigating how the scaling rate $\beta_\mathrm{FNL}$ emerges in such systems is a natural generalization.

\emph{Conformal field theory and integrable systems.}
A key question is to determine whether $\beta_{\mathrm{FNL}}$ captures universal features of the underlying conformal field theory, much as the central charge governs the scaling of entanglement entropy~\cite{Calabrese2004entanglement}.
Responding to this question requires extending our analysis to interacting integrable Hamiltonians, where Bethe-ansatz methods provide access to the entanglement spectrum.
On the dynamical side, it remains an open question whether the quasiparticle picture that successfully describes the spreading of entanglement after a quantum quench~\cite{Alba_2017,alba2019,Kudler_Flam_2021} can be adapted to account for the growth of non-local magic in these models.

\emph{Disordered systems.}
Our formula provides a natural framework to study non-local magic in the presence of quenched disorder. 
For example, in Anderson localization models~\cite{evers,sierant2020}, we expect deeply localized eigenstates to have a trivial structure, for which the fermionic non-local magic density $\mathfrak{m}_\RENYI(x)$ vanishes. 
Similarly, free fermionic models of strongly disordered systems~\cite{IGLOI_2005,Monthus_2018,ruggiero2022}, where rare Griffiths regions and anomalous entanglement scaling emerge, are another avenue of further study.

\emph{Topological systems.}
One and two-dimensional topological phases, such as the Kitaev chain and model~\cite{Kitaev2001chain,Kitaev_2006}, host boundary modes whose non-local magic content is expected to reflect the topological invariant.
In such systems, the interplay of non-local magic and quantum coherent errors in fermionic systems is another fascinating direction to pursue~\cite{beri,gskb-t5ql}.

\emph{Interplay with symmetry.}
When the system possesses a continuous symmetry, such as $\mathcal{U}(1)$ charge conservation, the covariance matrix acquires a block structure that could enable a symmetry-resolved decomposition of $M_2^{\mathrm{FNL}}$ into contributions from different charge sectors~\cite{Goldstein}.
This connects to the broader program of symmetry-resolved entanglement~\cite{Bianchi_2022}, coherence~\cite{aditya2026coherencedynamicsquantummanybody} and stabilizer entropies, and may reveal sector-dependent magic resources~\cite{iannotti2026nonstabilizernessu1symmetrychaotic, cepollaro2025stabilizerentropysubspaces, Esposito_2024}.
An intriguing related question is whether fermionic non-local magic exhibits quantum Mpemba-like effects~\cite{Ares_2023,Ares_2025,aditya2025mpembaeffectsquantumcomplexity}, where symmetry-broken initial conditions lead to anomalously fast relaxation of specific magic sectors.

\emph{Monitored systems and non-Hermitian dynamics.}
Hybrid systems combining unitary evolution with local measurements give rise to measurement-induced phase transitions in entanglement~\cite{FisherRandom2023,Li_2025,oshima}.
For free-fermionic setups, the dynamics remain within the Gaussian manifold, and our formula applies at each quantum trajectory.
Whether fermionic non-local magic exhibits its own measurement-induced transition, and how it relates to the entanglement transition, is an open question.
Similarly, non-Hermitian free fermionic Hamiltonians~\cite{Ashida_2020,kawabata2026nonhermitiandisorderedsystems} generate effective Gaussian dynamics with complex spectra.
Our framework makes amenable to analytic treatment the behavior of fermionic non-local magic in these settings~\cite{schirooooo,gallo,granet,lapierre2025entanglementtransitionsstructuredrandom}.

\begin{acknowledgments}
This work is dedicated to Ada. We are particularly indebted to P. Sierant for spotting a mistake in a prior version of this manuscript. 
We thank L. Leone, A. Russotto, P. Calabrese, L. Piroli, and J. Odavić for comments and collaborations on related topics.

\textbf{Fundings}
D.I. and B.M. would like to acknowledge the Les Houches Summer School 2025 on Exact Solvability and Quantum Information, which facilitated work on this subject.
D.I. acknowledges support from COST Action THEORY-CHALLENGES CA22113 under Short-Term Scientific Mission (STSM) grant E-COST-GRANT-CA22113-6901e967.  
B.M. and X.T. acknowledge support from DFG Emmy Noether Programme proposal "Digital Quantum Matter Out-of-Equilibrium" No. 560726973, DFG under Germany's Excellence Strategy – Cluster of Excellence Matter and Light for Quantum Computing (ML4Q) EXC 2004/2 – 390534769, and DFG Collaborative Research Center (CRC) 183 Project No. 277101999 - project B01. 
A.H. acknowledges support from  PNRR MUR project PE0000023-NQSTI and the PNRR MUR project CN00000013-ICSC.

\textbf{Data Availability.} The data generated in this study will be publicly shared at publication in Zenodo.



\textbf{Note Added.}
While preparing this manuscript, we discovered a closely related independent work by M. Collura, B. Beri, and E. Tirrito, appearing in the same arXiv posting~\cite{collura2026nonlocalnonstabilizernessfreefermion}. 
When overlapping, our results are consistent.
\end{acknowledgments}
\clearpage

\section{Appendix}
\label{sec:Appendix}

\subsection{Proof of Theorem~\ref{thm:global}: An Odyssey 
}
\label{app:proof}

\begin{figure*}[ht!]
\centering
\definecolor{oceanTop}{HTML}{EAF5F4}
\definecolor{oceanBottom}{HTML}{CDE5E7}
\definecolor{oceanEdge}{HTML}{78A7AB}
\definecolor{inkm}{HTML}{173E47}
\definecolor{mutedm}{HTML}{59757A}
\definecolor{sandm}{HTML}{F3DFAD}
\definecolor{sandLight}{HTML}{FFF1CE}
\definecolor{coastm}{HTML}{9A7138}
\definecolor{routem}{HTML}{A84A42}
\definecolor{accentm}{HTML}{187E79}
\begin{tikzpicture}[x=1cm,y=1cm,
  route/.style={draw=routem, line width=1.05pt, dash pattern=on 3.2pt off 2.6pt, line cap=round,
                -{Stealth[length=2.5mm,width=1.9mm]}},
  island/.style={draw=coastm, line width=0.85pt, top color=sandLight, bottom color=sandm},
  step badge/.style={circle, fill=inkm, draw=white, line width=0.6pt, minimum size=4.4mm,
                     inner sep=0pt, text=white, font=\fontsize{6.6}{6.6}\selectfont\bfseries},
  lab title/.style={text=inkm, font=\fontsize{8}{9}\selectfont\bfseries, align=center, inner sep=1pt},
  lab formula/.style={text=inkm, font=\fontsize{7.3}{8.2}\selectfont, align=center, inner sep=1pt},
  lab ref/.style={text=accentm, font=\fontsize{5.6}{6.4}\selectfont\bfseries, align=center, inner sep=1pt},
  small label/.style={text=mutedm, font=\fontsize{5.6}{6.3}\selectfont\bfseries}]
\path[rounded corners=10pt, top color=oceanTop, bottom color=oceanBottom,
      draw=oceanEdge, line width=0.8pt] (0,0) rectangle (17,5.2);
\draw[route] (0.5,2.35) to[out=50,in=230] (0.86,2.73);
\draw[route] (1.35,3.25) to[out=-30,in=150] (2.59,2.44);
\draw[route] (3.35,2.0) to[out=30,in=210] (4.69,2.76);
\draw[route] (5.45,3.2) to[out=-30,in=150] (6.54,2.49);
\draw[route] (7.3,2.05) to[out=25,in=205] (9.25,2.83);
\draw[route] (10.05,3.2) to[out=-30,in=150] (11.29,2.42);
\draw[route] (12.05,1.98) to[out=30,in=210] (13.29,2.81);
\draw[route] (14.05,3.25) to[out=-25,in=160] (15.21,2.92);
\foreach \wx/\wy/\ws in {2.35/4.42/0.85, 4.6/4.05/0.7, 11.6/4.42/0.8, 3.0/1.25/0.75,
                         10.5/1.75/0.7, 13.3/1.3/0.8, 5.9/1.7/0.65, 15.9/4.35/0.7}{
  \begin{scope}[shift={(\wx,\wy)},scale=\ws]
    \draw[white, draw opacity=0.78, line width=0.55pt, line cap=round]
      (-0.31,0) .. controls (-0.18,0.08) and (-0.08,0.08) .. (0.03,0)
      .. controls (0.14,-0.08) and (0.23,-0.08) .. (0.34,0);
    \draw[white, draw opacity=0.52, line width=0.5pt, line cap=round]
      (-0.22,-0.14) .. controls (-0.10,-0.07) and (0.00,-0.07) .. (0.10,-0.14)
      .. controls (0.19,-0.20) and (0.27,-0.20) .. (0.34,-0.15);
  \end{scope}}
\begin{scope}[shift={(0.82,0.88)},scale=0.5]
  \draw[inkm, draw opacity=0.55, line width=0.65pt] (0,0) circle (0.42);
  \fill[inkm, fill opacity=0.78] (0,0.72) -- (0.11,0.06) -- (0,0.16) -- (-0.11,0.06) -- cycle;
  \fill[inkm, fill opacity=0.48] (0,-0.72) -- (0.11,-0.06) -- (0,-0.16) -- (-0.11,-0.06) -- cycle;
  \fill[accentm, fill opacity=0.72] (0.72,0) -- (0.06,0.11) -- (0.16,0) -- (0.06,-0.11) -- cycle;
  \fill[accentm, fill opacity=0.44] (-0.72,0) -- (-0.06,0.11) -- (-0.16,0) -- (-0.06,-0.11) -- cycle;
  \node[text=inkm, font=\fontsize{4.6}{4.6}\selectfont\bfseries] at (0,0.95) {N};
\end{scope}
\begin{scope}[shift={(9.6,0.95)},scale=0.6,line cap=round,line join=round]
  \draw[white, draw opacity=0.78, line width=0.55pt]
    (-0.86,-0.51) .. controls (-0.54,-0.42) and (-0.25,-0.44) .. (0.04,-0.52)
    .. controls (0.29,-0.59) and (0.59,-0.57) .. (0.84,-0.48);
  \draw[oceanEdge, draw opacity=0.55, line width=0.45pt]
    (-0.66,-0.65) .. controls (-0.37,-0.58) and (-0.17,-0.59) .. (0.09,-0.66)
    .. controls (0.30,-0.72) and (0.48,-0.70) .. (0.66,-0.65);
  \draw[inkm, draw opacity=0.58, line width=0.38pt]
    (-0.04,0.79) -- (-0.59,-0.18) (-0.04,0.79) -- (0.65,-0.18);
  \draw[inkm, line width=0.78pt] (-0.02,-0.24) -- (-0.04,0.82);
  \path[fill=sandLight, fill opacity=0.94, draw=inkm, line width=0.52pt]
    (0.02,0.71) .. controls (0.28,0.60) and (0.52,0.27) .. (0.59,-0.11)
    .. controls (0.36,-0.06) and (0.18,-0.07) .. (0.02,-0.12) -- cycle;
  \path[fill=accentm!68!white, fill opacity=0.94, draw=inkm, line width=0.52pt]
    (-0.09,0.58) .. controls (-0.34,0.40) and (-0.50,0.13) .. (-0.54,-0.12)
    .. controls (-0.36,-0.08) and (-0.21,-0.08) .. (-0.09,-0.12) -- cycle;
  \draw[white, draw opacity=0.72, line width=0.35pt]
    (0.09,0.54) .. controls (0.29,0.43) and (0.42,0.20) .. (0.47,-0.02);
  \draw[coastm, line width=0.52pt] (-0.07,-0.14) -- (0.61,-0.14);
  \path[fill=routem] (-0.04,0.80) -- (0.24,0.72) -- (-0.04,0.64) -- cycle;
  \path[fill=inkm, draw=inkm, line width=0.55pt]
    (-0.67,-0.22) -- (0.69,-0.22)
    .. controls (0.47,-0.54) and (0.15,-0.62) .. (-0.18,-0.55)
    .. controls (-0.40,-0.50) and (-0.56,-0.39) .. (-0.67,-0.22) -- cycle;
  \draw[sandLight, line width=0.62pt] (-0.57,-0.25) -- (0.58,-0.25);
  \fill[sandLight] (-0.25,-0.38) circle (0.035);
  \fill[sandLight] (0.02,-0.41) circle (0.035);
  \fill[sandLight] (0.29,-0.38) circle (0.035);
\end{scope}
\foreach \px/\py/\sc/\num/\coast in {1.35/3.25/0.42/1/1, 3.35/2.0/0.42/2/2, 5.45/3.2/0.42/3/3,
                                     7.3/2.05/0.42/4/4, 10.05/3.2/0.42/5/1, 12.05/1.98/0.42/6/2,
                                     14.05/3.25/0.42/7/3, 16.15/2.58/0.5/8/4}{
  \begin{scope}[shift={(\px,\py)}]
    \begin{scope}[scale=\sc]
      \ifnum\coast=1
        \path[island] plot[smooth cycle,tension=0.72] coordinates {
          (-1.67,0.05) (-1.49,0.57) (-0.83,0.83) (0.02,0.88) (0.86,0.77) (1.50,0.45)
          (1.67,-0.05) (1.47,-0.55) (0.70,-0.79) (-0.19,-0.83) (-1.08,-0.69) (-1.59,-0.34)};
      \fi
      \ifnum\coast=2
        \path[island] plot[smooth cycle,tension=0.72] coordinates {
          (-1.66,0.18) (-1.34,0.64) (-0.58,0.84) (0.32,0.85) (1.13,0.67) (1.61,0.25)
          (1.63,-0.27) (1.16,-0.68) (0.35,-0.82) (-0.57,-0.78) (-1.35,-0.53) (-1.68,-0.12)};
      \fi
      \ifnum\coast=3
        \path[island] plot[smooth cycle,tension=0.72] coordinates {
          (-1.65,0.00) (-1.46,0.48) (-0.92,0.76) (-0.11,0.89) (0.74,0.83) (1.45,0.55)
          (1.67,0.08) (1.52,-0.44) (0.86,-0.73) (0.02,-0.84) (-0.84,-0.76) (-1.49,-0.48)};
      \fi
      \ifnum\coast=4
        \path[island] plot[smooth cycle,tension=0.72] coordinates {
          (-1.67,0.11) (-1.30,0.61) (-0.54,0.84) (0.36,0.82) (1.17,0.71) (1.59,0.30)
          (1.64,-0.20) (1.31,-0.61) (0.56,-0.81) (-0.34,-0.82) (-1.16,-0.61) (-1.61,-0.29)};
      \fi
      \path[draw=white, draw opacity=0.45, line width=0.45pt]
        plot[smooth cycle,tension=0.72] coordinates {
          (-1.55,0.08) (-1.28,0.50) (-0.61,0.71) (0.21,0.73) (0.95,0.61) (1.38,0.28)
          (1.42,-0.17) (1.08,-0.54) (0.35,-0.68) (-0.47,-0.66) (-1.18,-0.49) (-1.50,-0.20)};
    \end{scope}
    \node[step badge] at (0,0) {\num};
  \end{scope}}
\fill[accentm] (0.5,2.35) circle (2.2pt);
\fill[routem] (16.62,2.44) circle (2.0pt);
\foreach \ax/\ay/\bx/\by in {1.35/3.68/1.45/4.1, 5.45/3.63/5.55/4.1, 10.05/3.63/10.1/4.1, 14.05/3.68/14.1/4.1,
                             3.35/1.55/3.35/1.18, 7.3/1.6/7.2/1.18, 12.05/1.53/12.1/1.18, 16.15/2.05/15.6/1.18}
  \draw[inkm!45, densely dotted, semithick] (\ax,\ay) -- (\bx,\by);
\node[lab title] at (1.5,4.72) {collision probability};
\node[lab formula] at (1.5,4.44) {$\zeta_2=d\sum_S p_S^2$};
\node[lab ref] at (1.5,4.2) {LEMMA~\ref{lem:DPP}};
\node[lab title] at (3.35,0.9) {the Slater bound};
\node[lab formula] at (3.35,0.62) {$\mathcal{C}(Q)\le 2^{-k}\mathcal{D}(\Sigma)$};
\node[lab ref] at (3.35,0.36) {LEMMA~\ref{lem:Slater-bound}};
\node[lab title] at (5.6,4.72) {BCS parametrization};
\node[lab formula] at (5.6,4.44) {$\zeta_2\propto\mathsf{F}(Z)$,\; $O_A=O_B$};
\node[lab ref] at (5.6,4.2) {LEMMAS~\ref{lem:all-minors}--\ref{lem:two-to-one}};
\node[lab title] at (7.2,0.9) {principal-minor residual};
\node[lab formula] at (7.2,0.62) {$\mathsf{F}(Z)\le\mathsf{R}(H)$,\; $H=Z^2$};
\node[lab ref] at (7.2,0.36) {PROP.~\ref{prop:residual}};
\node[lab title] at (10.1,4.72) {two sign averages};
\node[lab formula] at (10.1,4.44) {$\mathsf{R}(H)=\Delta_\Omega^{-2}\sum_S\det(\mathcal{A}_S)^2$};
\node[lab ref] at (10.1,4.2) {LEMMA~\ref{lem:PMT}, PROP.~\ref{prop:Cayley}};
\node[lab title] at (12.1,0.9) {down to one unitary};
\node[lab formula] at (12.1,0.62) {$\mathsf{R}(H)\propto\mathsf{L}(U)$};
\node[lab title] at (14.1,4.72) {the rank-two bound};
\node[lab formula] at (14.1,4.44) {$\mathsf{L}(U)^2\le2^{2\ell}\det(\mathbb{I}+T^\dagger T)$};
\node[lab ref] at (14.1,4.2) {THMS.~\ref{thm:MD2}--\ref{thm:UL}};
\node[lab title] at (15.45,0.9) {optimum at $Z_0$};
\node[lab formula] at (15.45,0.62) {equality $\Rightarrow$ Thm.~\ref{thm:global}};
\end{tikzpicture}
\caption{\textbf{The route of the proof of Theorem~\ref{thm:global}: an odyssey in eight steps.} The dashed route visits eight islands, one per step of the argument and one per subsection of this Appendix; each island carries its step number, and the matching label at the border states what the step does and its key relation. The essence of each step is spelled out in the text below.}
\label{fig:odyssey}
\end{figure*}

In this Appendix, we prove Theorem~\ref{thm:global}, i.e.\ the closed form of Eq.~\eqref{eq:analytic_NL} for the fermionic non-local magic at $\RENYI=2$. The proof is a chain of eight steps, cf. Fig.~\ref{fig:odyssey}, one per subsection below.
(1)~\emph{Collision probability} (Sec.~\ref{app:proof:dpp}):  Wick's theorem turns the stabilizer purity $\zeta_2$ into the collision probability of a classical distribution $p_S$  built from the covariance matrix, $\zeta_2=d\sum_Sp_S^2$ (Lemma~\ref{lem:DPP})~\cite{collura2026nonstabilizernessfermionicgaussianstates}.
(2)~\emph{The Slater bound}  (Sec.~\ref{app:proof:slater}): we bound the collision probability of a single Slater determinant by its transpose overlap, $\mathcal{C}(Q)\le2^{-k}\mathcal{D}(\Sigma)$ (Lemma~\ref{lem:Slater-bound}) --- the first key ingredient, used in Step~4.
(3)~\emph{BCS parametrization} (Sec.~\ref{app:proof:thouless}): all the dependence of $\zeta_2$ on the local rotations is packaged into the single matrix $Z=O_AZ_0O_B^{\trans}$, and $\zeta_2$ becomes, up to a constant along the orbit, the sum $\mathsf{F}(Z)$ of the fourth powers of the moduli of all the minors of $Z$ (Lemma~\ref{lem:all-minors})~\cite{Mbeng2024,surace}; a Cauchy--Schwarz argument then reduces the two independent rotations to a single one, $O_A=O_B$  (Lemma~\ref{lem:two-to-one}).
(4)~\emph{Principal-minor residual} (Sec.~\ref{app:proof:residual}): the Slater bound of Step~2  converts the all-minors functional into the residual  $\mathsf{R}(H)$ of $H=Z^2$, $\mathsf{F}(Z)\le\mathsf{R}(H)$, with equality at the canonical point (Proposition~\ref{prop:residual}).
(5)~\emph{Two sign averages} (Sec.~\ref{app:proof:cayley}): averages over independent random signs, combined with a Cayley transform, rewrite the residual exactly as $\mathsf{R}(H)=\Delta_\Omega^{-2}\sum_S\det(\mathcal{A}_S)^2$, a sum of squared principal minors of a single Hermitian matrix $\mathcal{A}$ (Lemma~\ref{lem:PMT}, Proposition~\ref{prop:Cayley}).
(6)~\emph{Down to one unitary} (Sec.~\ref{app:proof:graph}): interpreting this sum as a probability, a change of basis within pairs of coordinates collapses it to the functional $\mathsf{L}(U)=\sum_S|\det U_S|^2$ of a single unitary matrix, times a prefactor constant along the orbit, see also Refs.~\cite{collura2026nonstabilizernessfermionicgaussianstates,sfeo}.
(7)~\emph{The rank-two bound} (Sec.~\ref{app:proof:MD2}): a determinant inequality for positive matrices of rank at most two --- the second key ingredient --- bounds $\mathsf{L}(U)$ through the diagonal of $U$, hence by the orbit invariant $\det(\mathbb{I}_{2\ell}+T^\dagger T)$ with $T=(U+U^{\trans})/2$ (Theorems~\ref{thm:MD2}--\ref{thm:UL}).
(8)~\emph{Conclusion} (Sec.~\ref{app:proof:completion}): the canonical BCS form saturates every inequality of the chain; a continuity argument covers the endpoints $\lambda_j=1$, and padding with ancillas removes the balanced-cut assumption, proving Theorem~\ref{thm:global}.

\subsubsection{The stabilizer purity as a collision probability}
\label{app:proof:dpp}

We use the conventions of the main text, recalled here for convenience: $|\Psi\rangle$ is a pure fermionic Gaussian state on $N$ modes, $d=2^N$, $A$ is without loss of generality the smaller subsystem, $\ell\le N/2$,  and $\Gamma_{mn}=i\langle\Psi|[\gamma_m,\gamma_n]|\Psi\rangle/2$, with purity condition equivalent to
\begin{equation}
    \Gamma^{\trans}=-\Gamma\;,\qquad \Gamma^2=-\mathbb{I}_{2N}\;.
    \label{eq:pure-covariance}
\end{equation}
Pauli strings are in one-to-one correspondence, up to irrelevant phases, with Majorana monomials $\gamma_I=\prod_{m\in I}\gamma_m$, $I\subseteq[2N]$ (we denote throughout $[m]:=\{1,\dots,m\}$). Wick's theorem gives $|\langle\Psi|\gamma_I|\Psi\rangle|^2=\mathrm{det}(\Gamma_I)$ for even $|I|$ and zero otherwise, where $\Gamma_I$ is the principal submatrix of $\Gamma$ on $I$. Since $\det(\Gamma_I)=\mathrm{Pf}(\Gamma_I)^2$, the stabilizer purity is the \emph{Wick determinant sum}
\begin{equation}
    \zeta_2(\Gamma)=d^{-1}\sum_{I\subseteq[2N]}\det(\Gamma_I)^2\;.
    \label{eq:zeta-M}
\end{equation}
It is convenient to introduce
\begin{equation}
    q_I(\Gamma)=d^{-1}\det\Gamma_I=d^{-1}\mathrm{Pf}(\Gamma_I)^2\ge 0\;.
    \label{eq:q-definition}
\end{equation}
Odd-dimensional principal minors vanish, and the principal-minor identity $\sum_I\det\Gamma_I=\det(\mathbb{I}+\Gamma)=d$ shows that $q(\Gamma)$ is a probability distribution, with
\begin{equation}
    \zeta_2(\Gamma)=d\sum_{I\subseteq[2N]} q_I(\Gamma)^2\;.
    \label{eq:q-collision}
\end{equation}
Local Gaussian unitaries act on covariance matrices as
\begin{equation}
    \Gamma\longmapsto (O_A\oplus O_B)\,\Gamma\,(O_A\oplus O_B)^{\trans}\;,
    \label{eq:local-action}
\end{equation}
with $O_A\in\mathrm{O}(2\ell)$ and $O_B\in\mathrm{O}(2(N-\ell))$, so that maximizing $\zeta_2$ over the local Gaussian orbit is exactly the problem defined by Eq.~\eqref{eq:defFNLmagic}.
 We write $\mathrm{spec}(i\Gamma_A)=\{\pm\lambda_1,\dots,\pm\lambda_\ell\}$ with $0\le\lambda_j\le 1$ and set $\nu_j=\sqrt{1-\lambda_j^2}$.
For completeness, we recall the covariance version of the modewise normal form of Eq.~\eqref{eq:can_form}, cf.~\cite{BoteroReznik2004}, which also fixes the orientation conventions used below. 

\begin{proposition}[Local covariance normal form]
\label{prop:covariance-normal-form}
There exist local real orthogonal bases in which $\Gamma$ is the direct sum, up to the ordering that places all $A$ coordinates before all $B$ coordinates, of the four-dimensional blocks
\begin{equation}
    \Gamma^{\mathrm{can}}_j=\begin{pmatrix}\lambda_jJ& \nu_j\mathbb{I}_2\\ -\nu_j\mathbb{I}_2&-\lambda_jJ\end{pmatrix},
    \quad J=\begin{pmatrix}0&1\\-1&0\end{pmatrix},
    \label{eq:canonical-covariance-block}
\end{equation}
together with $N-2\ell$ pure $2\times 2$ blocks on $B$.
\end{proposition}

\begin{proof}
For self-containment we briefly recall the proof of the result, cf. Ref.~\cite{BoteroReznik2004} for details.
Write $\Gamma$ in local blocks as $\Gamma=\left(\begin{smallmatrix} \Gamma_{AA}&\Gamma_{AB}\\ -\Gamma_{AB}^{\trans}&\Gamma_{BB}\end{smallmatrix}\right)$. Purity, Eq.~\eqref{eq:pure-covariance}, is equivalent to
\begin{equation}\begin{split}\Gamma_{AB}\Gamma_{AB}^{\trans}&=\mathbb{I}_{2\ell}+\Gamma_{AA}^2,\\ \Gamma_{AB}^{\trans}\Gamma_{AB}&=\mathbb{I}_{2(N-\ell)}+\Gamma_{BB}^2,\\ \Gamma_{AA}\Gamma_{AB}+\Gamma_{AB}\Gamma_{BB}&=0\;.
    \label{eq:purity-block-relations}
\end{split}\end{equation}

A real orthogonal change of basis puts the antisymmetric matrix $\Gamma_{AA}$  into $2\times 2$ blocks $\lambda_jJ$, with $\lambda_j\in[0,1]$, the interval following from the first identity. On the corresponding two-dimensional block $\Gamma_{AB}\Gamma_{AB}^{\trans}=\nu_j^2\mathbb{I}_2$. If $\nu_j>0$, the map $\Gamma_{AB}^{\trans}/\nu_j$ is an isometry from that block into $B$; choosing an orthonormal basis of its image makes the corresponding block of $\Gamma_{AB}$ equal to $\nu_j\mathbb{I}_2$, and the last identity in Eq.~\eqref{eq:purity-block-relations} forces the restriction of $\Gamma_{BB}$ to that image to be $-\lambda_jJ$. This construction works simultaneously on degenerate spectral subspaces by choosing the polar isometry first and a compatible orthonormal basis afterwards. If $\nu_j=0$, the associated rows of $\Gamma_{AB}$ vanish, so the corresponding $A$ mode is already pure and decoupled; the orthogonal complement of the image of $\Gamma_{AB}^{\trans}$  in $B$ is invariant under $\Gamma_{BB}$ and obeys $\Gamma_{BB}^2=-\mathbb{I}$, hence decomposes into pure $\pm J$ blocks. Pairing as many of these as needed with the decoupled pure $A$ blocks gives Eq.~\eqref{eq:canonical-covariance-block} at $\lambda_j=1$, the remaining blocks being the $B$ spectators.
\end{proof}

\begin{remark}
The modewise decomposition in Ref.~\cite{BoteroReznik2004} uses a slightly different convention for the Majorana
coordinates, in which the covariance matrix of one entangled pair is
\begin{equation}
 \widetilde{\Gamma}^{\mathrm{can}}_j
 =
 \begin{pmatrix}
  0&-\lambda_j&0&\nu_j\\
  \lambda_j&0&\nu_j&0\\
  0&-\nu_j&0&-\lambda_j\\
  -\nu_j&0&\lambda_j&0
 \end{pmatrix},
 \qquad
 \lambda_j^2+\nu_j^2=1.
 \label{eq:botero-canonical-covariance-block}
\end{equation}
This is locally \textit{orthogonally equivalent} to
\eqref{eq:canonical-covariance-block}. Indeed, let
\begin{equation}
 X=\begin{pmatrix}0&1\\1&0\end{pmatrix},
 \qquad
 S_j=X\oplus \mathbb{I}_2.
\end{equation}
Since \(X^{\mathsf T}=X=X^{-1}\) and \(XJX^{\mathsf T}=-J\), one has
\begin{equation}
 \widetilde{\Gamma}^{\mathrm{can}}_j
 =S_j\Gamma^{\mathrm{can}}_jS_j^{\mathsf T}.
 \label{eq:botero-our-convention}
\end{equation}
Thus the two forms differ only by exchanging the two Majorana coordinates
of the \(A\) mode of each entangled pair.

For the complete covariance matrix, let
\(S_A=\bigoplus_{j=1}^{\ell}X\) on the paired \(A\) modes and let \(S_B\)
be the identity on \(B\), including the spectator modes. Then
\begin{equation}
\begin{aligned}
 &(O_A\oplus O_B)\widetilde{\Gamma}^{\mathrm{can}}
      (O_A\oplus O_B)^{\mathsf T} \\
 &\quad =
 \bigl((O_AS_A)\oplus O_B\bigr)\Gamma^{\mathrm{can}}
 \bigl((O_AS_A)\oplus O_B\bigr)^{\mathsf T}.
\end{aligned}
\label{eq:canonical-orbit-convention-equivalence}
\end{equation}
Right multiplication by \(S_A\) is a bijection of the local orthogonal
group. Consequently, the two conventions give the same result concerning Theorem~\ref{thm:global} (same optimized stabilizer entropy). 
In the following, we use Eq.~\eqref{eq:canonical-covariance-block} because it gives more convenient and transparent formulas.
\end{remark}

A direct four-Majorana computation on Eq.~\eqref{eq:canonical-covariance-block} gives the six potentially nonzero entries $q^{(j)}_{\mathrm{can}}=\tfrac14(1,1,\lambda_j^2,\lambda_j^2,\nu_j^2,\nu_j^2)$, while a pure unpaired, or spectator, mode --- one that lives entirely in $A$ or $B$ --- contributes $(1,1)/2$. Equation~\eqref{eq:q-collision} then yields the canonical value
\begin{equation}
    \zeta_{2,\mathrm{can}}=\prod_{j=1}^{\ell}g_j\;,\qquad g_j=1-\lambda_j^2+\lambda_j^4\;,
    \label{eq:canonical-value}
\end{equation}
which is the two-qubit result quoted in the main text. Theorem~\ref{thm:global} therefore amounts to the statement that Eq.~\eqref{eq:canonical-value} is the \emph{maximum} of $\zeta_2$ on the orbit~\eqref{eq:local-action}.

The starting point of the proof is the following reformulation.

\begin{lemma}[Projection-DPP identity]
\label{lem:DPP}
Let $K=(\mathbb{I}+i\Gamma)/2$. Then $K$ is a Hermitian rank-$N$ projector, the numbers $p_S=\det K_S$ with $|S|=N$ form the projection determinantal point process (DPP)~\cite{collura2026nonstabilizernessfermionicgaussianstates} associated with $K$, and
\begin{equation}
    \zeta_2(\Gamma)=d\sum_{|S|=N}p_S^2\;.
    \label{eq:DPP-collision}
\end{equation}
\end{lemma}

\begin{proof}
Equation~\eqref{eq:pure-covariance} gives $K^\dagger=K$, $K^2=K$ and $\mathrm{Tr}\,K=N$. For a projection kernel, the Cauchy--Binet formula gives $p_S\ge0$ and $\sum_{|S|=N}p_S=1$. Extend $p$ by zero from the $N$-th slice to all subsets of $[2N]$ and use the Boolean Fourier convention $\widehat p(I)=\sum_{|S|=N}p_S(-1)^{|S\cap I|}$. The determinantal generating function gives
\begin{equation}
    \widehat p(I)=\det(\mathbb{I}_{|I|}-2K_I)=(-i)^{|I|}\det\Gamma_I\;,
\end{equation}
so that Parseval's identity on the Boolean cube yields $\sum_{|S|=N}p_S^2=d^{-2}\sum_I\det(\Gamma_I)^2$. Combining with Eq.~\eqref{eq:zeta-M} proves Eq.~\eqref{eq:DPP-collision}.
\end{proof}

\subsubsection{A sharp collision inequality for one Slater determinant}
\label{app:proof:slater}

The first new ingredient bounds the collision probability of an arbitrary Slater determinant~\cite{Mbeng2024,surace}, specified by the isometry $Q$ whose columns span the occupied space, by its \emph{transpose} overlap $Q^{\trans}Q$. The transpose, rather than the adjoint, is what remembers the particle--hole reality constraint in the fermionic application.

Throughout this subsection we use the shorthand
\begin{equation}
    \mathcal{D}(X):=\det\!\left(\mathbb{I}+X^\dagger X\right)
    =\sum_{I,J}\left|\det X_{I,J}\right|^2 ,
    \label{eq:Dfunctional}
\end{equation}
the second equality being the Cauchy--Binet expansion, in which $(I,J)$ runs over all pairs of index sets of equal cardinality, the empty pair contributing $1$.

\begin{lemma}[Sharp Slater collision bound]
\label{lem:Slater-bound}
Let $Q\in\mathbb{C}^{n\times k}$ satisfy $Q^\dagger Q=\mathbb{I}_k$, and set $P=QQ^\dagger$ and $\Sigma=Q^{\trans}Q$. Then the collision probability of the associated projection DPP,
\begin{equation}
    \mathcal{C}(Q):=\sum_{|I|=k}p_I^2\,,
    \qquad p_I=\det P_I=\left|\det Q_I\right|^2 ,
\end{equation}
obeys
\begin{equation}
    \mathcal{C}(Q)\;\le\;2^{-k}\,\mathcal{D}(\Sigma)\;.
    \label{eq:Slater-bound}
\end{equation}
Both sides are invariant under a unitary change of column gauge $Q\mapsto QW$.
\end{lemma}

\begin{proof}
First, we note that gauge invariance is immediate: $P$ is unchanged by $Q\mapsto QW$, while $\Sigma\mapsto W^{\trans}\Sigma W$ leaves $\Sigma^\dagger\Sigma$ unitarily equivalent to itself.

The proof is an induction on the number of rows $n$, in which \emph{both} sides of Eq.~\eqref{eq:Slater-bound} are shown to satisfy the same deletion recursion. The inequality enters only through a single determinant estimate, cf. Eq.~\eqref{eq:determinant-deletion} below.

Let us fix a coordinate $a\in \{1,\dots,n\}$ and set
\begin{equation}
    p=\langle e_a, P e_a\rangle=\left\lVert Q^\dagger e_a\right\rVert^2 ,\qquad r=1-p ,
\end{equation}
assuming first $0<p<1$, where $e_a$ is the canonical basis vector with $(e_a)_m=\delta_{ma}$.
We choose the column gauge so that $Q^\dagger e_a=\sqrt p\,e_1$ with $e_1\in\mathbb{C}^k$; the $a$-th row of $Q$ is then $\sqrt p\,e_1^{\trans}$. Let $\widehat Q\in\mathbb{C}^{(n-1)\times k}$ be $Q$ with that row deleted, so that $Q^\dagger Q=\mathbb{I}_k$ gives
\begin{equation}
    \widehat Q^\dagger\widehat Q=\mathbb{I}_k-p\,e_1e_1^{\trans}=\Delta_r^2 ,
    \quad
    \Delta_r:=\mathrm{diag}(\sqrt r,1,\dots,1).\end{equation}
Hence
\begin{equation}
    Q_0:=\widehat Q\,\Delta_r^{-1}\in\mathbb{C}^{(n-1)\times k}
\end{equation}
is an isometry, and so are the last $k-1$ columns of $\widehat Q$, which we call $Q_1\in\mathbb{C}^{(n-1)\times(k-1)}$. These describe the process conditioned respectively on the exclusion and on the inclusion of $a$: expanding $\det Q_I$ along the $a$-th row for $a\in I$, and using $\det\Delta_r=\sqrt r$ for $a\notin I$, gives $\det P_I=p\,\det (P_1)_{I\setminus\{a\}}$ and $\det P_I=r\,\det (P_0)_I$ respectively, with $P_0=Q_0Q_0^\dagger$ and $P_1=Q_1Q_1^\dagger$. Squaring and summing,
\begin{equation}
    \mathcal{C}(Q)=p^2\,\mathcal{C}(Q_1)+r^2\,\mathcal{C}(Q_0)\;.
    \label{eq:collision-deletion}
\end{equation}
We now write the transpose overlap of $Q_0$ in block form, singling out its first column,
\begin{equation}
\Sigma_0:=Q_0^{\trans}Q_0:=\left[\,u\;\;V\,\right]
    =\begin{pmatrix} u_1&y^{\trans}\\ y&\Sigma_1\end{pmatrix},
    \quad \Sigma_1=Q_1^{\trans}Q_1 .\end{equation}
Since $\Sigma=Q^{\trans}Q=\widehat Q^{\trans}\widehat Q+p\,e_1e_1^{\trans}$ and $\widehat Q=Q_0\Delta_r$,
\begin{equation}\Sigma=p\,e_1e_1^{\trans}+\Delta_r\,\Sigma_0\,\Delta_r
    =\begin{pmatrix} p+r\,u_1 & \sqrt r\,y^{\trans}\\[2pt] \sqrt r\,y & \Sigma_1\end{pmatrix}.
    \label{eq:S-deletion}
\end{equation} 
The matrices $Q_1$ and $Q_0$ have $k-1$ and $k$ columns respectively, so the induction hypothesis reads $\mathcal{C}(Q_1)\le 2^{-(k-1)}\mathcal{D}(\Sigma_1)$ and $\mathcal{C}(Q_0)\le 2^{-k}\mathcal{D}(\Sigma_0)$. Inserting these into Eq.~\eqref{eq:collision-deletion}, Eq.~\eqref{eq:Slater-bound} follows provided we prove that
\begin{equation}
    \mathcal{D}(\Sigma)\;\ge\;2p^2\,\mathcal{D}(\Sigma_1)+r^2\,\mathcal{D}(\Sigma_0)\;;
    \label{eq:determinant-deletion}
\end{equation}
where the factor $2$ in front of $p^2$ is precisely the mismatch $2^{-(k-1)}=2\cdot 2^{-k}$ between the normalizations at $k-1$ and $k$ columns.

To this end, let us apply Eq.~\eqref{eq:Dfunctional} to $\Sigma_0$ and split the sum according to whether the index $1$ belongs to the row set, to the column set, to both, or to neither:
\begin{equation}
\begin{split}
    \mathcal{N}&:=\!\!\sum_{1\notin I,\,1\notin J}\!\!\left|\det(\Sigma_0)_{I,J}\right|^2=\mathcal{D}(\Sigma_1),\\
    \mathcal{X}&:=\!\!\sum_{1\in I,\,1\notin J}\!\!\left|\det(\Sigma_0)_{I,J}\right|^2,\\
    \mathcal{Y}&:=\!\!\sum_{1\in I,\,1\in J}\!\!\left|\det(\Sigma_0)_{I,J}\right|^2,
\end{split}
\end{equation}
so that, $\Sigma_0$ being symmetric, the column-only class is another copy of $\mathcal{X}$ and
\begin{equation}
    \mathcal{D}(\Sigma_0)=\mathcal{N}+2\mathcal{X}+\mathcal{Y},
    \qquad
    \mathcal{D}(V)=\mathcal{N}+\mathcal{X}.
    \label{eq:minor-classes-S0}
\end{equation}
We now expand $\mathcal{D}(\Sigma)$ in the same way, comparing each minor of $\Sigma$ with the corresponding minor of $\Sigma_0$ through Eq.~\eqref{eq:S-deletion}. If $1\notin I$ and $1\notin J$ the submatrix is unchanged; if exactly one of $I,J$ contains $1$, the corresponding row or column is scaled by $\sqrt r$ and the squared modulus picks up a factor $r$. If instead $1\in I$ and $1\in J$, we set $I'=I\setminus\{1\}$, $J'=J\setminus\{1\}$ and note that Eq.~\eqref{eq:S-deletion} restricted to those indices reads$\Sigma_{I,J}=p\,e_1e_1^{\trans}+\Delta_r\,(\Sigma_0)_{I,J}\,\Delta_r$. Expanding the determinant by multilinearity in the first row, the term proportional to $p$ is developed along $(p,0,\dots,0)$ and leaves the complementary minor, which involves only rows and columns $>1$ and is therefore unscaled, while the remaining term carries $\det\Delta_r^2=r$:
\begin{equation}\det\Sigma_{I,J}=p\,\det(\Sigma_1)_{I',J'}+r\,\det(\Sigma_0)_{I,J}\;.
    \label{eq:minor-splitting}
\end{equation}\DIFaddend 
Taking squared moduli, $|p\,\alpha+r\,\beta|^2=p^2|\alpha|^2+r^2|\beta|^2+2pr\,\mathrm{Re}(\alpha\bar\beta)$, and summing over all such pairs produces $p^2\mathcal{N}+r^2\mathcal{Y}+2pr\,\mathrm{Re}\,\mathcal{Z}$, where we defined the cross term
\begin{equation}
    \mathcal{Z}:=\sum_{I',J'}\det(\Sigma_1)_{I',J'}\;
    {\det(\Sigma_0^\dagger)_{\{1\}\cup I',\,\{1\}\cup J'}}\;.
    \label{eq:cross-term-def}
\end{equation}\DIFaddend 
Collecting the four classes,
\begin{equation}
    \mathcal{D}(\Sigma)=\mathcal{N}+2r\mathcal{X}+p^2\mathcal{N}+r^2\mathcal{Y}+2pr\,\mathrm{Re}\,\mathcal{Z}\;.
    \label{eq:minor-classes-S}
\end{equation}

We claim that $\mathrm{Re}\,\mathcal{Z}\ge-\mathcal{D}(V)$. The strategy is to compute one and the same derivative in two different ways. Shift the $(1,1)$ entry of $\Sigma_0$ by a real amount $q$ and set
\begin{equation}
    f(q):=\mathcal{D}\!\left(\Sigma_0+q\,e_1e_1^{\trans}\right).
\end{equation}

\emph{(i) By the minor expansion.} Equation~\eqref{eq:minor-splitting}, applied with$\Delta_r=\mathbb{I}$ and $p$ replaced by $q$, says that a minor of $\Sigma_0+q\,e_1e_1^{\trans}$ equals $\det(\Sigma_0)_{I,J}+q\det(\Sigma_1)_{I',J'}$ when $1\in I\cap J$, and is independent of $q$ otherwise. Taking squared moduli and summing over all pairs, using the definition~\eqref{eq:cross-term-def} of $\mathcal{Z}$, we obtain
\begin{equation}
    f(q)=\mathcal{D}(\Sigma_0)+2q\,\mathrm{Re}\,\mathcal{Z}+q^2\mathcal{N}\;,
    \label{eq:f-minor}
\end{equation}
so that in particular $\mathrm{Re}\,\mathcal{Z}=\tfrac12f'(0)$.

\emph{(ii) In closed form.} For an arbitrary matrix $[\,w\;\;V\,]$, split into its first column $w$ and the remaining columns $V$, then by Schur's complement formula
\begin{equation}\mathcal{D}\!\left(\left[\,w\;V\,\right]\right)=\mathcal{D}(V)\left[1+w^\dagger w-w^\dagger V\left(\mathbb{I}+V^\dagger V\right)^{-1}V^\dagger w\right].
\end{equation}
The Woodbury identity~\cite{MR38136} $V(\mathbb{I}+V^\dagger V)^{-1}V^\dagger=\mathbb{I}-(\mathbb{I}+VV^\dagger)^{-1}$ collapses the bracket into
\DIFdelbegin \DIFdel{,
}\DIFdelend \begin{equation}
\begin{split}
    \mathcal{D}\!\left(\left[\,w\;V\,\right]\right)&=\mathcal{D}(V)\left[1+w^\dagger\mathcal{R}\,w\right],\\
    \mathcal{R}&:=\left(\mathbb{I}+VV^\dagger\right)^{-1}.
\end{split}
    \label{eq:schur-complement}
\end{equation}
Since $\Sigma_0=[\,u\;V\,]$, the shifted matrix is $[\,u+q\,e_1\;\;V\,]$, so that $f(q)=\mathcal{D}(V)\left[1+(u+q\,e_1)^\dagger\mathcal{R}(u+q\,e_1)\right]$ and therefore
\begin{equation}
    \tfrac12f'(0)=\mathcal{D}(V)\,\mathrm{Re}\!\left(e_1^\dagger\mathcal{R}\,u\right).
\end{equation}
Consistently, the coefficient of $q^2$ here is $\mathcal{D}(V)\,e_1^\dagger\mathcal{R}e_1=\mathcal{D}([\,e_1\;V\,])-\mathcal{D}(V)=\mathcal{N}$ from Eq.~\eqref{eq:schur-complement}, matching Eq.~\eqref{eq:f-minor}.

\emph{(iii) Comparison and bound.} Equating the two evaluations of $\tfrac12f'(0)$,
\begin{equation}
    \mathrm{Re}\,\mathcal{Z}=\mathcal{D}(V)\,\mathrm{Re}\!\left(e_1^\dagger\mathcal{R}\,u\right),
\end{equation}
and it remains to bound the scalar. The matrix $Q_0$ is an isometry, so $\lVert Q_0\rVert=1$; since transposition and complex conjugation preserve the operator norm, $\lVert\Sigma_0\rVert=\lVert Q_0^{\trans}Q_0\rVert\le1$, and in particular its first column $u=\Sigma_0e_1$ obeys $\lVert u\rVert\le1$. Moreover $VV^\dagger\succeq0$ gives $0<\mathcal{R}\le\mathbb{I}$, so that Cauchy--Schwarz in the inner product defined by $\mathcal{R}$ yields
\begin{equation}
    \left|e_1^\dagger\mathcal{R}\,u\right|\le\lVert\mathcal{R}^{1/2}e_1\rVert\,\lVert\mathcal{R}^{1/2}u\rVert\le\lVert e_1\rVert\,\lVert u\rVert\le1 .
\end{equation}
Hence $\mathrm{Re}(e_1^\dagger\mathcal{R}u)\ge-1$ and, $\mathcal{D}(V)$ being nonnegative,
\begin{equation}
    \mathrm{Re}\,\mathcal{Z}\ge-\mathcal{D}(V)=-(\mathcal{N}+\mathcal{X})\;.
    \label{eq:cross-lower}
\end{equation}

We prove Eq.~\eqref{eq:determinant-deletion} using Eq.~\eqref{eq:minor-classes-S0}, Eq.~\eqref{eq:minor-classes-S} and $r=1-p$, so that $1-p^2-r^2=2pr$,
\begin{equation}
\begin{split}
    &\mathcal{D}(\Sigma)-2p^2\mathcal{D}(\Sigma_1)-r^2\mathcal{D}(\Sigma_0)\\
    &\qquad\qquad =2pr\left(\mathcal{N}+\mathcal{X}+\mathrm{Re}\,\mathcal{Z}\right)\;\ge\;0
    \end{split}
\end{equation} 
by Eq.~\eqref{eq:cross-lower}, closing the induction step. For the base cases, $k=0$ gives $1\le1$, while for $n=k$ the matrix $Q$ is unitary, so $\mathcal{C}(Q)=1$ and $\Sigma$ is unitary as well, hence $\mathcal{D}(\Sigma)=2^k$ and Eq.~\eqref{eq:Slater-bound} holds with equality. Finally, if $p=0$ or $p=1$ only one of the two conditional processes is present and the conclusion follows directly, or by continuity in $p$.
\end{proof}

\subsubsection{ BCS parametrization and reduction to a single rotation}
\label{app:proof:thouless}

From now on we assume a balanced cut with no unpaired modes, so that $N=2\ell$ and each side carries $2\ell$ Majorana coordinates; the general case is recovered at the end of Sec.~\ref{app:proof:completion}.
We work first in the interior $0\le\lambda_j<1$, so that every $\nu_j>0$,
and introduce the canonical BCS matrix~\cite{Mbeng2024}
\begin{equation}
Z_0=\bigoplus_{j=1}^{\ell}\frac{\mathbb{I}_2+i\lambda_jJ}{\nu_j}=\bigoplus_{j=1}^{\ell}e^{i\vartheta_jJ}\;,\quad \tanh \vartheta_j=\lambda_j\;,
    \label{eq:Z0}
\end{equation} 
which is positive Hermitian and complex orthogonal, $Z_0^{\trans}Z_0=\mathbb{I}$ and $\overline{Z_0}=Z_0^{-1}$. Each $2\times2$ block of $Z_0$ encodes one canonical pair of the BCS-like form of Eq.~\eqref{eq:can_form}. In this parametrization, local Gaussian unitaries act as the two-sided rotation $Z_0\mapsto Z=O_AZ_0O_B^{\trans}$ by real orthogonal matrices.

\begin{lemma}[Balanced all-minors representation]
\label{lem:all-minors}
For $Z=O_AZ_0O_B^{\trans}$ define the all-minors functional
\begin{equation}
    \mathsf{F}(Z)=\sum_{I,J\,:\,|I|=|J|}\left|\det Z_{I,J}\right|^4\;.
    \label{eq:F-definition}
\end{equation}
Then, writing with a slight abuse of notation $\zeta_2(Z)$ for the stabilizer purity of the state parametrized by $Z$, the collision functional is exactly
\begin{equation}\zeta_2(Z)=d^{-1}\Big(\prod_{j=1}^{\ell}\nu_j^4\Big)\,\mathsf{F}(Z)\;,
    \label{eq:physical-F}
\end{equation}
and at the canonical point
\begin{equation}
    \mathsf{F}(Z_0)=4^{\ell}\prod_{j=1}^{\ell}\frac{1-\lambda_j^2+\lambda_j^4}{\nu_j^4}\;.
    \label{eq:F-canonical}
\end{equation}
\end{lemma}

\begin{proof}
After a gauge choice in the occupied space, the rank-$2\ell$ projector of Lemma~\ref{lem:DPP} is represented by the isometry
\begin{equation}
    Q_Z=\begin{pmatrix}iZ\\ \mathbb{I}_{2\ell}\end{pmatrix}\left(\mathbb{I}+Z^\dagger Z\right)^{-1/2}\;,
    \label{eq:graph-isometry-Z}
\end{equation}
whose projector, for Hermitian $Z$, is
\begin{equation}
    K_Z=Q_ZQ_Z^\dagger=\begin{pmatrix} Z(\mathbb{I}_{2\ell}+Z^2)^{-1}Z& iZ(\mathbb{I}_{2\ell}+Z^2)^{-1}\\[2pt] -i(\mathbb{I}_{2\ell}+Z^2)^{-1}Z&(\mathbb{I}_{2\ell}+Z^2)^{-1}\end{pmatrix}.
    \label{eq:graph-projector-Z}
\end{equation}

For a single canonical block $Z_{0,j}=(\mathbb{I}_2+i\lambda_jJ)/\nu_j$, direct block algebra gives
\begin{equation}
\begin{split}
    Z_{0,j}(\mathbb{I}_2+Z_{0,j}^2)^{-1}Z_{0,j}&=\tfrac12(\mathbb{I}_2+i\lambda_jJ)\;,\\
    Z_{0,j}(\mathbb{I}_2+Z_{0,j}^2)^{-1}&=\tfrac12\nu_j\mathbb{I}_2\;,\\
    (\mathbb{I}_2+Z_{0,j}^2)^{-1}&=\tfrac12(\mathbb{I}_2-i\lambda_jJ)\;,
    \end{split}
\end{equation}
so that $Q_{Z_{0,j}}Q_{Z_{0,j}}^\dagger=(\mathbb{I}_4+i\Gamma^{\mathrm{can}}_j)/2$ : Eq.~\eqref{eq:graph-isometry-Z} reproduces exactly the canonical covariance of Proposition~\ref{prop:covariance-normal-form}, and in particular $i\Gamma_A=2(K_Z)_{AA}-\mathbb{I}$ has eigenvalues $\pm\lambda_j$. We note that, for Hermitian $Z$, the complex orthogonality $\overline{Z}=Z^{-1}$ is precisely what makes the state physical: a direct computation on Eq.~\eqref{eq:graph-projector-Z} shows that it implies $\overline{K_Z}=\mathbb{I}_{4\ell}-K_Z$, i.e.\ that $K_Z=(\mathbb{I}_{4\ell}+i\Gamma)/2$ with $\Gamma$ real, as in Lemma~\ref{lem:DPP}. Acting with $O_A\oplus O_B$ on the rows and using a right occupied-space gauge replaces $Z_0$ by $O_AZ_0O_B^{\trans}$, so Eq.~\eqref{eq:graph-isometry-Z} holds throughout the local orbit.

Given a set $S$ with $|S|=2\ell$, we know from Lemma~\ref{lem:DPP} that $p_S=|\det(Q_Z)_S|^2$, where $(Q_Z)_S$ is the square $2\ell\times2\ell$ submatrix formed by the rows of $Q_Z$ selected by $S$. The right factor $(\mathbb{I}_{2\ell}+Z^\dagger Z)^{-1/2}$ multiplies all rows at once, so it factors out of the square determinant: $\det(Q_Z)_S=\kappa\,\det\widetilde{Q}_S$, where $\widetilde{Q}_S$ collects the selected rows of the unnormalized matrix $\big(\begin{smallmatrix}iZ\\ \mathbb{I}_{2\ell}\end{smallmatrix}\big)$ and $\kappa=\det(\mathbb{I}_{2\ell}+Z^\dagger Z)^{-1/2}$. 
The set $S$ selects $k$ rows in the upper block $iZ$ of Eq.~\eqref{eq:graph-isometry-Z}, labeled by a subset $I$ with $|I|=k$, and $2\ell-k$ rows in the lower block $\mathbb{I}_{2\ell}$, whose column complement is labeled by $J$.
We are left to determine $\det\widetilde{Q}_S$: expanding along the rows selected from the identity block, what survives is the minor of $iZ$ on the rows $I$ and on the columns $J$ left untouched. Up to a sign fixed by the row ordering, $\det\widetilde{Q}_S=\pm\,i^k\det Z_{I,J}$, and $|J|=2\ell-(2\ell-k)=k=|I|$ automatically. Explicitly,
\begin{equation}
    \kappa=\det\!\left(\mathbb{I}_{2\ell}+Z_0^\dagger Z_0\right)^{-1/2}=2^{-\ell}\prod_{j=1}^{\ell}\nu_j\;,
    \label{eq:kappa}
\end{equation}
where the first equality uses $Z^\dagger Z=O_B(Z_0^\dagger Z_0)O_B^{\trans}$, which has the same determinant as $Z_0^\dagger Z_0$, and the second follows blockwise: from $(\mathbb{I}_2+i\lambda_jJ)^2=(1+\lambda_j^2)\mathbb{I}_2+2i\lambda_jJ$ one finds $\mathbb{I}_2+Z_{0,j}^2=2(\mathbb{I}_2+i\lambda_jJ)/\nu_j^2$ , so that $\det(\mathbb{I}_2+Z_{0,j}^2)=4/\nu_j^2$ and $\kappa=\prod_{j}(\nu_j/2)$. Altogether, $p_S=\kappa^2|\det Z_{I,J}|^2$: squaring and summing over all $S$, i.e.\ over all pairs $(I,J)$ with $|I|=|J|$, gives $\sum_{|S|=2\ell}p_S^2=\kappa^4\mathsf{F}(Z)$, and Lemma~\ref{lem:DPP} with $N=2\ell$, $d=2^{2\ell}$ and $\kappa^4=2^{-4\ell}\prod_j\nu_j^4$ proves Eq.~\eqref{eq:physical-F}.

The functional $\mathsf{F}$ is multiplicative on direct sums, and a direct $2\times2$ computation gives
\begin{equation}
    \mathsf{F}\!\left(\frac{\mathbb{I}_2+i\lambda J}{\sqrt{1-\lambda^2}}\right)=\frac{4\left(1-\lambda^2+\lambda^4\right)}{(1-\lambda^2)^2}\;,
    \label{eq:F_Z_0}
\end{equation}
which multiplied over the canonical blocks proves Eq.~\eqref{eq:F-canonical}. As a check, inserting Eq.~\eqref{eq:F-canonical} into Eq.~\eqref{eq:physical-F} with $d=4^{\ell}$ returns exactly the canonical value~\eqref{eq:canonical-value}.
\end{proof}
By Eqs.~\eqref{eq:physical-F}--\eqref{eq:F-canonical}, Theorem~\ref{thm:global} is equivalent to the inequality
\begin{equation}
    \mathsf{F}(O_AZ_0O_B^{\trans})\le \mathsf{F}(Z_0)\;.
    \label{eq:Cartan-target}
\end{equation}
The next exact reduction shows that it suffices to control a single common real basis.

\begin{lemma}
\label{lem:two-to-one}
For all $O_A,O_B\in \mathrm{O}(2\ell)$,
\begin{equation}
    \mathsf{F}\!\left(O_AZ_0O_B^{\trans}\right)\le\sqrt{\mathsf{F}\!\left(O_AZ_0O_A^{\trans}\right)\,\mathsf{F}\!\left(O_BZ_0O_B^{\trans}\right)}\;,
    \label{eq:two-to-one}
\end{equation}
and consequently
\begin{equation}
    \sup_{O_A,O_B}\mathsf{F}\!\left(O_AZ_0O_B^{\trans}\right)=\sup_{O}\mathsf{F}\!\left(OZ_0O^{\trans}\right)\;.
    \label{eq:sup-two-to-one}
\end{equation}
\end{lemma}

\begin{proof}
Let $\mathcal{E}(X)$ denote the matrix collecting all the minors of $X$: its rows and columns are labeled by subsets $I,J\subseteq \{1,\dots,2\ell\}$, and its entries are the minors,
\begin{equation}
    \mathcal{E}(X)_{I,J}=\begin{cases}\det X_{I,J}, & |I|=|J|,\\ 0,&\text{otherwise.}\end{cases}
    \label{eq:compound}
\end{equation}
The Cauchy--Binet formula is precisely the statement that $\mathcal{E}$ is multiplicative, $\mathcal{E}(XY)=\mathcal{E}(X)\mathcal{E}(Y)$, while $\mathcal{E}(X^{\trans})=\mathcal{E}(X)^{\trans}$ and $\mathcal{E}(\overline X)=\overline{\mathcal{E}(X)}$ are immediate from Eq.~\eqref{eq:compound}; together these give $\mathcal{E}(X^\dagger)=\mathcal{E}(X)^\dagger$. Two consequences are all we need. First, $\mathcal{E}(O)$ is real orthogonal whenever $O$ is, since $\mathcal{E}(O)\mathcal{E}(O)^{\trans}=\mathcal{E}(OO^{\trans})=\mathbb{I}$. Second, $\mathcal{E}_0:=\mathcal{E}(Z_0)$ is positive Hermitian: diagonalizing $Z_0=\Upsilon\Lambda_0\Upsilon^\dagger$ with $\Upsilon$ unitary and $\Lambda_0>0$  diagonal, multiplicativity gives $\mathcal{E}_0=\mathcal{E}(\Upsilon)\mathcal{E}(\Lambda_0)\mathcal{E}(\Upsilon)^\dagger$, where $\mathcal{E}(\Upsilon)$  is unitary and $\mathcal{E}(\Lambda_0)$ is diagonal with positive entries $\prod_{i\in I}(\Lambda_0)_{ii}$.

By multiplicativity, $\mathcal{E}(O_AZ_0O_B^{\trans})=\mathcal{E}(O_A)\,\mathcal{E}_0\,\mathcal{E}(O_B)^{\trans}$, so that
\begin{equation}
    \det\!\left(O_AZ_0O_B^{\trans}\right)_{I,J}=\alpha_I^{\trans}\,\mathcal{E}_0\,\beta_J\;,
\end{equation}
where $\alpha_I$ and $\beta_J$ denote the (real) rows of $\mathcal{E}(O_A)$ and $\mathcal{E}(O_B)$. Splitting $\mathcal{E}_0=\mathcal{E}_0^{1/2}\mathcal{E}_0^{1/2}$ and using that $\mathcal{E}_0^{1/2}$ is Hermitian while $\alpha_I,\beta_J$ are real, this is an ordinary overlap,
\begin{equation}
\begin{split}
    \det\!\left(O_AZ_0O_B^{\trans}\right)_{I,J}&=\langle u_I,v_J\rangle\;,\\
    u_I=\mathcal{E}_0^{1/2}\alpha_I\;,&\qquad v_J=\mathcal{E}_0^{1/2}\beta_J\;.
\end{split}
    \label{eq:minors-as-overlaps}
\end{equation} 
Hence the all-minors functional is a sum of fourth powers of overlaps between two families of vectors, obtained from the same positive matrix $\mathcal{E}_0$ by two real orthogonal changes of basis; note that $\langle u_I,v_J\rangle=0$ whenever $|I|\neq|J|$, since the $\mathcal{E}$-matrices vanish between blocks of different cardinality, so the unrestricted sum over pairs $(I,J)$ below coincides with $\mathsf{F}$:
\begin{equation}
    \mathsf{F}\!\left(O_AZ_0O_B^{\trans}\right)=\sum_{I,J}\left|\langle u_I,v_J\rangle\right|^4 .
    \label{eq:F-as-overlaps}
\end{equation}

Introduce the rank-one matrices $\mathcal{U}_I=u_Iu_I^\dagger$ and $\mathcal{V}_J=v_Jv_J^\dagger$, so that $|\langle u_I,v_J\rangle|^2=\mathrm{Tr}(\mathcal{U}_I\mathcal{V}_J)$, and define the positive semidefinite operators acting on two copies of the space
\begin{equation}
    \mathcal{T}_A=\sum_I \mathcal{U}_I\otimes \mathcal{U}_I\;,\qquad
    \mathcal{T}_B=\sum_J \mathcal{V}_J\otimes \mathcal{V}_J\;.
\end{equation} 
Since $\mathrm{Tr}\!\left[(X\otimes X)(Y\otimes Y)\right]=\left(\mathrm{Tr}\,XY\right)^2$, Eq.~\eqref{eq:F-as-overlaps} becomes
\begin{equation}
    \mathsf{F}\!\left(O_AZ_0O_B^{\trans}\right)=\mathrm{Tr}\!\left(\mathcal{T}_A\mathcal{T}_B\right).
    \label{eq:F-as-trace}
\end{equation}
The same identity applied with $O_B=O_A$ gives $\mathrm{Tr}(\mathcal{T}_A^2)=\mathsf{F}(O_AZ_0O_A^{\trans})$, and likewise for $B$.

Both $\mathcal{T}_A$ and $\mathcal{T}_B$ are positive semidefinite, so the Cauchy--Schwarz inequality for the trace inner product applies to Eq.~\eqref{eq:F-as-trace},
\begin{equation}
    \mathrm{Tr}\!\left(\mathcal{T}_A\mathcal{T}_B\right)\le\sqrt{\mathrm{Tr}\!\left(\mathcal{T}_A^2\right)\mathrm{Tr}\!\left(\mathcal{T}_B^2\right)}\;,
\end{equation}
which is exactly Eq.~\eqref{eq:two-to-one}. Since $\sqrt{xy}\le\max(x,y)$, taking suprema in Eq.~\eqref{eq:two-to-one} bounds the two-sided supremum by $\sup_O\mathsf{F}(OZ_0O^{\trans})$, while restricting to $O_A=O_B$ gives the reverse inequality, proving Eq.~\eqref{eq:sup-two-to-one}.
\end{proof}

\subsubsection{From all minors to a principal-minor residual}
\label{app:proof:residual}

By Lemma~\ref{lem:two-to-one} it suffices to consider $Z=OZ_0O^{\trans}$ with $O\in\mathrm{O}(2\ell)$,  which is again positive Hermitian and complex orthogonal. Put $H=Z^2$, so that $H$ is positive Hermitian and $HH^{\trans}=\mathbb{I}_{2\ell}$.
The purpose of this subsection is to bound $\mathsf{F}(Z)$ by a simpler quantity involving only the principal minors of $H$, using Lemma~\ref{lem:Slater-bound}.
\begin{proposition}[All-minors residual]
\label{prop:residual}
With $Z$ and $H$ as above,
\begin{equation}
    \mathsf{F}(Z)\le \mathsf{R}(H)\;,\quad \mathsf{R}(H)=\sum_{I\subseteq[m]}2^{-|I|}\det\!\left(\mathbb{I}+H_I\overline{H_I}\right),
    \label{eq:R-definition}
\end{equation}
with equality at the canonical point, $\mathsf{F}(Z_0)=\mathsf{R}(H_0)$, $H_0=Z_0^2$.
\end{proposition}

\begin{proof}

Fix $S\subseteq\{1,\dots,2\ell\}$ with $|S|=k$, let $E_S$ be the $2\ell\times k$ matrix whose columns are the canonical basis vectors labeled by $S$, and let $G_S=E_S^\dagger Z^\dagger ZE_S=H_S$ be the corresponding matrix. Then $Q_S=ZE_SG_S^{-1/2}$ is an isometry, i.e. exactly the object to which Lemma~\ref{lem:Slater-bound} applies.
Its collision probability reproduces, up to normalization, the row sums of $\mathsf{F}$: the rows of $Q_S$ labeled by a set $J$ with $|J|=k$ form the square matrix $Z_{J,S}G_S^{-1/2}$, so the probabilities of Lemma~\ref{lem:Slater-bound} are $p_J=|\det Z_{J,S}|^2/\det G_S$ and $\mathcal{C}(Q_S)=\det(G_S)^{-2}\sum_{|J|=k}|\det Z_{J,S}|^4$; moreover, $Z=Z^\dagger$ gives $\det Z_{J,S}=\overline{\det Z_{S,J}}$, so $|\det Z_{J,S}|=|\det Z_{S,J}|$ and the sum can equivalently be taken over column sets at fixed row set $S$. We next compute the transpose overlap. Complex orthogonality, $Z^{\trans}Z=\mathbb{I}_{2\ell}$, gives $E_S^{\trans}Z^{\trans}ZE_S=E_S^{\trans}E_S=\mathbb{I}_k$, so that $\Sigma=Q_S^{\trans}Q_S=G_S^{-\trans/2}G_S^{-1/2}$. Since $G_S$ is positive Hermitian, $\Sigma^\dagger\Sigma=G_S^{-1/2}G_S^{-\trans}G_S^{-1/2}$, and
\begin{equation}
\mathbb{I}_k+\Sigma^\dagger\Sigma=G_S^{-1/2}\left(G_S+G_S^{-\trans}\right)G_S^{-1/2}\;.
\end{equation}%
Factoring $\det(G_S+G_S^{-\trans})=\det(G_S^{-\trans})\det(\mathbb{I}_k+G_S^{\trans}G_S)$ and using $G_S^{\trans}=\overline{H_S}$ together with $\det(\mathbb{I}_k+\overline{H_S}H_S)=\det(\mathbb{I}_k+H_S\overline{H_S})$, we obtain $\mathcal{D}(\Sigma)=\det(G_S)^{-2}\det(\mathbb{I}_k+H_S\overline{H_S})$. Applying the bound $\mathcal{C}(Q_S)\le2^{-k}\mathcal{D}(\Sigma)$ of Eq.~\eqref{eq:Slater-bound}, the factor $\det(G_S)^{-2}$ appears on both sides and cancels, leaving
\begin{equation}
    \sum_{|J|=k}\left|\det Z_{S,J}\right|^4\le 2^{-k}\det\!\left(\mathbb{I}_k+H_S\overline{H_S}\right).
\end{equation}
Summing over all $S$ proves Eq.~\eqref{eq:R-definition}.

At the canonical point, it suffices to check equality block by block, since both sides multiply over mutually orthogonal blocks. For one canonical $2\times2$ block, one can directly compute
\begin{equation}
    H_0= \left(
\begin{array}{cc}
 \frac{\lambda ^2+1}{1-\lambda ^2} & -\frac{2 i \lambda }{\lambda ^2-1} \\
 \frac{2 i \lambda }{\lambda ^2-1} & \frac{\lambda ^2+1}{1-\lambda ^2} \\
\end{array}
\right)
\end{equation}
from which follows
\begin{equation}
\begin{split}
    \mathsf{R}(H_0)&= \frac{1}{4} \det(\mathbb{I}_2+ H_0\overline{H_0}) + \frac{1}{2} (1+ H_0|_{11} \overline{H_0} |_{11} ) \\
    &+ \frac{1}{2} (1+ H_0|_{22} \overline{H_0} |_{22} ) + 1= \frac{4 \left(\lambda ^4-\lambda ^2+1\right)}{\left(\lambda ^2-1\right)^2}\,.
\end{split}
\end{equation}
This matches exactly the single-block value of $\mathsf{F}$ in Eq.~\eqref{eq:F_Z_0}.
\end{proof}
\subsubsection{Principal-minor and Cayley transforms}
\label{app:proof:cayley}

We are left with the maximization of the residual $\mathsf{R}(H)$ over the orbit $H=OH_0O^{\trans}$. In this subsection we rewrite it exactly --- no inequality is involved --- as an unweighted sum of squared principal minors of a single Hermitian matrix $\mathcal{A}$, at the price of a prefactor that will turn out to be constant along the orbit (Sec.~\ref{app:proof:graph}). The only tool is, once more, the average over independent random signs, used twice.
For notational convenience, let us fix $m=2\ell$ in the following. 
\begin{lemma}[Orthogonal principal-minor transform]
\label{lem:PMT}
Let $H\in\mathbb{C}^{m\times m}$ satisfy $HH^{\trans}=\mathbb{I}_m$, let $t\ge0$ and let $a,b\ge0$ obey $a+b=1+2t$ and $ab=t$. Then
\begin{equation}
    \sum_{I\subseteq[m]}t^{|I|}\det\!\left(\mathbb{I}+H_IH_I^{\trans}\right)=\sum_{S\subseteq[m]}a^{|S|}b^{m-|S|}\det(H_S)^2\;.
    \label{eq:PMT}
\end{equation}
\end{lemma}

\begin{proof}
Let $D=\mathrm{diag}(\varepsilon_1,\dots,\varepsilon_m)$, where $\varepsilon_i=\pm1$ are independent random signs with uniform distribution, so that the average of any nonempty product of distinct signs vanishes; we write $\varepsilon_S=\prod_{i\in S}\varepsilon_i$.
We recall the principal minor expansion
\begin{equation}
    \det\!\left(\sqrt b\,\mathbb{I}+\sqrt a\,DH\right)= \sum_{S\subseteq[m]}a^{|S|/2}b^{(m-|S|)/2} \varepsilon_S \det(H_S)\,.
\end{equation}
Thus, upon average over the signs, the right-hand side can be written
\begin{equation}
    \sum_{S\subseteq[m]}a^{|S|}b^{m-|S|}\det(H_S)^2=\mathbb{E}_D\det\!\left(\sqrt b\,\mathbb{I}+\sqrt a\,DH\right)^2 .
\end{equation}
Because $HH^{\trans}=\mathbb{I}$ and $D^2=\mathbb{I}$, we have
\begin{equation}
    \begin{split}
      &\left(\sqrt b\,\mathbb{I}+\sqrt a\,DH\right)   \left(\sqrt b\,\mathbb{I}+\sqrt a\,DH\right)^\trans\\
&=
b\,\mathbb{I}
+\sqrt{ab}\left(DH+H^{\trans}D\right)
+a\,DHH^{\trans}D \\
&=
(a+b)\mathbb{I}
+\sqrt{ab}\left(DH+H^{\trans}D\right),
    \end{split}
\end{equation}
so the determinant inside the expectation equals
\begin{equation}
\begin{split}
    &\det\!\left((a+b)\mathbb{I}+\sqrt{ab}\,(H^{\trans}D+DH)\right)\\
    &\qquad\qquad=\det\!\left(t\mathbb{I}+XX^{\trans}\right),\quad X=\mathbb{I}+\sqrt t\,DH\;,
\end{split}
\end{equation}
with $ab=t$ and $a+b=1+2t$.
Using the Cauchy--Binet formula, $\mathbb{E}_D\det(t\mathbb{I}+XX^{\trans})=\sum_{|I|=|J|}t^{m-|I|}\mathbb{E}_D\det(X_{I,J})^2$. Since $X_{ij}=\delta_{ij}+ \sqrt{t}\, \varepsilon_i H_{ij}$, expanding $\det(X_{I,J})$ over the set $Q'\subseteq I\cap J$ of entries supplied by the identity gives
\begin{equation}
    \det(X_{I,J})= \sum_{Q' \subseteq I \cap J} \pm\,t^{|I \setminus Q'|/2} \big(\prod_{i\in I\setminus Q'}\varepsilon_i\big)\det(H_{I\setminus Q',J\setminus Q'})\,,
\end{equation} 
up to $Q'$-dependent signs fixed by the row and column orderings, which are immaterial in what follows. Upon squaring, the sign  averaging kills the cross terms from two distinct sets $Q'_1\neq Q'_2$ --- their sign products run over the nonempty symmetric difference $Q'_1\,\triangle\,Q'_2$ --- while the $Q'$-dependent signs square to one, giving
\begin{equation}
    \mathbb{E}_D\det(X_{I,J})^2=\sum_{Q'\subseteq I\cap J}t^{|I|-|Q'|}\det(H_{I\setminus Q',J\setminus Q'})^2 .
\end{equation} 
Hence, changing variables to $I'=I\setminus Q'$ and $J'=J\setminus Q'$, and setting $R'=[m]\setminus (I' \cup J')$ and $L'=R' \setminus Q'$ (primes denote auxiliary index sets),
\begin{equation}
    \begin{split}
        &\sum_{S\subseteq[m]}a^{|S|}b^{m-|S|}\det(H_S)^2\\
        &=\sum_{|I|=|J|} \sum_{Q'\subseteq I\cap J}t^{m-|Q'|}   \det(H_{I\setminus Q',J\setminus Q'})^2\\
        &= \sum_{\substack{I',J'\subseteq[m]\\ |I'|=|J'|}} \sum_{Q'\subseteq[m]\setminus(I'\cup J')}
t^{m-|Q'|}
\det(H_{I',J'})^2
\\
&=\sum_{\substack{I',J'\subseteq[m]\\ |I'|=|J'|}} t^{|I'\cup J'|}\sum_{Q'\subseteq R'}
t^{|R'|-|Q'|}
\det(H_{I',J'})^2
\\
&=\sum_{\substack{I',J'\subseteq[m]\\ |I'|=|J'|}} t^{|I'\cup J'|}\sum_{L' \subseteq R'}
t^{|L'|}
\det(H_{I',J'})^2
\\
&=\sum_{\substack{I',J'\subseteq[m]\\ |I'|=|J'|}} t^{|I'\cup J'|}\sum_{|L'|=0}^{|R'|} \binom{|R'|}{|L'|}
t^{|L'|}
\det(H_{I',J'})^2
\\
&=\sum_{\substack{I',J'\subseteq[m]\\ |I'|=|J'|}} t^{|I'\cup J'|}(1+t)^{|R'|}
\det(H_{I',J'})^2\\
&=\sum_{\substack{I',J'\subseteq[m]\\ |I'|=|J'|}} t^{|I'\cup J'|}(1+t)^{m-|I' \cup J'|}
\det(H_{I',J'})^2 \,.
    \end{split}
\end{equation} 

For the left-hand side of Eq.~\eqref{eq:PMT}, we arrive at the same result:
\begin{equation}
    \begin{split}
        &\sum_{I\subseteq[m]}t^{|I|}\det\!\left(\mathbb{I}+H_IH_I^{\trans}\right)= \sum_{I\subseteq[m]}t^{|I|}\sum_{\substack{I',J'\subseteq I \\ |I'|=|J'|}} \det(H_{I',J'})^2\\
        &=\sum_{\substack{I',J'\subseteq [m] \\ |I'|=|J'|}} \left( \sum_{\substack{I\subseteq [m] \\ I' \cup J' \subseteq I}} t^{|I|} \right) \det(H_{I',J'})^2\\
        &= \sum_{\substack{I',J'\subseteq [m] \\ |I'|=|J'|}} \left( \sum_{K' \subseteq [m] \setminus (I' \cup J')} t^{|I' \cup J'|+|K'|} \right) \det(H_{I',J'})^2\\
        &= \sum_{\substack{I',J'\subseteq[m]\\ |I'|=|J'|}} t^{|I'\cup J'|}(1+t)^{m-|I' \cup J'|}
\det(H_{I',J'})^2\,.
    \end{split}
\end{equation}
The two chains coincide term by term, proving Eq.~\eqref{eq:PMT}.
\end{proof}

We apply Lemma~\ref{lem:PMT} at $t=1/2$, i.e.\ with
\begin{equation}
    a=1+2^{-1/2}\;,\qquad b=1-2^{-1/2}\;,\qquad ab=\tfrac12\;.
    \label{eq:a-b}
\end{equation}
Since $H$ is Hermitian, $\overline{H_I}=H_I^{\trans}$, and Eq.~\eqref{eq:R-definition} becomes (recall that here $m=2\ell$)
\begin{equation}
    \mathsf{R}(H)=\sum_{S\subseteq[2\ell]}a^{|S|}b^{2\ell-|S|}\det(H_S)^2\;.
    \label{eq:R-weighted-principal}
\end{equation}
These weights are precisely the ones for which a second sign transform produces a simple Cayley numerator.

\begin{proposition}[Cayley reduction]
\label{prop:Cayley}
Write
\begin{equation}
\begin{split}
    H&=(\mathbb{I}_{2\ell}+i\Omega)(\mathbb{I}_{2\ell}-i\Omega)^{-1}\;,\\ \Omega&=O\Big(\bigoplus_{j=1}^{\ell}\lambda_jJ\Big)O^{\trans}\;,
    \label{eq:H-Cayley}
\end{split}
\end{equation}
and set $\Delta_\Omega=\det(\mathbb{I}_{2\ell}-i\Omega)>0$ and $\mathcal{A}=\mathbb{I}_{2\ell}+i\sqrt2\,\Omega$. Then $\mathcal{A}$ is Hermitian and
\begin{equation}
    \mathsf{R}(H)=\Delta_\Omega^{-2}\sum_{S\subseteq[2\ell]}\det(\mathcal{A}_S)^2\;.
    \label{eq:Cayley-reduction}
\end{equation}
\end{proposition}

\begin{proof}
The representation~\eqref{eq:H-Cayley} follows blockwise from $H=OZ_0^2O^{\trans}$ and real orthogonal conjugation: for one block, $(\mathbb{I}_2-i\lambda_jJ)^{-1}=(\mathbb{I}_2+i\lambda_jJ)/\nu_j^2$, so $(\mathbb{I}_2+i\lambda_jJ)(\mathbb{I}_2-i\lambda_jJ)^{-1}=(\mathbb{I}_2+i\lambda_jJ)^2/\nu_j^2=Z_{0,j}^2$, and conjugating by $O$ preserves the Cayley form; since $0\le\lambda_j<1$, $\mathbb{I}_{2\ell}-i\Omega$ is positive Hermitian and $\Delta_\Omega>0$ (the eigenvalues of $i\Omega$ are $\pm\lambda_j$, so those of $\mathbb{I}_{2\ell}-i\Omega$ are $1\mp\lambda_j>0$); moreover $\Omega$ is real antisymmetric, so $i\Omega$, and with it $\mathcal{A}$, is Hermitian. 

Let $D=\operatorname{diag}(\varepsilon_1,\ldots,\varepsilon_{2\ell})$,
where the $\varepsilon_k$ are independent uniform signs. By Eq.~\eqref{eq:R-weighted-principal} and the sign-average identity established at the beginning of the proof of Lemma~\ref{lem:PMT}, we have
\begin{equation}
\mathsf{R}(H)
=
\mathbb{E}_D
\det\!\left(\sqrt{b}\,\mathbb{I}_{2\ell}+\sqrt{a}\,DH\right)^2\,,
\label{eq:R-sign-average}
\end{equation}
and using the Cayley representation of $H$, we have
\begin{equation}
    \begin{split}
\sqrt{b}\,&\mathbb{I}_{2\ell}+\sqrt{a}\,DH
=\\&
\left[
\sqrt{b}\,(\mathbb{I}_{2\ell}-i\Omega)
+\sqrt{a}\,D(\mathbb{I}_{2\ell}+i\Omega)
\right]
(\mathbb{I}_{2\ell}-i\Omega)^{-1}.
\end{split}
\end{equation}
Taking determinants we obtain
\begin{equation}
\mathsf{R}(H)
=
\Delta_\Omega^{-2}
\mathbb{E}_D
\det\!\left(
\sqrt{b}\,(\mathbb{I}_{2\ell}-i\Omega)
+\sqrt{a}\,D(\mathbb{I}_{2\ell}+i\Omega)
\right)^2.
\label{eq:R-Cayley-average}
\end{equation}

We now compare the random matrix inside this expectation with
$\mathbb{I}_{2\ell}+D\mathcal{A}$. We use the elementary fact that, for
a random matrix with independent rows, the expected squared
determinant depends only on the second moments of its individual
rows. Indeed, expanding the determinants with the Leibniz formula, we get
\begin{equation}
\begin{split}
    \mathbb{E}\det(M)^2
&=
\mathbb{E}\!\left[\left(\sum_{\sigma\in\mathfrak{S}_{2\ell}}
\mathrm{sgn}(\sigma)
\prod_{k=1}^{2\ell}
M_{k,\sigma(k)}
\right)^2\right]\;,\\
&=\sum_{\sigma,\tau\in \mathfrak{S}_{2\ell}} \mathrm{sgn}(\sigma) \mathrm{sgn}(\tau)
\prod_{k,j=1}^{2\ell}
\mathbb{E}\!\left[M_{k,\sigma(k)}M_{j,\tau(j)}\right]\\
&=\sum_{\sigma,\tau\in \mathfrak{S}_{2\ell}} \mathrm{sgn}(\sigma) \mathrm{sgn}(\tau)
\prod_{k=1}^{2\ell}
\mathbb{E}\!\left[M_{k,\sigma(k)}M_{k,\tau(k)}\right]
\end{split}
\end{equation}
where $\mathrm{sgn}(\sigma)$ denotes the sign of the permutation $\sigma$ and $\mathfrak{S}_{2\ell}$ the permutation group over $2\ell$ elements and in the last step we used the statistical independence of the rows.

To compare the row second moments, fix $k$, let $u=e_k^{\trans}$
be row $k$ of $\mathbb{I}_{2\ell}$, and let $v=e_k^\trans\cdot(i\Omega)$ be
row $k$ of $i\Omega$. Row $k$ of the random matrix in
Eq.~\eqref{eq:R-Cayley-average} is $\sqrt{b}\,(u-v)+\varepsilon_k\sqrt{a}\,(u+v)$.
Its second moment is therefore
\begin{equation}
    b(u-v)^{\trans}(u-v)+a(u+v)^{\trans}(u+v).
\end{equation}
Using $a+b=2,\, a-b=\sqrt{2}$, we find
\begin{equation}
    \begin{aligned}
&b(u-v)^{\trans}(u-v)+a(u+v)^{\trans}(u+v)\\
&\qquad=
2u^{\trans}u
+\sqrt{2}\left(u^{\trans}v+v^{\trans}u\right)
+2v^{\trans}v\\
&\qquad=
u^{\trans}u+
(u+\sqrt{2}v)^{\trans}(u+\sqrt{2}v).
\end{aligned}
\end{equation}
On the other hand, since
$\mathcal{A}=\mathbb{I}_{2\ell}+i\sqrt{2}\,\Omega$, row $k$ of
$\mathbb{I}_{2\ell}+D\mathcal{A}$ is $u+\varepsilon_k(u+\sqrt{2}v)$.
Its second moment is precisely
\begin{equation}
    u^{\trans}u+(u+\sqrt{2}v)^{\trans}(u+\sqrt{2}v).
\end{equation}
Thus the corresponding rows of the two random matrices have the
same second moments. Their rows are independent because row $k$
depends only on $\varepsilon_k$. It follows that
\begin{equation}
    \begin{aligned}
&\mathbb{E}_D
\det\!\left(
\sqrt{b}\,(\mathbb{I}_{2\ell}-i\Omega)
+\sqrt{a}\,D(\mathbb{I}_{2\ell}+i\Omega)
\right)^2\\
&\qquad=
\mathbb{E}_D\det(\mathbb{I}_{2\ell}+D\mathcal{A})^2.
\end{aligned}
\end{equation}
Finally, the principal-minor expansion combined with the prefactor in Eq.~\eqref{eq:R-Cayley-average} gives Eq.~\eqref{eq:Cayley-reduction}.
\end{proof}

\subsubsection{ Reduction to a unitary-matrix functional }
\label{app:proof:graph}

\DIFdelbegin \DIFdel{\rev{The sum in Eq.~\eqref{eq:Cayley-reduction} can be evaluated as the probability of a simple event for another distribution of the same type as in Lemma~\ref{lem:DPP}. Let}
}

It remains to evaluate the sum $\sum_S\det(\mathcal{A}_S)^2$ in Eq.~\eqref{eq:Cayley-reduction}, which is not obviously maximized at the canonical point. To that end, we interpret it as the probability of a simple event for another distribution of the same type as in Lemma~\ref{lem:DPP}: this probabilistic reading lets us change basis freely within pairs of coordinates, and in a suitable basis the sum collapses to a functional of a single unitary matrix. Let
\begin{equation}\mathcal{M}=\begin{pmatrix}\mathbb{I}_{2\ell}\\ \mathcal{A}\end{pmatrix},\quad \mathcal{G}=\mathcal{M}^\dagger \mathcal{M}=\mathbb{I}_{2\ell}+\mathcal{A}^2,\quad Q=\mathcal{M}\mathcal{G}^{-1/2}.
\end{equation}
The columns of $Q$ are orthonormal, since
\begin{equation}
Q^\dagger Q=\mathcal{G}^{-1/2}\left(\mathbb{I}_{2\ell}+\mathcal{A}^2\right)\mathcal{G}^{-1/2}=\mathbb{I}_{2\ell}\;.
\end{equation}
As in Lemma~\ref{lem:DPP}, the rank-$2\ell$ projector $QQ^\dagger$ defines a probability distribution over the $2\ell$-element subsets $\mathcal{S}$ of the $4\ell$ coordinates, with probabilities $p_{\mathcal{S}}=\det[(QQ^\dagger)_{\mathcal{S}}]=|\det Q_{\mathcal{S}}|^2$, where $Q_{\mathcal{S}}$ is the square $2\ell\times2\ell$ submatrix formed by the rows of $Q$ selected by $\mathcal{S}$. Let $E$ be the event that exactly one coordinate is chosen from every pair $\{i,2\ell+i\}$, $i\in\{1,\dots,2\ell\}$: such a configuration is labeled by the set $S\subseteq[2\ell]$ of pairs where the lower coordinate is chosen, the upper one being chosen on $S^c$. By the same factorization used in the proof of Lemma~\ref{lem:all-minors}, the right factor $\mathcal{G}^{-1/2}$ multiplies all rows at once, so it factors out of the square determinant,
\begin{equation}
\det Q_{\mathcal{S}}=\det(\mathcal{G}^{-1/2})\,\det\widetilde Q_{\mathcal{S}}\;,
\end{equation}
where $\widetilde Q_{\mathcal{S}}$ collects the selected rows of $\mathcal{M}$; expanding $\det\widetilde Q_{\mathcal{S}}$ along the selected rows of the identity block, labeled by $S^c$, removes one by one the columns in $S^c$, and what survives is the principal minor of $\mathcal{A}$ on $S$, up to a sign which is immaterial after squaring. Since the principal minors of the Hermitian matrix $\mathcal{A}$ are real,
\begin{equation}
   p_{\mathcal{S}}=\frac{\det(\mathcal{A}_S)^2}{\det \mathcal{G}}\;,
\end{equation} 
and summing over the configurations in $E$,
\begin{equation}
\mathbb{P}(E)=\frac{\sum_{S\subseteq[2\ell]}\det(\mathcal{A}_S)^2}{\det \mathcal{G}}\;.
\label{eq:graph-event-A}
\end{equation}
We now transform $Q$ without changing $\mathbb{P}(E)$. 
We introduce a unitary rotation acting pairwise on the coordinates $\{i,2\ell+i\}$,
\begin{equation}
\mathcal{W}=\frac{1}{\sqrt2}\begin{pmatrix} \mathbb{I}_{2\ell}&i\,\mathbb{I}_{2\ell}\\ \mathbb{I}_{2\ell}&-i\,\mathbb{I}_{2\ell}\end{pmatrix},
\end{equation}
which maps $Q$ to
\begin{equation}
    \mathcal{W}Q=\frac{1}{\sqrt2}\begin{pmatrix} U_+\\ U_-\end{pmatrix}.
\end{equation}
The probability of $E$ is insensitive to the rotation $\mathcal{W}$. This can be seen as standard compariso between distributions of the type of Lemma~\ref{lem:DPP} and free fermions: for an isometry $Q$, the probabilities $p_{\mathcal{S}}$ are the occupation-number probabilities of the Slater-determinant state whose occupied single-particle space is spanned by the columns of $Q$, and a number-conserving single-particle unitary maps free fermion states to free fermion states, transforming $Q$ accordingly. In this language, $E$ is the event that each pair of modes $\{i,2\ell+i\}$ hosts exactly one particle; the subspace of the Fock space with this occupation pattern is invariant under arbitrary unitaries mixing only the two modes of each pair, and $\mathcal{W}$ is precisely of this form, so $\mathbb{P}(E)$ is unchanged.

By Lemma~\ref{lem:Slater-bound} a further invvariant transformation can be introduced.
Define
\begin{equation}
U_{\pm}=(\mathbb{I}_{2\ell}\pm i\mathcal{A})\,\mathcal{G}^{-1/2}\;.
\end{equation}
Because $\mathcal{A}$ is Hermitian and $[\mathcal{A},\mathcal{G}]=0$, both $U_+$ and $U_-$ are unitary\begin{equation}
\begin{split}
U_\pm U_\pm^\dagger&=(\mathbb{I}_{2\ell}\pm i\mathcal{A})\,\mathcal{G}^{-1}\,(\mathbb{I}_{2\ell}\mp i\mathcal{A})\\
    &=(\mathbb{I}_{2\ell}+\mathcal{A}^2)\,\mathcal{G}^{-1}=\mathbb{I}_{2\ell}\;,
\end{split}
\end{equation}
and $U_\pm^\dagger=(\mathbb{I}_{2\ell}\mp i\mathcal{A})\mathcal{G}^{-1/2}=U_\mp$.
By right multiplication with $U_+^\dagger$, gives
\begin{equation}
    \mathcal{W}Q\,U_+^\dagger=\frac{1}{\sqrt2}\begin{pmatrix} U_+U_+^\dagger\\ U_-U_+^\dagger\end{pmatrix}=\frac{1}{\sqrt2}\begin{pmatrix} \mathbb{I}_{2\ell}\\ U\end{pmatrix},
\end{equation}
where
\begin{equation}
U=U_-U_+^\dagger=U_-^2=(\mathbb{I}_{2\ell}-i\mathcal{A})(\mathbb{I}_{2\ell}+i\mathcal{A})^{-1}\in \mathrm{U}(2\ell)\;.
    \label{eq:U-Cayley-A}
\end{equation}

For the final isometry $\frac{1}{\sqrt2}\big(\begin{smallmatrix}\mathbb{I}_{2\ell}\\ U\end{smallmatrix}\big)$, each of the $2\ell$ selected rows carries the normalization $2^{-1/2}$, and the same expansion along the selected identity rows leaves $\det U_S$ up to a sign, so the configuration with the lower coordinate chosen on $S\subseteq[2\ell]$ and the upper one on $S^c$ has probability
\begin{equation}
    p_S=2^{-2\ell}\,|\det U_S|^2\;.
\end{equation}
Defining
\begin{equation}
    \mathsf{L}(U)=\sum_{S\subseteq[2\ell]}|\det U_S|^2\;,
    \label{eq:L-definition}
    \end{equation}
we conclude that $\mathbb{P}(E)=2^{-2\ell}\,\mathsf{L}(U)$, and comparing with Eqs.~\eqref{eq:Cayley-reduction} and~\eqref{eq:graph-event-A},
\begin{equation}
    \mathsf{R}(H)=\frac{\det(\mathbb{I}_{2\ell}+\mathcal{A}^2)}{2^{2\ell}\,\Delta_\Omega^2}\,\mathsf{L}(U)\;.
    \label{eq:R-unitary}
\end{equation}
Finally, write $\Omega=O\Omega_0O^{\trans}$ with $\Omega_0=\bigoplus_{j=1}^{\ell}\lambda_jJ$ the block form of Eq.~\eqref{eq:H-Cayley}, and let $\mathcal{A}_0=\mathbb{I}_{2\ell}+i\sqrt2\,\Omega_0$ and $U_0=(\mathbb{I}_{2\ell}-i\mathcal{A}_0)(\mathbb{I}_{2\ell}+i\mathcal{A}_0)^{-1}$ be the corresponding block representatives. Since $O$ is real orthogonal,
\begin{equation}
    \mathcal{A}=O\mathcal{A}_0O^{\trans}\;,\qquad U=OU_0O^{\trans}\;,
\end{equation}
so both $\det(\mathbb{I}_{2\ell}+\mathcal{A}^2)$ and $\Delta_\Omega=\det(\mathbb{I}_{2\ell}-i\Omega)$ depend only on the spectrum of $\Omega_0$, i.e.\ on the $\lambda_j$: the prefactor in Eq.~\eqref{eq:R-unitary} is constant along the real orthogonal orbit. The remaining task is to show that $\mathsf{L}(U)$ is maximized at its canonical block representative $U_0$.

\subsubsection{ The rank-two bound and the unitary principal-minor inequality}
\label{app:proof:MD2}

Now we introduce a matrix inequality independent of the fermionic application; we will use it below to bound $\mathsf{L}(U)$ through the diagonal entries of $U$ alone. 
For a polynomial $f(z_1, \cdots,z_n)$ the symbol, $[z_1\cdots z_n]f$ indicates the coefficient of the monomial $z_1\cdots z_n$, e.g. if $f(z_1,z_2)=2z_1+z_2+3 z_1z_3$, then $[z_1 \, z_2]f=3$.
The left-hand side of Eq.~\eqref{eq:MD2} below is known in the linear-algebra literature as a mixed discriminant, cf.\ Ref.~\cite{Bapat1989}.

\begin{theorem}[Rank-two determinant inequality]
\label{thm:MD2}
Let $B_1,\dots,B_n\succeq0$ be Hermitian $n\times n$ matrices with $\mathrm{rank}\,B_i\le2$. Then
\begin{equation}
    [z_1\cdots z_n]\det\!\Big(\sum_{i=1}^nz_iB_i\Big)\le\prod_{i=1}^n\|B_i\|_{\mathrm{HS}}\;,
    \label{eq:MD2}
\end{equation}
with no factorial normalization.
\end{theorem}

\begin{proof}
The strategy is to realize the coefficient in Eq.~\eqref{eq:MD2} as the norm of an operator built by $n$ successive applications of maps $\Phi_{B_j}$, each of which is a norm contraction.

Consider the Fock space of $n$ fermionic modes, with creation operators $c_i^\dagger$ associated with the canonical basis vectors $e_i$ of $\mathbb{C}^n$, and for $x\in\mathbb{C}^n$ let
\begin{equation}
  c_x=\sum_{i=1}^n x_i\,c_i^\dagger
\end{equation} 
be the operator creating one particle in the single-particle state $x$; it maps the $r$-particle sector into the $(r+1)$-particle one. Given a rank-2 matrix $B=\lambda\,uu^\dagger+\mu\,vv^\dagger$ with $u,v \in \mathbb{C}^n$ orthonormal vectors and $\lambda,\mu\ge0$, we define the map
\begin{equation}
    \Phi_B(X)=\lambda\,c_uXc_u^\dagger+\mu\,c_vXc_v^\dagger\;,
\end{equation}
which raises the particle number by one.
We claim that, for the Hilbert--Schmidt norm $||\cdot||_2$, the following holds
\begin{equation}\|\Phi_B(X)\|_{2}\le(\lambda^2+\mu^2)^{1/2}\|X\|_{2}=\|B\|_{\mathrm{HS}}\|X\|_{2}\;,
    \label{eq:Phi-contraction}
\end{equation}
where the last equality follows from
\begin{equation}
    \|B\|_{\mathrm{HS}}^2=\mathrm{Tr}(B^2)=\lambda^2\,\mathrm{Tr}(uu^\dagger)+\mu^2\,\mathrm{Tr}(vv^\dagger)=\lambda^2+\mu^2\;,\end{equation}
the cross terms vanishing by $u^\dagger v=0$.

To prove the claim, we decompose the Fock space into the four occupation sectors $00,10,01,11$ of the two modes $u,v$, and correspondingly the operator $X$ into blocks $X_{s,s'}$ labeled by pairs of sectors. All four sectors are canonically identified with the Fock space of the remaining $n-2$ modes through the maps $c_u,c_v$; the sign ambiguity of this identification is absorbed in the constant $\chi$ below. Creation by $u$ uses the input sectors $00,01$ and sends them to $10,11$, whereas creation by $v$ uses $00,10$ and sends them to $01,11$: the operator $c_uXc_u^\dagger$ has support on the output blocks $(10,10)$, $(10,11)$, $(11,10)$, $(11,11)$, and $c_vXc_v^\dagger$ on $(01,01)$, $(01,11)$, $(11,01)$, $(11,11)$. The two branches are therefore mutually orthogonal in the trace inner product except in the single common output block $(11,11)$, where the inputs $X_{01,01}$ and $X_{10,10}$ meet with a relative fermionic sign $\chi\in\{\pm1\}$ arising from $c_uc_v=-c_vc_u$.
Thus, 
\begin{equation}
\label{eq:norm_squared}
\begin{aligned}
\|\Phi_B(X)\|_2^2
&=
\lambda^2 \|c_u X c_u^\dagger\|_2^2
+\mu^2 \|c_v X c_v^\dagger\|_2^2
\\
&+2\lambda\mu\,
\operatorname{Re}
\left\langle
c_u X c_u^\dagger,\,
c_v X c_v^\dagger
\right\rangle_{\mathrm{HS}}
\\[1ex]
&=
\lambda^2
(
\|X_{00,00}\|_2^2
+\|X_{00,01}\|_2^2
\\
&+\|X_{01,00}\|_2^2
+\|X_{01,01}\|_2^2 )
\\
&\quad
+\mu^2
(
\|X_{00,00}\|_2^2
+\|X_{00,10}\|_2^2
\\
&+\|X_{10,00}\|_2^2
+\|X_{10,10}\|_2^2 )
\\
&\quad
+2\chi\lambda\mu\,
\operatorname{Re}
\left\langle
X_{01,01},\,X_{10,10}
\right\rangle_{\mathrm{HS}} .
\end{aligned}
\end{equation}

Focussing on the block with the mixed product,
\begin{equation}
\begin{split}
    &\lambda^2\|X_{01,01}\|_2^2+\mu^2\|X_{10,10}\|_2^2\\
    &\qquad +2\chi\lambda\mu\,\mathrm{Re}\langle X_{01,01},X_{10,10}\rangle\\
    &\le(\lambda^2+\mu^2)\left(\|X_{01,01}\|_2^2+\|X_{10,10}\|_2^2\right),
\end{split}\end{equation}
which holds because the right-hand side minus the  left-hand side is a complete square,
\begin{equation}
    \left\|\mu X_{01,01}-\chi\lambda X_{10,10}\right\|_2^2\ge0\;.
\end{equation}
Inserting this majorization in Eq.~\eqref{eq:norm_squared} proves Eq.~\eqref{eq:Phi-contraction}.

Let now $C_0=|\mathrm{vac}\rangle\langle\mathrm{vac}|$ be the vacuum projector, with $\|C_0\|_{2}=1$, and set recursively $C_j=\Phi_{B_j}(C_{j-1})$ for $1\le j\le n$. Iterating Eq.~\eqref{eq:Phi-contraction} $n$ times gives
\begin{equation}
    \|C_n\|_{2}\le\prod_{j=1}^n\|B_j\|_{\mathrm{HS}}\;.
    \label{eq:Cn-bound}
\end{equation}
Choosing spectral decompositions $B_j=\sum_{r=1}^2\lambda_{j,r}\,u_{j,r}u_{j,r}^\dagger$, with orthonormal eigenvectors and $\lambda_{j,r}\ge0$ (allowing zero eigenvalues), the recursion unrolls to
\begin{equation}
\begin{split}
    C_n&=\sum_{r_1,\dots,r_n}\Big(\prod_{j=1}^n\lambda_{j,r_j}\Big)\,|\psi_{\vec r}\rangle\langle\psi_{\vec r}|\;,\\
    |\psi_{\vec r}\rangle&=c_{u_{n,r_n}}\cdots c_{u_{1,r_1}}|\mathrm{vac}\rangle\;.
    \end{split}
\end{equation}
The $n$-particle sector of $n$ modes is one-dimensional, spanned by the fully occupied state $|\mathrm{full}\rangle=c_n^\dagger\cdots c_1^\dagger|\mathrm{vac}\rangle$, and applying $n$ creation operators to the vacuum gives
\begin{equation}
    c_{x_n}\cdots c_{x_1}|\mathrm{vac}\rangle=\pm\det(x_1,\dots,x_n)\,|\mathrm{full}\rangle\;,
\end{equation}
where $\det(x_1,\dots,x_n)$ is the determinant of the matrix with columns $x_1,\dots,x_n$. Hence
\begin{equation}
\begin{split}C_n=\Big[\sum_{r_1,\dots,r_n}&\Big(\prod_{j=1}^n\lambda_{j,r_j}\Big)\left|\det(u_{1,r_1},\dots,u_{n,r_n})\right|^2\Big]\\
    &\times|\mathrm{full}\rangle\langle\mathrm{full}|\;.
\end{split}
\label{eq:Cn-scalar}
\end{equation}
It remains to identify the scalar in brackets with the coefficient in Eq.~\eqref{eq:MD2}. 
Let $V$ be the matrix whose columns are $v_{j,r}:=\sqrt{z_j\lambda_{j,r}}\,u_{j,r}.$
then $VV^\dagger=\sum_{j=1}^n z_jB_j$. By the Cauchy--Binet formula,
\begin{equation}
    \det\!\left(\sum_{j=1}^n z_jB_j\right)
=\det(VV^\dagger)
=\sum_{|S|=n} |\det V_S|^2 .
\end{equation}
The coefficient of $z_1\cdots z_n$ receives contributions precisely from
those subsets $S$ containing one column from each pair $(v_{j,1},v_{j,2})$. Thus, writing $X_{\vec r}:=(u_{1,r_1},\ldots,u_{n,r_n}),
\quad r_j\in\{1,2\}$,

\begin{equation}
    |\det V_{\vec r}|^2
=
\left(\prod_{j=1}^n z_j\lambda_{j,r_j}\right)
|\det X_{\vec r}|^2.
\end{equation}
and summing over $\vec{r}$ we have
\begin{equation}
    [z_1\cdots z_n]\det\!\left(\sum_{j=1}^n z_jB_j\right)
=
\sum_{r_1,\ldots,r_n}
\left(\prod_{j=1}^n\lambda_{j,r_j}\right)
|\det X_{\vec r}|^2.
\end{equation}
Since the Hilbert--Schmidt norm of a nonnegative multiple of the rank-one projector $|\mathrm{full}\rangle\langle\mathrm{full}|$ is the multiple itself, combining with Eq.~\eqref{eq:Cn-bound} proves Eq.~\eqref{eq:MD2}. \end{proof}

\begin{remark}
The rank restriction in Theorem~\ref{thm:MD2} is essential: for $B_i=\mathbb{I}_n/n$ the left-hand side of Eq.~\eqref{eq:MD2} is $n!/n^n$ while the right-hand side is $n^{-n/2}$, the former being already larger at $n=3$ therefore does not extend to positive matrices of arbitrary rank, so that Eq.~\eqref{eq:MD2} does not extend to matrices of arbitrary rank.\end{remark}

We now combine Theorem~\ref{thm:MD2} with a standard singular-value argument: the diagonal of $U$ controls $\mathsf{L}(U)$, and the diagonal is in turn controlled by the matrix $T=(U+U^{\trans})/2$, whose singular values are invariant along the orbit $U\mapsto OUO^{\trans}$. This yields a bound on $\mathsf{L}(U)$ by an orbit-invariant quantity.

\begin{theorem}[Unitary principal-minor inequality]
\label{thm:UL}
For $U\in \mathrm{U}(n)$ let $\mathsf{L}(U)=\sum_{S\subseteq[n]}|\det U_S|^2$ and $T=(U+U^{\trans})/2$. Then
\begin{equation}
    \mathsf{L}(U)^2\le 2^n\det\!\left(\mathbb{I}_n+T^\dagger T\right).
    \label{eq:UL}
\end{equation}
\end{theorem}

\begin{proof}
Let $e_i$ be the standard basis, put $v_i=Ue_i$ and define
\begin{equation}
    B_i=\tfrac12\left(e_ie_i^\dagger+v_iv_i^\dagger\right),
\end{equation}
 which are positive semidefinite of rank at most two, chosen so that the mixed coefficient below reproduces exactly $2^{-n}\mathsf{L}(U)$. Indeed, expanding the determinant by multilinearity similarly to the previous proof, the coefficient of $z_1\cdots z_n$ selects one of the two rank-1 terms of $B_i$, the choice being recorded by the set $S$ of indices where $\tfrac12\,v_iv_i^\dagger$ is taken:
\begin{equation}
\begin{split}
    &[z_1\cdots z_n]\det\!\Big(\sum_iz_iB_i\Big)\\ &\qquad =2^{-n}\!\!\sum_{S\subseteq[n]}\left|\det(e_i:i\notin S;\,v_i:i\in S)\right|^2 .
\end{split}
\end{equation}In the determinant labeled by $S$, each column $e_i$ with $i\notin S$ can be expanded away, removing row $i$ together with it; the surviving matrix has entries $(v_i)_j=u_{ji}$ with $j,i\in S$, where $u_{ab}$ denote the entries of $U$, so that, up to a sign,
\begin{equation}
    \det\left(e_i:i\notin S;\,v_i:i\in S\right)=\pm\det U_S\;,
\end{equation}
and the coefficient above equals $2^{-n}\mathsf{L}(U)$. On the other hand,
\begin{equation}
\begin{split}
    \|B_i\|^2_{\mathrm{HS}}&=\mathrm{Tr}(B_i^2)=\tfrac14\left(1+1+2\,|e_i^\dagger v_i|^2\right)\\
    &=\frac{1+|u_{ii}|^2}{2}\;,
\end{split}
\end{equation}
since $e_i^\dagger v_i=u_{ii}$. Theorem~\ref{thm:MD2} therefore gives
\begin{equation}
    \mathsf{L}(U)\le 2^{n/2}\prod_{i=1}^n\left(1+|u_{ii}|^2\right)^{1/2} .
    \label{eq:UL-first}
\end{equation}
Since $t_{ii}=u_{ii}$, it remains to compare the diagonal of $T$ with its singular values $s_1(T)\ge\dots\ge s_n(T)$. For any set $I$ of $k$ indices, define the diagonal rank-$k$ partial isometry
\begin{equation}
    \Theta_I=\sum_{i\in I}e^{i\theta_i}\,e_ie_i^\dagger\;,
\end{equation}
with the phases $\theta_i$ chosen so that $e^{-i\theta_i}t_{ii}=|t_{ii}|$; then
\begin{equation}
    \sum_{i\in I}|t_{ii}|=\left|\mathrm{Tr}(\Theta_I^\dagger T)\right|\le\sum_{j=1}^ks_j(T)\;,
\end{equation} 
the inequality being the standard variational bound for rank-$k$ partial isometries~\cite{bhatia2013matrix}. In other words, every partial sum of the ordered moduli $|t_{ii}|$ is bounded by the corresponding partial sum of singular values. Moreover, $T$ has norm less or equal to one, being the average of the two unitaries $U$ and $U^{\trans}$, and considering the function
\begin{equation}
    \phi(x)=\log(1+x^2)\;,\qquad \phi''(x)=\frac{2(1-x^2)}{(1+x^2)^2}\;,
\end{equation}
which is increasing and convex on $x \in[0,1]$, with $\phi(0)=0$, the following holds: $\sum_i\phi(|t_{ii}|)\le\sum_i\phi(s_i(T))$~\cite{bhatia2013matrix}. Exponentiating this bound gives
\begin{equation}
\begin{split}
\prod_i\left(1+|u_{ii}|^2\right)&=\prod_i\left(1+|t_{ii}|^2\right)\\&\le\prod_i\left(1+s_i(T)^2\right)=\det(\mathbb{I}_n+T^\dagger T).\end{split}
\end{equation} 
Squaring Eq.~\eqref{eq:UL-first} and inserting this bound proves Eq.~\eqref{eq:UL}.
\end{proof}

\subsubsection{Conclusion of the proof}
\label{app:proof:completion}

We return to the unitary $U$ of Eq.~\eqref{eq:U-Cayley-A}. For one canonical block $\Omega_j=\lambda_jJ$, we have $\mathcal{A}_j=\mathbb{I}_2+i\sqrt2\,\lambda_jJ$, hence
\begin{equation}
\begin{split}
    U_j&=(\mathbb{I}_2-i\mathcal{A}_j)(\mathbb{I}_2+i\mathcal{A}_j)^{-1}\\
    &=\big[(1-i)\mathbb{I}_2+\sqrt2\lambda_jJ\big]\,\frac{(1+i)\mathbb{I}_2+\sqrt2\lambda_jJ}{2(i+\lambda_j^2)}\\
    &=\frac{(1-\lambda_j^2)\mathbb{I}_2+\sqrt2\lambda_jJ}{i+\lambda_j^2}=e^{i\varphi_j}\left(\mathsf{c}_j\mathbb{I}_2+\mathsf{s}_jJ\right),
\end{split}
\end{equation} 
where we used $(\alpha\,\mathbb{I}_2+\beta J)^{-1}=\frac{\alpha\,\mathbb{I}_2-\beta J}{\alpha^2+\beta^2}$, and we defined the overall phase $e^{i\varphi_j}=\sqrt{1+\lambda_j^4}/(i+\lambda_j^2)$ and the coefficients
\begin{equation}
    \mathsf{c}_j=\frac{1-\lambda_j^2}{\sqrt{1+\lambda_j^4}}\;,\qquad \mathsf{s}_j=\frac{\sqrt2\lambda_j}{\sqrt{1+\lambda_j^4}}\;,\qquad \mathsf{c}_j^2+\mathsf{s}_j^2=1\;.
\end{equation}
The four principal minors of $U_j$ contribute $1$ (the empty set), $|e^{i\varphi_j}\mathsf{c}_j|^2=\mathsf{c}_j^2$ twice (the two diagonal entries), and $|\det U_j|^2=1$ (unitarity), so that
\begin{equation}
    \mathsf{L}(U_j)=1+2\mathsf{c}_j^2+1=2(1+\mathsf{c}_j^2)\;,
\end{equation}
and since $J^{\trans}=-J$, the transpose flips the sign of $\mathsf{s}_j$, so
\begin{equation}
    T_j=\frac{U_j+U_j^{\trans}}{2}=e^{i\varphi_j}\mathsf{c}_j\,\mathbb{I}_2\;.
\end{equation}
Theorem~\ref{thm:UL} is therefore an equality on each canonical block,
\begin{equation}
    2^2\det(\mathbb{I}_2+T_j^\dagger T_j)=4\,(1+\mathsf{c}_j^2)^2=\mathsf{L}(U_j)^2\;,
\end{equation}
and since both $\mathsf{L}$ and the determinant on the right of Eq.~\eqref{eq:UL} multiply over direct sums, equality also holds for $U_0=\bigoplus_{j=1}^{\ell}U_j$ .

For a general point of the real orbit, $U=OU_0O^{\trans}$ and hence $T^\dagger T=OT_0^\dagger T_0O^{\trans}$, so $\det(\mathbb{I}_{2\ell}+T^\dagger T)$ is orbit invariant. Applying Theorem~\ref{thm:UL} with $n=2\ell$ and using the canonical equality,
\begin{equation}
\begin{split}
    \mathsf{L}(U)^2&\le 2^{2\ell}\det(\mathbb{I}_{2\ell}+T^\dagger T)\\&=2^{2\ell}\det(\mathbb{I}_{2\ell}+T_0^\dagger T_0)=\mathsf{L}(U_0)^2 ,
\end{split}
\end{equation}
so that $\mathsf{L}(U)\le \mathsf{L}(U_0)$, since $\mathsf{L}$ is nonnegative. The prefactor in Eq.~\eqref{eq:R-unitary} being orbit invariant, we get $\mathsf{R}(H)\le \mathsf{R}(H_0)$, and Proposition~\ref{prop:residual}, including its canonical equality, gives the inequality in the common basis
\begin{equation}
    \mathsf{F}(OZ_0O^{\trans})\le \mathsf{R}(H)\le \mathsf{R}(H_0)=\mathsf{F}(Z_0)\;.
\end{equation}
Lemma~\ref{lem:two-to-one} then restores the two independent local bases,
\begin{equation}
\begin{split}
    \mathsf{F}(O_AZ_0O_B^{\trans})&\le\sqrt{\mathsf{F}(O_AZ_0O_A^{\trans})\mathsf{F}(O_BZ_0O_B^{\trans})}\\&\le \mathsf{F}(Z_0)\;,
\end{split}
\end{equation}
which is Eq.~\eqref{eq:Cartan-target}. Substituting into the exact physical normalization~\eqref{eq:physical-F} proves
\begin{equation}
    \zeta_2\!\left(O_AZ_0O_B^{\trans}\right)\le\prod_{j=1}^{\ell}\left(1-\lambda_j^2+\lambda_j^4\right)
    \label{eq:balanced-physical-bound}
\end{equation}
whenever $\lambda_j<1$, with equality in canonical axes by Eq.~\eqref{eq:canonical-value}.

The BCS matrix $Z_0$ and the unnormalized functional $\mathsf{F}$ are singular when some $\lambda_j=1$, but this is only a singularity of the parametrization: for fixed $O_A,O_B$, consider the rescaled spectra $\lambda_j^{(r)}=r\lambda_j$ and $\nu_j^{(r)}=\sqrt{1-r^2\lambda_j^2}$ with $0<r<1$, for which every $\lambda_j^{(r)}<1$ and Eq.~\eqref{eq:balanced-physical-bound} applies. The functional $\zeta_2$ of Eq.~\eqref{eq:zeta-M} is a polynomial in the covariance entries, hence continuous, so both sides of Eq.~\eqref{eq:balanced-physical-bound} are continuous in $r$; taking $r\to1^-$ proves the balanced result at every endpoint as well.

It remains to remove the balanced-cut assumption. The strategy is to append trivial pure modes to the smaller side, so as to reuse the balanced result, and then to check that the addition changes neither side of the bound. Indeed, tensoring a state with a pure vacuum mode does not change $\zeta_2$. Removing the untouched ancillas therefore gives, with $\Gamma$ the canonical covariance of Proposition~\ref{prop:covariance-normal-form},
\begin{equation}
    \zeta_2\!\left[(O_A\oplus O_B)\Gamma(O_A\oplus O_B)^{\trans}\right]\le\prod_{j=1}^{\ell}\left(1-\lambda_j^2+\lambda_j^4\right)
\end{equation}
for every unbalanced cut, with equality attained in the canonical Botero--BCS axes. Together with $M_2=-\log_2\zeta_2$ this proves Theorem~\ref{thm:global}. \hfill$\blacksquare$

\subsection{Verification via variational optimization}
\label{app:variational}

We discuss here the numerical methods behind the results in Fig.~\ref{fig:Ising_model}(a), and substantiate our claim that fermionic non-local magic as defined in Eq.~\eqref{eq:defFNLmagic} coincides with magic of the canonical modewise form of the Gaussian state in Eq.~\eqref{eq:can_form}, for every R\'enyi index $\RENYI\ge2$; for $\RENYI=2$ this is a direct benchmark of the content of Theorem~\ref{thm:global}, proved in Appendix~\ref{app:proof}.

We adopt a standard variational optimization approach: given a Gaussian state $|\Psi\rangle$ on $N$ sites ($N$ even for simplicity), we define the cost function
\begin{equation}
    \mathcal{C}(U_A, U_B) := M_2\left[U_A \otimes U_B |\Psi \rangle \right]
    \label{eq:cost_function}
\end{equation}
with $U_A, U_B$ unitary operators acting on regions $A=\{1, \ldots, {N/2}\}$ and $B=\{{N/2+1}, \ldots, N\}$ respectively. The state is represented by its $2^N$ complex coefficients on the computational basis, and the unitaries as $d_I \times d_I$ complex matrices, with $d_I=2^{|I|}, \ I \in\{A,B\}$. The unitaries are then parametrized by a set of real numbers (or angles) $\vec{\theta}$, depending on their constraints. For generic unitaries, we use the exponential map $U(\vec{\theta})=\exp(-i\vec{\theta} \cdot \vec{\tau})$, with $\vec{\tau}$ the generators of the $\mathfrak{u}(d_{I})$ algebra. For fermionic Gaussian unitaries, we write them as $U(\vec{\theta})=\exp(H(\vec{\theta}))$, with $H$ quadratic form of the Majorana operators as in Eq.~\eqref{eq:ham}, where the entries of the antisymmetric matrix $h$ are given by the angles $\vec{\theta}$. With this parametrization, the cost function becomes a function of the angles $\mathcal{C}(\vec{\theta})$, and its gradient at a given point $\vec{\theta}$ in an arbitrary direction $\vec{e_j}=(0,\ldots,1_j,\ldots,0)$ can be approximated via the finite-difference method:
\begin{equation}
    \partial_{\theta_j}\mathcal{C}(\vec{\theta}) \approx \frac{\mathcal{C}(\vec{\theta}+\epsilon\vec{e_j}) - \mathcal{C}(\vec{\theta})}{ \epsilon},
\end{equation} 
where we set $\epsilon$ to be of order $10^{-8}$.

The cost function and its gradient are then fed into a gradient-based optimization algorithm such as L-BFGS, which updates the parameters until a minimum of the cost function is reached. To help with convergence to a global minimum, the whole optimization routine is repeated multiple times, with each component of the initial angle vector $\vec{\theta}_0$  chosen uniformly at random in $[0,2\pi)$.

For the results shown in Fig.~\ref{fig:Ising_model}(a), we minimize the cost function above over generic random unitaries, choosing for $|\Psi(\mu)\rangle$ the ground state of the Ising Hamiltonian obtained by setting $\eta=1$ in Eq.~\eqref{eq:xy}. The state is obtained via exact diagonalization (ED) on $N=8$ sites, and the optimization procedure is repeated for multiple values of $\mu$.
\begin{figure}
    \centering
    \includegraphics[width=0.9\linewidth]{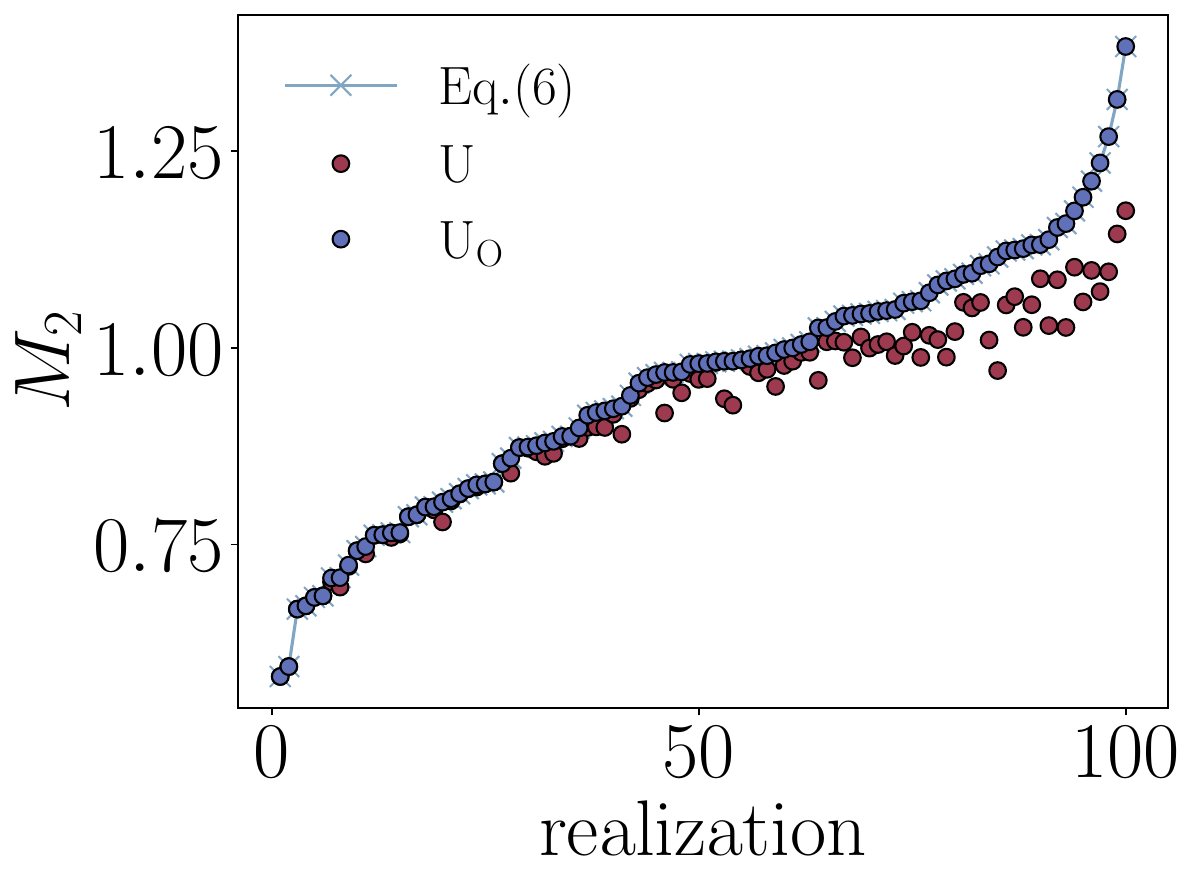}
    \caption{Results of the comparison between non-local magic computed with the formula in Eq.~\eqref{eq:analytic_NL} (crosses) and the one obtained by minimizing the cost function in Eq.~\eqref{eq:cost_function} with both generic unitaries $U$ (red marker) and fermionic Gaussian unitaries $U_O$ (blue marker), for $100$ different realizations of random fermionic Gaussian states on $N=10$ sites, $\ell=5$. The realizations are ordered by increasing $M_2^{\mathrm{FNL}}$. The formula agrees with minimization over Gaussian unitaries to machine precision in all cases.}
    \label{fig:numerics_verification}
\end{figure}
Instead, in Fig.~\ref{fig:numerics_verification} we show the results of the minimization for $100$ realizations of random Gaussian states on $N=10$ qubits, $\ell=5$, where the minimum is taken both on generic and on fermionic Gaussian unitaries. We also plot the values of fermionic non-local magic obtained with our formula for these states, which agree within numerical precision with the results obtained by minimizing over Gaussian unitaries.

A few optimizations were considered in the steps to compute the cost function. When applying the unitaries on the state, it is convenient to reshape the state vector into a $d_A \times d_B$ matrix $\Psi$, such that the application of the unitary $U_A \otimes U_B$ can be computed as $U_A \Psi U_B^{T}$, with reduced cost $\mathcal{O}(d_A^2 d_B + d_A d_B^2)$.
Moreover, the stabilizer entropy of the state vector can be computed using the algorithm introduced in Refs.~\cite{huang2026fastexactapproachstabilizer, sierant2026computingquantummagicstate}, which implements fast Walsh-Hadamard transforms to reduce the time complexity of the computation from $\mathcal{O}(8^N)$ to $\mathcal{O}(N4^N)$, and can also be parallelized for substantial speedup. Finally, the finite-difference method can be replaced with backpropagation, which allows the full $N_p$-dimensional gradient to be computed at a cost comparable to a small constant number of evaluations of the cost function, rather than $O(N_p)$. Note that the number of parameters $N_p$ grows exponentially in $N$ in the case of generic unitaries, and quadratically in $N$ for fermionic Gaussian ones.

\subsection{Error analysis for experimental accessibility}
\label{app:experimental}
We provide the error propagation analysis connecting the accuracy on $\Gamma_A$ to the accuracy on $M_2^{\mathrm{FNL}}$.
Since $\mathfrak{m}_2(\lambda^2) = -\log_2(1-\lambda^2+\lambda^4)$ is smooth on $[0,1]$, its derivative
\begin{equation}
    \frac{d}{d\lambda}\,\mathfrak{m}_2(\lambda^2) = \frac{2\lambda(1-2\lambda^2)}{(1-\lambda^2+\lambda^4)\ln 2}
\end{equation}
is bounded, so $\mathfrak{m}_2$ is uniformly Lipschitz on $[0,1]$ with constant $L = \max_{\lambda \in [0,1]} \bigl|\frac{d}{d\lambda}\mathfrak{m}_2(\lambda^2)\bigr| = O(1)$.
Combining with the Hoffman--Wielandt stability bound for eigenvalues of $i\Gamma_A$,
\begin{equation}
    |\delta M_2^{\mathrm{FNL}}| \lesssim L\,\sqrt{\ell}\,\|\delta\Gamma_A\|_F\;.
\end{equation}
A uniform entrywise error $\epsilon_\Gamma$ implies $\|\delta\Gamma_A\|_F = O(\ell \epsilon_\Gamma)$, giving $|\delta M_2^{\mathrm{FNL}}| = \mathcal{O}(\ell^{3/2}\epsilon_\Gamma)$.
To estimate the non-local magic density $m_2^{\mathrm{FNL}} = M_2^{\mathrm{FNL}}/\ell$ with additive accuracy $\varepsilon$, one needs $\epsilon_\Gamma = \mathcal{O}(\varepsilon/\sqrt{\ell})$, yielding
\begin{equation}
    N_{\mathrm{shots}} = O\!\left(\frac{\ell^3\log(\ell/\delta)}{\varepsilon^2}\right).
\end{equation}
This estimate is pessimistic for gapped phases, where $\Gamma_A$ is approximately banded and only $\mathcal{O}(1)$ entanglement modes near the boundary contribute appreciably; near criticality the bound is closer to the true cost.

\subsection{$\alpha$-Generalization of Results}
\label{Sec:generaliz}

In the Main Text, several results were discussed exclusively for the Rényi index $\alpha = 2$. For completeness, we now present their extensions to arbitrary integer $\alpha \ge 2$. These extensions are obtained from the conjectured weights of Eq.~\eqref{eq:analytic_NL_alpha}; at $\alpha=2$ they reduce to the results of the Main Text, which follow from Theorem~\ref{thm:global}.

We begin with the Page curve. From Eq.~\eqref{eq:ziopin2}, replacing $\mathfrak{m}_2$ with $\mathfrak{m}_\alpha$, we find that the behavior for $r \simeq 0$ remains the same for every $\alpha$ up to a constant:
\begin{equation}
\overline{\mathsf{m}}_\alpha(r \simeq 0) = - \frac{\alpha}{1-\alpha} \frac{r^2}{2 \ln 2} + O(r^3).
\end{equation}

Next, we consider the critical behavior in the Ising model. From Eq.~\eqref{eq:alpha_NL}, substituting $\mathfrak{m}_2 \mapsto \mathfrak{m}_\alpha$, a change of variables $\lambda = \frac{1-t}{1+t}$ leads to
\begin{equation}
\beta_{\mathrm{FNL}}(\alpha) = \frac{1}{\pi^2 (1-\alpha) \ln 2} \int_0^1 \frac{\hat{\mathfrak{m}}_\alpha(t)}{t}\, dt,
\end{equation}
where
\begin{equation}
\hat{\mathfrak{m}}_\alpha(t) = \ln \left[ \frac{(1+t)^{2\alpha} + (1-t)^{2\alpha} + 4^\alpha t^\alpha}{2 (1+t)^{2\alpha}} \right].
\end{equation}
Noting that $\hat{\mathfrak{m}}_\alpha(t) = \hat{\mathfrak{m}}_\alpha(1/t)$, the integral can equivalently be written as
\begin{equation}
\beta_{\mathrm{FNL}}(\alpha) = \frac{1}{2 \pi^2 (1-\alpha) \ln 2} \int_0^\infty \frac{\hat{\mathfrak{m}}_\alpha(t)}{t}\, dt,
\end{equation} 
which corresponds to the Mellin transform of $\hat{\mathfrak{m}}_\alpha(t)$ evaluated as $s \to 0$.
The result of the integral is given by 
\begin{equation}
\beta_{\mathrm{FNL}}(\alpha) = \frac{1}{2 (1-\alpha) \pi^2 \ln(2)}\sum_{j=1}^\alpha \mathrm{arccosh}^2(-r_j/2)\;,
\end{equation}
where $r_j$ are the (generally complex) roots of the polynomial $R_\alpha(x)=(x-2)^\alpha+(x+2)^\alpha+4^\alpha$.

\bibliography{ref_current}
\bibliographystyle{apsrev4-2}

\end{document}